\documentclass[a4paper,cleveref,thm-restate]{lipics-v2021}

\hideLIPIcs
\nolinenumbers

\usepackage{preamble-fv-simple}
\usepackage{macros}
\usepackage{todonotes}
\usepackage{multicol}

\title{Tight Bounds for some W[1]-hard Problems Parameterized by Multi-clique-width}
\date{}


\authorrunning{B.\ Bergougnoux, V.\ Chekan, S.\ Kratsch}

\author{Benjamin Bergougnoux}{Aix Marseille Université, CNRS, LIS, Marseille, France}{benjamin.bergougnoux@lis-lab.fr}{https://orcid.org/0000-0002-6270-3663}{}

\author{Vera Chekan}{Humboldt-Universität zu Berlin, Germany}{vera.chekan@hu-berlin.de}{https://orcid.org/0000-0002-6165-1566}{}

\author{Stefan Kratsch}{Humboldt-Universität zu Berlin, Germany}{stefan.kratsch@hu-berlin.de}{https://orcid.org/0000-0002-0193-7239}{}

\begin{document}

\maketitle

\begin{abstract}
    In this work we contribute to the study of the fine-grained complexity of problems parameterized by multi-clique-width, which was initiated by Fürer~[ITCS 2017] and pursued further by Chekan and Kratsch [MFCS 2023].
    Multi-clique-width is a parameter defined analogously to clique-width but every vertex is allowed to hold multiple labels simultaneously.
    This parameter is upper-bounded by both clique-width and treewidth (plus a constant), hence it generalizes both of them without an exponential blow-up.
    Conversely, graphs of multi-clique-width $k$ have clique-width at most $2^k$, and there exist graphs with clique-width at least $2^{\Omega(k)}$.
    Thus, while the two parameters are functionally equivalent, the fine-grained complexity of problems may differ relative to them.
    
    As our first and main result we show that under ETH the \textsc{Max Cut} problem cannot be solved in time $n^{2^{o(k)}} \cdot f(k)$ on graphs of multi-clique-width $k$ for any computable function $f$.
    For clique-width~$k$ an $n^{\cO(k)}$ algorithm by Fomin et al.~[SIAM J.\ Comput.\ 2014] is tight under ETH.
    This makes \textsc{Max Cut} the first known problem for which the tight running times differ for parameterization by clique-width and multi-clique-width and it contributes to the short list of known lower bounds of form $n^{2^{o(k)}} \cdot f(k)$.
    As our second contribution we show that \textsc{Hamiltonian Cycle} and \textsc{Edge Dominating Set} can be solved in time $n^{\cO(k)}$ on graphs of multi-clique-width $k$ matching the tight running time for clique-width.
    These results answer three questions left open by Chekan and Kratsch~[MFCS 2023].
\end{abstract}

\newpage
\section{Introduction}

The framework of parameterized complexity allows a more detailed study of computational complexity by expressing problem complexity as a function of both input size $n$ and one or more parameters. A problem with parameter $k$ is fixed-parameter tractable (FPT) if it can be solved in time $f(k)\cdot n^{\cO(1)}$, and slice-wise polynomial (XP) if it can be solved in time $n^{g(k)}$. While these definitions allow any dependence on $k$, it is natural to seek to make this as low as possible. Moreover, there has been much research activity on fine-grained (parameterized) complexity, that is to find FPT and XP algorithms with optimal dependence on $k$ subject to hypotheses such as the (Strong) Exponential-Time Hypothesis (ETH and SETH).\footnote{ETH is the hypothesis that there is $c>0$ such that \textsc{3-SAT} cannot be solved in time $\cO(2^{cn})$, while SETH states that for each $\varepsilon>0$ there is $q\in\mathbb{N}$ such that \textsc{$q$-SAT} cannot be solved in time $\cO((2-\varepsilon)^n)$.}

This has been most successful relative to structural parameters of graphs such as treewidth ($\tw$) and clique-width.
For example, many standard problems (e.g., \textsc{Vertex Cover}, \textsc{Dominating Set}, \textsc{3-Coloring}) cannot be solved in subexponential FPT time, i.e., in~$2^{o(\tw)} n^{\cO(1)}$, under ETH \cite{0086373,LokshtanovMS11}. 
On the other hand, for these problems standard dynamic-programming algorithms on tree decompositions achieve those bounds implying their tightness.
Similar single-exponential tightness results for treewidth were shown for several problems, e.g., \textsc{Feedback Vertex Set}, \textsc{Hamiltonian Cycle}, \textsc{Treewidth}~\cite{BodlaenderCKN15,CyganNPPRW22,Bonnet25}.
The results of this flavor are not restricted to single-exponential running times.
It was shown that for some problems (e.g.,  \textsc{Planar Vertex Cover}, \textsc{Planar Hamiltonian Cycle})~$2^{\Theta(\sqrt{\tw})} n^{\cO(1)}$ is optimal \cite{LokshtanovMS11}, for some other problems (e.g., \textsc{Vertex Disjoint Paths}, \textsc{Chromatic Number}) $2^{\Theta(\tw \log \tw)} n^{\cO(1)}$ is optimal~\cite{LokshtanovMS18slighly,CyganNPPRW22,HanakaL22,EsmerM25,LampisMV25}, and sometimes (e.g., for \textsc{Metric Dimension}, \textsc{Geodetic Set}) even double-exponential running time $2^{2^{\Theta(\tw)}} n^{\cO(1)}$ is required~\cite{MarxM16,LampisM17,BliznetsH24,FoucaudGK0IST24,HanakaKL26}.
ETH was also used to prove some tight bounds for XP algorithms for \W[1]-hard problems. 
For example, \textsc{List Coloring} and \textsc{CSP} (with the treewidth of the primal graph) cannot be solved in time $n^{o(\tw)} f(\tw)$ for any computable function~$f$~\cite{FellowsFLRSST11}.

For problems solvable in single-exponential time, even finer bounds can be achieved under the stronger hypothesis, namely SETH.
This line of research was initiated by Lokshtanov et al.~\cite{LokshtanovMS18} who showed that for several standard problems (e.g., \textsc{Independent Set}, \textsc{Dominating Set}, \textsc{Odd Cycle Transversal}), folklore dynamic-programming algorithms on tree decompositions are optimal under SETH.
For instance, \textsc{Independent Set} cannot be solved in time $(2-\varepsilon)^{\tw} n^{\cO(1)}$ for any $\varepsilon > 0$ unless SETH fails. 
Since then, many results of this flavor were shown for treewidth (see e.g., \cite{RooijBR09,CyganKN18,MarxSS21,CyganNPPRW22,DubloisLP22,EsmerFMR24,FockeMR24,LampisV24,EsmerM25,toct/FockeMINSSW25,FockeMINSSW25,abs-2506-01645,GreilhuberSW25,LampisV25}).

One of the limitations of treewidth is that only sparse graphs have bounded treewidth,
more precisely, a graph of treewidth $k$ with $n$ vertices contains at most $kn$ edges.
To overcome this limitation, a variety of parameters generalizing treewidth was studied, most importantly clique-width ($\cw$).
The clique-width of a graph is at most~$k$ if it 
can be constructed using at most $k$ labels where the operations, informally speaking, are such that vertices assigned the same label may only get the same neighbors in the future, and relabeling is possible to allow reusing the labels.
For clique-width, many (S)ETH-tight results are also known (e.g.,~\cite{KoblerR03,BodlaenderLRV10,FominGLS14,CurticapeanM16,FominGLSZ19,BergougnouxKK20,Lampis20,HegerfeldK23,BojikianK24,GanianHKOS24,BojikianK25,LampisV25}).
Even though bounded clique-width generalizes bounded treewidth, an exponential blow-up is necessary sometimes: there exist graphs of treewidth $k$ and clique-width $2^{\Theta(k)}$~\cite{CorneilR05}. Thus, being FPT or XP relative to clique-width transfers to parameter treewidth, but the same is not true for fine-grained bounds.

With this motivation, Martin Fürer introduced the two parameters fusion-width ($\fw$)~\cite{Furer14} and multi-clique-width ($\mcw$)~\cite{Furer17} both generalizing treewidth and clique-width but with no such exponential blow-up, i.e., the following relations apply:
\begin{align*}
    \fw \leq \cw && \mcw \leq \cw && \fw \leq \tw + 2 && \mcw \leq \tw + 2.
\end{align*}
Fusion-width generalizes clique-width by allowing an additional operation which merges all vertices of a certain label,
while multi-clique-width essentially uses the same operations as clique-width but it allows every vertex to hold multiple labels simultaneously.
Even though the three parameters are functionally equivalent~\cite{Furer14,Furer17}, in the context of fine-grained complexity, the precise choice of the parameter matters.
Fürer~\cite{Furer14,Furer17} initiated the study of the complexity of problems parameterized by these two parameters by considering variants of the \textsc{Independent Set} problem.
Chekan and Kratsch~\cite{ChekanK23} continued this line of research by first, establishing the relation $\mcw \leq \fw + 1$ between the two parameters.
Second, they showed that several problems (e.g., $q$-\textsc{Coloring}, \textsc{Chromatic Number}, \textsc{Connected Vertex Cover}) can be solved in the same running time when parameterized by clique-width and its generalization multi-clique-width.
The existing matching lower bounds for clique-width then imply the tightness of their algorithms under (S)ETH.
They left open though whether \textsc{Hamiltonian Cycle}, \textsc{Max Cut}, and \textsc{Edge Dominating Set} can also be solved in the same running time as for clique-width. 
For those problems, they only managed to achieve the same running time for the weaker parameter fusion-width.

In this work, we answer these three questions.
First, we show that \textsc{Max Cut}, known to be solvable in $n^{\cO(k)}$ for clique-width, cannot be solved in time $n^{2^{o(k)}} \cdot f(k)$ for any computable function $f$ for multi-clique-width unless ETH fails.
An algorithm with matching running time can be derived from the $n^{\cO(k)}$ algorithm for clique-width by Fomin et al.~\cite{FominGLS14} using the fact that a multi-$k$-expression can trivially be transformed into a $2^k$-expression~\cite{Furer17}.
This makes \textsc{Max Cut} the first and only problem known to separate clique-width and multi-clique-width in terms of fine-grained parameterized complexity of problems.
Also it contributes to the short
list of ETH-based lower bounds of form $n^{2^{o(k)}}$: to the best of our knowledge, for structural graph parameters such a lower bound is only known for the \textsc{Chromatic Number},~\textsc{$b$-Coloring}, and \textsc{Fall Coloring} problems parameterized by clique-width (and therefore, of course, also for smaller parameters)~\cite{FominGLSZ19,JaffkeLL24}.
Our construction relies on cuts forming anti-matchings.
Crucially, we find out that ``paths of anti-matchings'' have logarithmic multi-clique-width leading to the desired lower bound.

\begin{restatable}{theorem}{thmMaxCut}
\label{thm:maxcut-mcw-lb}
    Unless ETH fails, there is no algorithm solving \textsc{Max Cut} in time $n^{2^{o(k)}} \cdot f(k)$ when the input graph $G$ on $n$ vertices is provided with a multi-$k$-expression of $G$ for any computable function $f$.
\end{restatable}

Second, we show that \textsc{Hamiltonian Cycle} and \textsc{Edge Dominating Set} can both be solved in time $n^{\cO(k)}$ for multi-clique-width, i.e., in the same running time as for clique-width, matching the lower bounds by Fomin et al.~\cite{FominGLS14,FominGLSZ19}.

\begin{restatable}{theorem}{thmHC}
\label{thm:hc-algorithm}
    Given a graph $G$ on $n$ vertices and a multi-$k$-expression of $G$, in time $n^{\cO(k)}$ we can decide whether $G$ admits a Hamiltonian cycle.
\end{restatable}

\begin{restatable}{theorem}{thmEDS}
\label{thm:eds-algorithm}
    Given a graph $G$ on $n$ vertices, a multi-$k$-expression of $G$, and integer $\cardinality$, in time $n^{\cO(k)}$ we can decide whether $G$ admits an edge dominating set of size at most $\cardinality$.
\end{restatable}

To obtain those algorithms, we rely on the fact that for every vertex, only a constant number of labels can participate in the creation of incident edges used by a solution.
For Hamiltonian cycles, we show that this idea can be combined with
the machinery of representative sets, which Bergougnoux et al.~\cite{BergougnouxKK20} used to obtain an $n^{\cO(k)}$ algorithm for clique-width.
As the main technical contribution behind this algorithm, we show that their representation still applies to the more general operations used for multi-clique-width.
To obtain an algorithm for edge dominating sets, we use a known reformulation of the problem in terms of vertex covers and matchings.
This property is in contrast with the \textsc{Max Cut} problem for which every edge may, in general, be used by a solution.
Hence, one intuitively needs to keep track of all labels of a vertex motivating the lower bound in \cref{thm:maxcut-mcw-lb}.

\subparagraph*{Related work} (S)ETH-tight bounds have also been actively studied for various further structural parameters. 
For example, going below treewidth, size of the modulator to constant treewidth and related parameters were considered~\cite{HegerfeldK22,EsmerFMR24-source}.
For cutwidth many standard problems~\cite{JansenN19,GeffenJKM20,GroenlandMNS22,BojikianCHK23}, connectivity problems~\cite{BojikianCK25}, and several more~\cite{MarxSS21,PiecykR21,GroenlandMNPR24} were studied.
Parameterizations by vertex-cover number~\cite{LampisM17,DumontLL0T25,FeldmannL25,LampisMV25,LampisV25-node}, treedepth~\cite{LampisV24,HanakaLMN0V25-steiner,LampisMV25}, vertex integrity~\cite{LampisM24}, and feedback-vertex-set number~\cite{PiecykR21} were also studied.
As for parameters for dense graphs, some results are known for rank-width~\cite{BergougnouxKN23} and modular-treewidth~\cite{HegerfeldK23-mtw}.
Recently, Lampis~\cite{Lampis26a,Lampis26,Lampis25-pw-seth} also initiated a line of research devoted to conjectures weaker than SETH from which the same tight bounds follow, for instance, pw-SETH.

\subparagraph*{Organization of the paper} First, we set up the notation and definitions in \cref{sec:preliminaries}.
In \cref{sec:max-cut} we present the lower-bound construction for \textsc{Max Cut}.
In \cref{sec:hamiltonian-cycle,sec:eds} we provide the algorithms for \textsc{Hamiltonian Cycle} and \textsc{Edge Dominating Set}, respectively.

\section{Preliminaries}\label{sec:preliminaries}
For an integer $k$, we define the sets $[k] = \{i \in \bN \mid 1 \leq i \leq k\}$ (in particular, we have $[0] = \emptyset)$ and $[k]_0 = [k] \cup \{0\}$.
For a set $S$ and an integer $0 \leq i \leq |S|$, we define the family ${S \choose i} = \{T \subseteq S \mid |T| = i\}$.
For a function $f \colon A \to B$, integer $t$, elements $b_1, \dots, b_t \in B$, and pairwise distinct elements $a_1, \dots, a_t \in A$ we define the mapping $f[a_1 \mapsto b_1, \dots, a_t \mapsto b_t] \colon A \to B$ as
\[
    f[a_1 \mapsto b_1, \dots, a_t \mapsto b_t](x) =
    \begin{cases}
        b_i & \text{if } x = a_i \text{ for some } i \in [t] \\
        f(x) & \text{otherwise}.
    \end{cases}
\]
For a subset $A \subseteq C$, the function $f_{|_{C}} \colon C \to B$ is the restriction of $f$ to the domain $C$. 
Further, for disjoint sets $A$ and $C$ and functions $f_1 \colon A \to B$ and $f_2 \colon C \to D$, we define the function $f_1 \cup f_2 \colon A \cup C \to B \cup D$ via 
\[
    (f_1 \cup f_2)(x) =
    \begin{cases}
        f_1(x) & \text{if } x \in A \\
        f_2(x) & \text{if } x \in C.
    \end{cases}
\]

Graphs in this paper are simple and undirected.
For a graph $G$ and a set $Q \subseteq E(G)$ of edges, the set $V(Q)$ denotes the set of endpoints of edges in $Q$.
The subgraph $G[Q]$ \emph{induced} by $Q$ is defined as $G[Q] = \big(V(Q), Q\big)$.
The set $Q$ is a \emph{matching} in~$G$ if no two distinct edges of $Q$ share an endpoint.
For a set $X \subseteq V(G)$ of vertices in $G$, the subgraph $G[X]$ \emph{induced} by~$X$ is $G[X] = \big(X, \{uv \in E(G) \mid u,v \in X\}\big)$.
The set $X$ is a \emph{vertex cover} of $G$ if the subgraph $G\big[V(G) \setminus X\big]$ is edgeless. 

Unless specified otherwise, a \emph{partition} of a graph $G$ is a partition of its vertex set. 
We only consider partitions into two parts, called \emph{sides}, and we allow either side of the partition to be empty.
Further for two sets $A, B$ (not necessarily subsets of $V(G)$), we define the set of edges $E_G(A, B)$ of edges \emph{crossed} by $(A, B)$ as $E_G(A, B) = \big\{ uv \in E(G) \mid u \in A, v \in B\big\}$.

For an integer $k$, a \emph{multi-$k$-labeled} graph is a pair $(H, \lab)$ where $H$ is a graph and~$\lab \colon V(H) \to 2^{[k]}$ is called a \emph{multi-labeling function}.
For a vertex $v \in V(H)$, the set $\lab(v)$ is called the \emph{label set} of $v$ (in $\lab$).
We will consider the following four operations on multi-$k$-labeled graphs.
\begin{enumerate}
    \item \emph{Introduce}: For labels $q \in [k]$, $i_1, \dots, i_q \in [k]$, and an element $v$, the operator $v \langle i_1, \dots, i_q \rangle$ creates a multi-$k$-labeled graph with a single vertex $v$ that has label set $\{i_1, \dots, i_q\}$.
    \item \emph{Union}: The operator $\oplus$ takes two vertex-disjoint multi-$k$-labeled graphs and creates their (disjoint) union. The label sets are preserved.
    \item \emph{Join}: For labels $i \neq j \in [k]$, the operator $\eta_{i, j}$ takes a multi-$k$-labeled graph $(H, \lab)$ satisfying $\{i,j\} \not\subseteq \lab(v)$ for all $v \in V(H)$, i.e., the label set of no vertex contains both labels $i$ and $j$ simultaneously.
    And it creates the supergraph $(H', \lab)$ of $H$ on the same vertex set, i.e., $V(H') = V(H)$, such that $E(H') = E(H) \cup \big\{uv \mid i \in \lab(u), j \in \lab(v)\big\}$.
    The label sets are preserved.
    \item \emph{Relabel}: For labels $i \in [k]$ and $S \subseteq [k]$, the operator $\rho_{i \to S}$ takes a multi-$k$-labeled graph~$(H, \lab)$ and creates the multi-$k$-labeled graph $(H', \lab')$ satisfying $H' = H$ and 
    \[
        \lab'(v) = 
        \begin{cases}
            \big(\lab(v) \setminus \{i\}\big) \cup S & \text{if } i \in \lab(v) \\
            \lab(v) & \text{otherwise}
        \end{cases}
    \]
    for every vertex $v \in V(H)$.
    Informally speaking, the operation replaces every occurrence of the label $i$ by the set $S$.
    Note that it is in particular allowed to have $i \in S$ or $S = \emptyset$.
\end{enumerate}
A well-formed sequence of these four operations is called a \emph{multi-$k$-expression} which we treat as a rooted tree where the nodes naturally correspond to operations of the expression.
A multi-$k$-expression is \emph{linear} if for every union-node at least one of the children is an introduce-node.
For a node $x$ of this tree, the multi-$k$-labeled graph arising in the subtree rooted at $x$ is denoted by $(G_x, \lab_x)$.
And $V_x$ and $E_x$ denote the sets of vertices and edges of $G_x$, respectively.
Note that for the root $r$ of the expression, the (unlabeled) graph $G_x$ is a subgraph of $G_r$ for every node $x$. 
Hence the number of vertices and edges in $G_x$ is upper-bounded by $|V_r|$ and $|E_r|$, respectively.
Chekan and Kratsch~\cite{abs-2307-04628} showed that one can simplify any multi-expression as follows, we will use this in our algorithmic applications:

\begin{lemma}[Lemma 37 in~\cite{abs-2307-04628}]\label{lem:mcw-special-relabels}
    Let $\phi$ be a multi-$k$-expression of a graph $G$ on $n$ vertices.
    Then given $\phi$ in time polynomial in $|\phi|$ and $k$, we can compute a multi-$k$-expression $\xi$ of $G$ such that in $\xi$ there are at most $\mathcal{O}(k^2 \cdot n)$ nodes and for every node $x$ the following holds. 
    If $x$ is a $\rho_{i \to S}$-node for some label $i \in [k]$ and some label set $S \subseteq [k]$, then we have $S = \emptyset$ or~$S = \{i, j\}$ for some $j \neq i \in [k]$.
    And if $x$ is a $v\langle S \rangle$-node for some label set $S \subseteq [k]$ and some vertex $v$, then we have $|S| = 1$.
\end{lemma}

In the view of this lemma, we say that an operation of form $\rho_{i \to \emptyset}$ \emph{forgets label~$i$}. 
And an operation of form $\rho_{i \to \{i, j\}}$ \emph{adds label $j$ to label $i$}.
We always assume that the number of labels used by an expression is upper-bounded by the number of vertices in the arising graph, i.e., $k \leq n$.  
The best known FPT-approximation of multi-clique-width goes via rank-width and has a double-exponential approximation ratio~\cite{ChekanK23}.
For this reason, as usual for results of this flavor, in our algorithms we assume that a multi-$k$-expression is a part of the input.
 
A \emph{multigraph} is a pair $M = (V, E)$ where $V$ is a finite set of vertices and every element of the finite set~$E$ is a pair $\big(e, \{u, v\}\big)$, called a \emph{(multi-)edge}, with $u, v \in V$.
We require that for any pair $\big(e, \{u, v\}\big) \neq \big(e', \{u', v'\}\big)$, we have $e \neq e'$.
Then $e$ is an identifier of the multi-edge~$\big(e, \{u, v\}\big)$, thus we call $u$ and $v$ the \emph{endpoints} of $e$. 
And we sometimes simply speak of the edge~$e$, for example, if the endpoints of~$e$ are irrelevant.
First, note that this definition allows distinct multi-edges to have the same endpoints.
Second, the endpoints $u$ and $v$ may, in general, be equal, allowing the \emph{loops}.
We sometimes use $V(M)$ and $E(M)$ to refer to $V$ and $E$, respectively.
For two vertices $u, v \in V$, the \emph{multiplicity} of $\{u,v\}$ is the number of edges with the endpoints $u$ and $v$.
For a vertex $v \in V$, the value~$\deg_M(v)$ is the \emph{degree} of $v$ in $M$, i.e., $\deg_M(v) = \sum_{\big(e, \{u,w\}\big) \in E} \alpha\big(v, \{u,w\}\big)$ where
\[
    \alpha\big(v,\{u,w\}\big) = 
    \begin{cases}
        0 & \text{if } v \notin \{u,w\} \\
        1 & \text{if } v \in \{u,w\} \text{ and } \{u,w\} \neq \{v\} \\
        2 & \text{if } \{u,w\} = \{v\}.
    \end{cases}
\]
Note that every loop at $v$ contributes two to the degree of~$v$.
We define the binary relation~$\sim_M$ on the vertex set $V$ of $M$ so that for any two vertices~$u \neq v \in V$ of~$M$, we have~$u \sim_M v$ if and only if there exists an edge~$\big(e, \{u, v\}\big) \in E$.
A \emph{connected component} of~$M$ is an equivalence class of the reflexive-transitive closure of the relation~$\sim_M$.

\section{Lower bound for Max Cut}\label{sec:max-cut}

In the \textsc{Max Cut} problem, given a graph $G$ and a value $\budget$, we are asked whether there is a partition $(V_1, V_2)$ of $G$ satisfying $\abs*{E_G(V_1, V_2)} \geq \budget$. 
In this section, we prove the following.
\thmMaxCut*
As a starting point of our reduction we use the \multicolored{} problem.
In this problem, we are given a graph $G = (U_1 \cup \dots \cup U_{k'}, E)$ where $U_1,\dots,U_{k'}$ are pairwise disjoint equal-sized independent sets, each of size $n'$.
The goal is to find a multicolored set~$X$---a set that contains exactly one vertex from each of the sets $U_1, \dots, U_{k'}$---that is an independent set in $G$.
The following lower bound for this problem is well-known:
\footnote{The lower bound in~\cite{LokshtanovMS11} is stated for \textsc{Multicolored Clique} but there is a trivial parameter-preserving reduction from \textsc{Multicolored Clique} to \multicolored{} which simply flips adjacencies between all pairs of sets $U_i$ and $U_j$ with $i \neq j$.}
\begin{theorem}[Theorem 5.2 in~\cite{LokshtanovMS11}]
\label{thm:multicolored-lb}
    Unless ETH fails, there is no $f(k') \cdot (n')^{o(k')}$ time algorithm for \multicolored{} for any computable function $f$.
\end{theorem}

So we now fix an instance $G = (U_1 \cup \dots \cup U_{k'}, E)$ of \multicolored.
We may assume that $n'>1$ holds as the instance is trivial otherwise.
We define the value~$n = n' - 1$.
In the following, we introduce some basic gadgets used in our reduction, then we present a high-level overview of the reduction itself, after that we formalize the construction, and finally we prove its correctness. 

\subsection{Basic gadgets}
Given a graph $Q$, the value $\mcut(Q)$ is defined as the maximum $\abs*{E_Q(V_1, V_2)}$ over all partitions~$(V_1, V_2)$ of $Q$.
A partition $(V_1, V_2)$ of $Q$ is called \emph{optimal} if $\abs*{E_Q(V_1, V_2)} = \mcut(Q)$, and it is called \emph{suboptimal} otherwise.

Let $D \geq 1$ be a value to be fixed later.
And let 
\begin{equation}\label{eq:def-of-c}
    C = D^2 \cdot {2n \choose 2} + 1.
\end{equation}
We use the values $C$ and $D$ in the construction of our gadgets in such a way that the number of edges crossed by a suboptimal partition of one gadget $Q$ is at most $\mcut(Q) - D$.
We will choose $D$ sufficiently large (but still polynomial in $n'$ and $k'$) so that every optimal partition of our final graph needs to respect the properties imposed by all the gadgets.

We will rely on the following gadgets introduced by Fomin et al.~\cite{FominGLS14}.
For two distinct vertices~$u$ and $v$, an $F$-\emph{gadget} between $u$ and $v$, denoted by $F(u, v)$, is a graph consisting of~$C$ vertex-disjoint paths, each of length 2, between $u$ and $v$. 
For two distinct vertices~$u$ and $v$, an~$F'$-\emph{gadget} between $u$ and $v$, denoted by $F'(u, v)$, is a graph consisting of $C$ vertex-disjoint paths, each of length 3, between $u$ and~$v$ (see \cref{fig:simple-gadgets}). 
Then these gadgets have the following properties:
\begin{figure}[b]
    \includegraphics{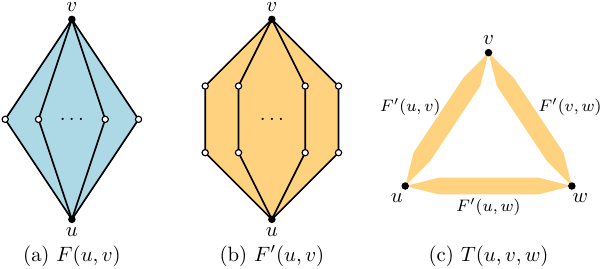}
    \centering
    \caption{Gadgets (a) $F$, (b) $F'$, and (c) $T$.} 
    \label{fig:simple-gadgets}
\end{figure}
\begin{lemma}[Lemma 4.2 in~\cite{FominGLS14}]\label{lem:F-gadgets}
The following properties apply (they are also summarized in \cref{tab:Fgadgets}):
    \begin{enumerate}
        \item We have $\mcut(F(u, v)) = 2C$ and $\mcut\big(F'(u, v)\big) = 3C$.
        \item For every partition $(V_1, V_2)$ of $F(u, v)$ with $u \in V_1$ and $v \in V_2$, we have $\abs*{E_{F(u,v)}(V_1, V_2)} \leq \mcut\big(F(u, v)\big) - C$. And every partition of $\{u,v\}$ placing the two vertices on different sides can be extended to a partition $(V_1, V_2)$ of $F(u, v)$ satisfying $\abs*{E_{F(u,v)}(V_1, V_2)} = \mcut\big(F(u, v)\big) - C$.
        \item For every partition $(V_1, V_2)$ of $F'(u, v)$ with $u \in V_1$ and $v \in V_1$, we have $\abs*{E_{F'(u,v)}(V_1, V_2)} \leq \mcut\big(F'(u, v)\big) - C$.
        And every partition of $\{u,v\}$ placing the two vertices on the same side can be extended to a partition $(V_1, V_2)$ of $F'(u, v)$ satisfying $\abs*{E_{F'(u,v)}(V_1, V_2)} = \mcut\big(F'(u, v)\big) - C$.
        \item Any partition of $\{u, v\}$ such that $u$ and $v$ are on the same side / on different sides can be extended to an optimal partition of $F(u, v)$ / $F'(u, v)$. 
    \end{enumerate}
    \begin{table}[htb]
        \centering
        \begin{tabular}{c|c|c}
             Gadget $Q$ & $u$ and $v$ on the same side & $u$ and $v$ on different sides\\
             \hline
             $F(u,v)$ & $\mcut(F) = 2C$ & $\mcut(F) - C$  \\
             \hline
             $F'(u,v)$ & $\mcut(F') - C$ & $\mcut(F') = 3C$
        \end{tabular}
        \caption{Maximum value of $\abs*{E_Q(V_1,V_2)}$ over the partitions $(V_1,V_2)$ of a gadget $Q$ depending on whether $u$ and $v$ on the same side or on different sides. 
        Furthermore, every partition of $\{u,v\}$ can be extended to a partition $(V_1, V_2)$ of the whole gadget so that the number of crossed edges is equal to the corresponding table entry.
        This follows from \cite[Lemma 4.2]{FominGLS14}.}
        \label{tab:Fgadgets}
    \end{table}
    
\end{lemma}
To simplify the notation, we will sometimes discard the endpoints of the gadget if either they are clear from the context or we only care about the maximum number of edges crossed.
In particular, we will use the notation $\mcut(F)$ and $\mcut(F')$.
In view of the above lemma, 
for vertices~$u$ and $v$
we say that a partition of a superset of $\{u,v\}$ \emph{satisfies} $F$-gadget (resp.~$F'$-gadget) if the vertices $u$ and $v$ are on the same side (resp.\ on different sides) of the partition.
On a high-level, the choice of a sufficiently large value $C$ will later ensure that if a partition of the final graph falsifies at least one of the used gadgets, there are not enough edges to compensate for the loss of $C$ in the remaining graph.
In other words, every optimal partition of the final graph will necessarily satisfy all gadgets.

Let now $u, v, w$ be three pairwise distinct vertices and let $T(u, v, w)$ denote the gadget that consists of these vertices, an $F'$-gadget between $u$ and $v$, an $F'$-gadget between $u$ and~$w$, and an $F'$-gadget between $v$ and $w$ (see \cref{fig:simple-gadgets}).
The following lemma follows from the properties of $F'$-gadgets:
\begin{lemma}\label{lem:triangle-gadget}
    \begin{enumerate}
        \item We have $\mcut\big(T(u, v, w)\big) = 3\cdot \mcut(F')-C$.
        \item For every partition $(V_1, V_2)$ of $T(u, v, w)$ with $u, v, w \in V_1$, we have $\abs*{E_{T(u, v, w)}(V_1, V_2)} \leq \mcut\big(T(u,v,w)\big) - 2C$.
        \item Any partition of $u, v, w$ not placing all vertices to the same side can be extended to an optimal partition of $T(u, v, w)$.
    \end{enumerate} 
\end{lemma}
\begin{proof}
    We rely on the properties of $F'$-gadgets from~\cref{lem:F-gadgets}.
    If all three vertices~$u, v, w$ are placed on the same side, then each of the three $F'$-gadgets is not satisfied, i.e., $\mcut(T(u, v, w)) \leq 3 \cdot \big(\mcut(F')- C\big) = 3 \cdot \mcut(F') - 3C$ holds.
    And if each side of the partition of $u, v, w$ contains at least one vertex, then precisely one of the $F'$-gadgets is not satisfied, namely the one with endpoints on the same side.
    So we can extend it to a partition $(V_1, V_2)$ of $T(u, v, w)$ satisfying $\abs*{E_{T(u, v, w)}(V_1, V_2)} = 3 \cdot \mcut(F') - C = \mcut\big(T(u, v, w)\big)$.
\end{proof}
We use the shortcut $\mcut(T)$ for $\mcut\big(T(u, v, w)\big)$.
In view of this lemma, 
for vertices $u$ and $v$
we say that a partition of a superset of $\{u,v,w\}$ \emph{satisfies} $T$-gadget if this partition does not place all vertices $u,v,w$ on the same side.

Next, for integers $1 \leq \alpha \leq t \leq n$ and pairwise distinct vertices $x_1, \dots, x_t, y, z$ we define the gadget $\gadgetIf_{\alpha, t}(x_1, \dots, x_t, y, z)$~(see \cref{fig:h-if-gadget}) as follows.
The vertices $x_1, \dots, x_t, y, z$ are called the \emph{entry points} of $\gadgetIf_{\alpha, t}(x_1, \dots, x_t, y, z)$.
We rely on the idea by Fomin et al.~\cite{FominGLS14} and generalize what they call ``an $H_{s,t}(x_1, \dots, x_s,y)$ gadget'' to our purposes.
Our gadget gets one more input vertex $z$ and exhibits different types of behavior depending on whether $y$ and $z$ are placed on the same side of the partition.
The gadget $\gadgetIf_{\alpha, t}(x_1, \dots, x_t, y, z)$ has the following structure.
First, for every $i \in [2n]$ we add a set $P_i = \big\{p_{i,j} \mid j \in [D]\big\}$ of new vertices called the \emph{$i$-th column}.
For $j \in [D-1]$ we add an $F$-gadget between $p_{i,j}$ and $p_{i,j+1}$.
And we add an edge $p_{i, j} p_{i', j'}$ for all $i \neq i' \in [2n]$ and $j, j' \in [D]$ so that the columns form a complete $2n$-partite graph.
In our future analysis, we will refer to the subgraph created so far by $H$.
Next, for $i \in [t]$ we add an $F$-gadget between $p_{i, D}$ and $x_i$, and we say that $P_i$ is \emph{attached} to $x_i$.
Further, for $i \in [n - \alpha]$, we add a new vertex $r_i$, an $F$-gadget between~$r_i$ and $y$, and a gadget $T(p_{n+i, D}, r_i, z)$.
And we say that the column $P_{n+i}$ is \emph{attached} to this~$T$-gadget.
This concludes the construction of the gadget $\gadgetIf_{\alpha,t}(x_1, \dots, x_t, y, z)$.

\begin{figure}[hbt]
    \includegraphics[width=0.75\linewidth]
    {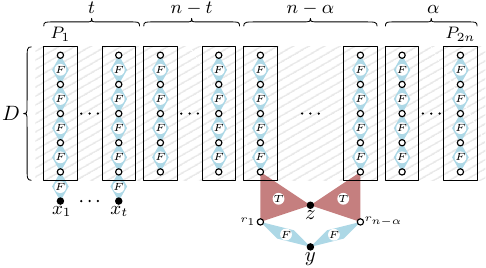}
    \centering
    \caption{Gadget $\gadgetIf_{\alpha, t}(x_1,\dots,x_t,y,z)$ (adapted Figure 2 in \cite{FominGLS14} by Fomin et al.). There is a complete bipartite graph between any two columns $P_i$ and $P_j$. The gray hatched area represents the subgraph $H$.}
    \label{fig:h-if-gadget}
\end{figure}

To prove the crucial properties of $\gadgetIf_{\alpha, t}(x_1,\dots,x_t,y,z)$, we first show some properties of the subgraph $H$ of this gadget:
\begin{lemma}\label{lem:properties-of-H}
    \begin{enumerate}
        \item We have $\mcut(H) = 2n \cdot (D-1) \cdot \mcut(F) + n^2 \cdot D^2$.
        \item For every optimal partition $(V_1,V_2)$ of $H$, every column is either contained in $V_1$ or it is contained in $V_2$. Furthermore, there exist precisely $n$ columns contained in $V_1$.
        \item Every partition $(V'_1, V'_2)$ of $\bigcup_{i \in [2n]} P_i$ such that $V'_1$ consists of precisely $n$ columns can be extended to an optimal partition $(V_1, V_2)$ of $H$.
        \item For every partition $(V_1, V_2)$ of $H$ such that either some column contains vertices in both $V_1$ and $V_2$, or each column is fully contained in one of $V_1$ and $V_2$ but $V_1$ contains a number of columns other than $n$, we have $\abs*{E_H(V_1, V_2)} \leq \mcut(H) - D^2$.
    \end{enumerate}
\end{lemma}
\begin{proof}
    There are ${2n \choose 2} \cdot D^2$ edges of $H$ with endpoints in different columns.
    First, consider a partition $(V_1, V_2)$ of $H$ in which there exists a column containing a vertex in each of $V_1$ and~$V_2$.
    Then at least one $F$-gadget with endpoints in this column is not satisfied, i.e., at most $\mcut(F) - C$ edges of this gadget are crossed (cf.\ \cref{lem:F-gadgets}).
    So we have
    \begin{align}
        \abs*{E_H(V_1, V_2)} &\leq \big(2n \cdot (D-1) \cdot \mcut(F) - C\big) + {2n \choose 2} \cdot D^2 \nonumber \\ 
        &\stackrel{\eqref{eq:def-of-c}}{=} 2n \cdot (D-1) \cdot \mcut(F) - 1 
        < 2n \cdot (D-1) \cdot \mcut(F) + n^2 \cdot D^2 - D^2. \label{eq:gap-h-gadget}
    \end{align}
    
    Next consider a partition $(V_1', V_2')$ of $\bigcup_{i \in [2n]} P_i$ for which there exists an integer $\ell$ with~$0 \leq \ell \leq n$ such that precisely $\ell$ columns are in $V'_1$ and the other $2n-\ell$ columns are in $V'_2$.
    By \cref{lem:F-gadgets} it is possible to extend it to a partition $(V_1, V_2)$ of $H$ which is optimal for every~$F$-gadget.  
    In this case we have
    \[
        \abs*{E_H(V_1, V_2)} \leq 2n \cdot (D-1) \cdot \mcut(F) + \ell \cdot (2n-\ell) \cdot D^2.
    \]
    The value $\ell = n$ is the unique maximizer of the product $\ell \cdot (2n-\ell)$ over the integers~$0 \leq \ell \leq 2n$.
    First, this implies the first three claims.
    Second, whenever $\ell \neq n$, we have $\abs*{E_H(V_1, V_2)} \leq \mcut(H) - D^2$.
    Together with \eqref{eq:gap-h-gadget}, this implies the last claim.
\end{proof}    

We emphasize that the gadget $H$ is independent of the vertices $x_1, \dots, x_t, y, z$ and the values $t$ and $\alpha$.
Now we are ready to prove the properties of $\gadgetIf_{\alpha, t}(x_1,\dots,x_t,y,z)$:
\begin{lemma}\label{lem:H-if-properties}
    \begin{enumerate}
        \item We have 
        \[
            \mcut\big(\gadgetIf_{\alpha,t}(x_1,\dots,x_t,y,z)\big) = \mcut(H) + (t + n - \alpha) \cdot \mcut(F) + (n - \alpha) \cdot \mcut(T).
        \]
        \item For every optimal partition $(V_1, V_2)$ of the gadget $\gadgetIf_{\alpha,t}(x_1,\dots,x_t,y,z)$ with $y,z\in V_2$, we have $\abs*{\{x_1, \dots, x_t\} \cap V_1} \leq \alpha$.
        \item Every partition $(V_1',V_2')$ of $\{x_1,\dots,x_t,y,z\}$ such that:
        \begin{itemize}
            \item either $y,z\in V_2'$ and $\abs*{\{x_1, \dots, x_t\} \cap V_1'} \leq \alpha$, 
            \item or $y \in V_2'$ and $z \in V_1'$ 
        \end{itemize}
        can be extended to an optimal partition $(V_1,V_2)$ of $\gadgetIf_{\alpha,t}(x_1,\dots,x_t,y,z)$. 

        \item For every partition $(V_1,V_2)$ of the gadget $\gadgetIf_{\alpha,t}(x_1,\dots,x_t,y,z)$ such that $y,z\in V_2$ and $\abs*{\{x_1, \dots, x_t\} \cap V_1} \geq \alpha+1$, we have 
        \[
            \abs*{E_{\gadgetIf_{\alpha, t}(\{x_1,\dots,x_t\},y,z)}(V_1,V_2)} \leq \mcut(\gadgetIf_{\alpha,t}(\{x_1,\dots,x_t\},y,z)) - D^2.
        \]
    \end{enumerate}
\end{lemma}

\begin{proof}
    Let $\gadgetIf$ be a shortcut for $\gadgetIf_{\alpha,t}(x_1,\dots,x_t,y,z)$.
    In this proof, by $F$-gadgets we refer to the $F$-gadgets of $\gadgetIf$ not contained in $H$.
    First of all observe that we have 
    \begin{equation}\label{eq:upper-bound-if}
        \mcut(\gadgetIf) \leq \mcut(H) + (t + n - \alpha) \cdot \mcut(F) + (n - \alpha) \cdot \mcut(T),
    \end{equation}
    this is simply because we can partition the edges of $\gadgetIf$ into the edges of the gadget~$H$,~$(t+n-\alpha)$~$F$-gadgets, and~$(n-\alpha)$ $T$-gadgets.
    
    Now consider an arbitrary partition $(V_1', V_2')$ of $\{x_1, \dots, x_t, y, z\}$ satisfying either $y,z \in V'_2$ and $\abs*{\{x_1, \dots, x_t\} \cap V'_1} \leq \alpha$, or $y \in V_2'$ and $z \in V_1'$.
    Let $I = \big\{i \in [t] \mid x_i \in V'_1\big\}$ and let~$\beta = \abs*{I}$.
    Then we extend $(V_1', V_2')$ to a partition $(V_1, V_2)$ of $H$ as follows.
    We start by putting for every $i \in [2n]$ all vertices of $P_i$:
    \begin{itemize}
        \item into $V_1$, if $i \in I$,
        \item into $V_1$, if $n+1 \leq i \leq n+(n - \beta)$,
        \item and into $V_2$, otherwise.
    \end{itemize}
    Observe that every column of $H$ is contained in either $V_1$ or $V_2$ and precisely $\beta + (n-\beta) = n$ columns are contained in $V_1$.
    So we extend the current partition to an optimal partition of $H$ (cf.\ \cref{lem:properties-of-H}).
    Next we put $r_1, \dots, r_{n-\alpha}$ into $V_2$.
    Now all $(n-\alpha+t)$ $F$-gadgets are satisfied: indeed, $x_i$ and $p_{i, D}$ are on the same side for $i \in [t]$, and $y$ and $r_i$ are on the same side for $i \in [n-\alpha]$. 
    So we extend the current partition to an optimal partition of $F$-gadgets.
    To see that the $T$-gadgets are also satisfied, we make a case distinction.
    
    First, suppose that $y,z \in V'_2$ and $\beta = \abs*{\{x_1, \dots, x_t\} \cap V'_1} \leq \alpha$ holds.
    We then have~$n + (n-\beta) \geq n + (n-\alpha)$.
    For every $i \in [n-\alpha]$, we then have $p_{n+i,D} \in V_1$ and $z \in V_2$ so the gadget $T(p_{n+i,D}, r_i, z)$ is satisfied (cf. \cref{lem:triangle-gadget}).  
    And we finally extend the current partition to an optimal partition of the gadget $T(p_{n+i,D}, r_i, z)$.
    
    We remain with the case $y \in V'_2$ and $z \in V'_1$.
    Then for every $i \in [n-\alpha]$, we have $r_i \in V_2$ and $z \in V_1$ so the gadget $T(p_{n+i,D}, r_i, z)$ is satisfied  (cf. \cref{lem:triangle-gadget}). 
    And we also extend the current partition to the optimal partition of $T(p_{n+i}, r_i, z)$.
    
    In both cases, we get 
    \[
        \abs*{E_{\gadgetIf}(V_1, V_2)} = \mcut(H) + (t + n - \alpha) \cdot \mcut(F) + (n - \alpha) \cdot \mcut(T)
    \]
    because we have an optimal partition for each of the $(t+n-\alpha)$ $F$-gadgets, $(n-\alpha)$ $T$-gadgets, and the subgraph $H$, and those subgraphs partition the edge set of $\gadgetIf$.
    This implies that inequality in \eqref{eq:upper-bound-if} is, in fact, an equality proving the first and the third item.
    Moreover, a partition is optimal if and only if it is optimal in each of those gadgets.
    Further, due to~$D^2 < C$ (cf.\ \eqref{eq:def-of-c}), any partition $(V_1, V_2)$ of $\gadgetIf$ that is suboptimal for at least one of those gadgets satisfies    \begin{equation}\label{eq:suboptimal-if-gadget}
        \abs*{E_{\gadgetIf}(V_1, V_2)} \leq \mcut(\gadgetIf) - D^2
    \end{equation}
    by \cref{lem:F-gadgets,lem:triangle-gadget,lem:properties-of-H}.
    
    Now consider a partition $(V_1, V_2)$ of $\gadgetIf$ satisfying $y, z \in V_2$ and $\abs*{E_{\gadgetIf}(V_1, V_2)} = \mcut(\gadgetIf)$.
    Then this partition is optimal for $H$ and all $F$- and $T$-gadgets. 
    Due to~$F$-gadgets, we have~$r_1, \dots, r_{n-\alpha} \in V_2$.
    Then for every $i \in [n - \alpha]$, because the gadget~$T(p_{n+i, D}, r_i, z)$ is satisfied and we have $z, r_i \in V_2$, we get $p_{n+i,D} \in V_1$.
    Now due to the optimal partition of $H$ (cf.\ \cref{lem:properties-of-H}), the following properties apply.
    First, the $n-\alpha$ columns~$P_{n+1}, \dots, P_{n+(n-\alpha)}$ are contained in $V_1$.
    Second, there exist precisely $\alpha$ further columns in~$P_1, \dots, P_n$, $P_{n+(n-\alpha)+1}$, $\dots$, $P_{2n}$ contained in $V_1$.
    Finally, due to~$F$-gadgets, for every $i \in [t]$, the vertex $x_i$ is in $V_1$ if and only if $P_i$ is contained in $V_1$.
    In particular, at most $\alpha$ vertices from $x_1, \dots, x_t$ belong to $V_1$.
    This shows that every partition $(V_1, V_2)$ of $\gadgetIf$ with $y, z \in V_2$ and $\abs*{\{x_1, \dots, x_t\} \cap V_1} \geq \alpha+1$ is suboptimal, in particular, proving the second item.
    Furthermore, above we argued that every suboptimal partition satisfies \eqref{eq:suboptimal-if-gadget} proving the last item and concluding the proof.
\end{proof}
Note that the value $\mcut\big(\gadgetIf_{\alpha,t}(x_1,\dots,x_n,y,z)\big)$ depends on $\alpha$ and $t$ but not on the entry points $x_1,\dots,x_n,y,z$ so we use the shortcut $\mcut\big(\gadgetIf_{\alpha,t}\big)$ for this value.

\paragraph*{Families $\cS$ and $\tildecS$}
Now we define two useful set families on which our reduction will rely on.
Recall that $k'$ denotes the size of the sought-after multicolored independent set in the $k'$-partite graph~$G$ in which each of the independent sets $U_1, \dots, U_{k'}$ has size $n' = n + 1$. 
We may assume that~$k' = {k \choose{k/2}} / 2$ holds for some integer $k$: indeed, we can achieve that $k'$ is of this form by at most doubling it and adding to $G$ an independent set with a suitable number of vertices each of which is adjacent to every vertex in $U_1 \cup \dots \cup U_{k'}$.
Let $k \in \cO(\log k')$ thus be fixed in the remainder of this section.
We define the sets:
\begin{align*}
    \cS = \left\{S \in {[k] \choose {k/2}} \mid 1 \in S\right\} && \tildecS = \left\{S \in {[k] \choose {k/2}} \mid 1 \notin S\right\}.
\end{align*}
For a set $S \subseteq [k]$, the set $\overline{S}$ is defined as $[k] \setminus S$.
Then we have: 
\begin{observation}
    \begin{enumerate}
        \item We have $\abs*{\cS} = \abs*{\tildecS} = k'$.
        \item The sets $\cS$ and $\tildecS$ are disjoint.
        \item For every set $S \in {[k] \choose{k/2}}$, precisely one of the sets $S$ and $\overline S$ contains the element $1$.
    \end{enumerate}
\end{observation}
So we fix an arbitrary bijection $\phi \colon [k'] \to \cS$. 
And for every $i \in [k]$, let $U_i = \big\{u_i^\gamma \mid \gamma \in [n]_0\big\}$ be the vertices in the independent set $U_i$.

Now are ready to present our lower-bound construction.
We start with a high-level overview and provide formal details after that.
In the end, we argue the correctness and provide a multi-expression of the arising graph using only $\cO(k)$ labels.

\subsection{High-level overview}
The construction of the instance $(G^*, \budget)$ of \textsc{Max Cut} can be sketched as follows (see \cref{fig:construction,fig:construction-path} for an illustration).
Even though we rely on some gadgets introduced by Fomin et al.~\cite{FominGLS14}, our reduction is quite different to theirs, for instance, we start with a different problem and we do not preserve the structure of the starting graph.

The main component of our graph $G^*$ is a \emph{selection gadget} such that every optimal partition $(V_1,V_2)$ of this gadget is associated with a multicolored set $X$ of $G$; we say that~$(V_1,V_2)$ \emph{selects} $X$ on this gadget.
Our graph $G^*$ has a copy of this selection gadget for each edge of $G$.
And these copies are connected in a path-like way such that every partition~$(V_1,V_2)$ of $G^*$ with $\abs*{E_{G^*}(V_1,V_2)} = \budget$ selects the same multicolored set in each selection gadget.
Moreover, each selection gadget associated with an edge $v_i w_i$ is connected to a combination of five $\gadgetIf$-gadgets forcing the selected multicolored set $X$ to contain at most one vertex among $v_i$ and $w_i$.

We start by describing how a single selection gadget works. 
The function $\phi$ is a bijection such that each of the independent sets $U_1, \dots, U_t$ of the graph $G = (U_1 \cup \dots \cup U_{k'}, E)$ corresponds to a set $S = \phi(i)$ from $\cS$.
Therefore, we start by adding a set $A_S$ on $n$ vertices forming a clique to the graph $G^*$.
In our reduction, \emph{selecting} a vertex $u_i^\gamma \in U_i$ into a multicolored set $X$ of $G$ transfers to placing precisely $\gamma$ vertices of $A_S$ to the side $V_1$ of a partition $(V_1, V_2)$ of $G^*$, i.e., having the equality $\abs*{A_S \cap V_1} = \gamma$.

However, the partitions of the clique $A_S$ crossing the largest number of edges are the ones where precisely half (i.e., $n/2$) of the vertices of $A_S$ are in $V_1$.
But our aim, for now, is to allow selecting any vertex of $U_i$, i.e., to place any number of vertices of $A_S$ into $V_1$ and still potentially be able to extend it to an optimal partition of $G^*$.
To achieve this, we add a set $B_S$ of $n$ new vertices to $G^*$ so that $A_S \cup B_S$ forms a clique.
Now no matter what the partition of $A_S$ is, it is possible to ``compensate'' it by partitioning $B_S$ accordingly so that the clique $A_S \cup B_S$ is crossed optimally.
Note that in this case $A_S$ and $B_S$ are partitioned in a \emph{complementary} way, i.e., $\abs*{A_S \cap V_1} = \abs*{B_S \cap V_2}$.

Next we ensure that the vertices selected in each of $U_1, \dots, U_{k'}$ form an independent set in $G$.
For every edge $u_{i_1}^{\alpha} u_{i_2}^{\beta}$ of $G$, we are allowed to select at most one vertex among the endpoints $u_{i_1}^{\alpha}$ and $u_{i_2}^{\beta}$ of this edge.
So in the graph $G^*$, we forbid that at the same time, exactly $\alpha$ vertices of $A_{\phi(i_1)}$ and exactly $\beta$ vertices of $A_{\phi(i_2)}$ are put into $V_1$.
In other words, at least one of the following four cases has to apply:
\begin{enumerate}
    \item at most $\alpha - 1$ vertices of $A_{\phi(i_1)}$ are put into $V_1$,
    \item \label{item:at-least-u} or at least $\alpha + 1$ vertices of $A_{\phi(i_1)}$ are put into $V_1$, 
    \item or at most $\beta - 1$ vertices of $A_{\phi(i_2)}$ are put into $V_1$,
    \item \label{item:at-least-v} or at least $\beta + 1$ vertices of $A_{\phi(i_2)}$ are put into $V_1$.
\end{enumerate}
We use an $\gadgetIf$-gadget to implement a disjunction of those cases. 
And to capture each of the cases, we also use a $\gadgetIf$-gadget (cf.\ \cref{lem:H-if-properties}).
An $\gadgetIf$-gadget enforces that at most a certain number of vertices are put into $V_1$.
We provide more details on the ``at least''-cases later.
Now if we attach the gadgets for all edges of $G$ (there may be $\Omega(n^2)$ of them) to the same selection gadget, we get a huge number of ``neighborhoods'' preventing an expression of low multi-clique-width.

Instead, as standard in lower-bound constructions of this flavor, we use copies of selection gadgets.
Namely, for every edge of $G$, we create a copy of $\bigcup_{S \in \cS} A_S \cup B_S$ and attach the corresponding five $\gadgetIf$-gadgets to this copy. 
It remains to ensure that an optimal partition of the final graph selects the same multicolored set on each the selection gadgets.
Let $m$ be the number of edges.
So for every $j \in [m-1]$ and every vertex of $B_S$ in the $j$-th copy, we connect this vertex and its counterpart in the $j+1$-st copy of $A_S$ by an $F'$-gadget (cf.~\cref{lem:F-gadgets}) to ensure that these vertices appear on different sides of the partition.
This ensures that any optimal partition $(V_1,V_2)$ splits $j$-th copy of $B_S$ and the $j+1$-st copy of $A_S$ in a complementary way. 
Since the $j$-th copies of $A_S$ and $B_S$ also need to be split this way, the partitions on the $j$-th copy of $A_S$ and on the $j+1$-st copy of $A_S$ coincide, and therefore the same selection is made on all selection gadgets.

\begin{figure}[t]
    \includegraphics{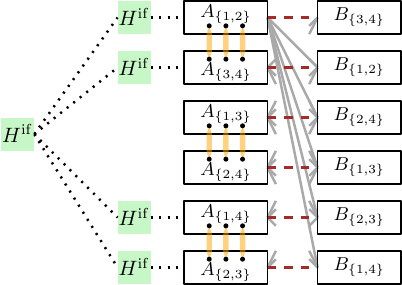}
    \centering
    \caption{A selection gadget for $n = 3$ and $k = 4$. 
    For every set $S \in {[4] \choose 2}$, there are $3$ vertices in each of the sets $A_S$ and $B_S$.
    Each white rectangle is a clique.
    Each vertex in the set $A_S$ is connected by an orange $F'$-gadget to its counterpart in $A_{\overline S}$.
    Between $A$- and $B$-vertices, there is a gray ``thickened anti-matching'', i.e., almost all edges are present, but there are no edges (in dashed red) between $A_S$ and $B_{\overline{S}}$ for every set $S$.
    The attached $\gadgetIf$-gadgets handle an edge of $G$ with one endpoint in $U_{\phi^{-1}(\{1,2\})}$ and the other in $U_{\phi^{-1}(\{1,4\})}$.
    Dotted lines sketch the connections from those $\gadgetIf$-gadgets to their entry points.
    We emphasize that those dotted lines are not simply edges/non-edges but rather the~$\gadgetIf$-connections sketched in \cref{fig:h-if-gadget}.
    } 
    \label{fig:construction}
\end{figure}

As it is now, the multi-clique-width of the construction is very high.
In every copy there is an induced matching of size $k'$ in the bipartite subgraph with sides $\bigcup_{S \in \cS} A_S$ and $\bigcup_{S \in \cS} B_S$: for every set $S \in \cS$, it uses one edge between $A_S$ and $B_S$.
Further, from every vertex in $B_S$, there is an $F'$-gadget to its counterpart in the next copy of $A_S$. 
And on the level of \emph{cuts of the expression}, each $F'$-gadget behaves similarly to an edge of an induced matching: it yields a private neighborhood disjoint from the neighborhoods of the other vertices.
So there is a long sequence of \emph{cuts} with, alternatingly, large induced matchings and large numbers of~$F'$-gadgets.
To overcome this issue, we finally adapt every copy as follows.

First, we extend the whole construction from the family $\cS$ to the family $\cS \cup \tildecS$.
And we add an $F'$-gadget between every vertex of $A_S$ and its counterpart in $A_{\overline{S}}$ in the same copy to enforce complementary behavior.
Also, to ensure that \emph{at least} a certain number of vertices of $A_S$ are put into $V_1$ as desired (cf.\ \cref{item:at-least-u,item:at-least-v}), we use a $\gadgetIf$-gadget attached to $A_{\overline{S}}$ and ensure that \emph{at most} the same number of vertices of $A_{\overline{S}}$ are in $V_1$.
Second, to decrease the size of an induced matching, in every copy of a selection gadget we make a clique of each set~$A_S \cup B_{T}$ satisfying $S \in \cS \cup \tildecS$ and $T \in (\cS \cup \tildecS) \setminus \{\overline{S}\}$.
Observe that because $\cS \cup \tildecS = {[k] \choose {k/2}}$ holds, for every set $S \in \cS \cup \tildecS$, the only set $T \in \cS \cup \tildecS$ with $S \cap T = \emptyset$ is precisely $T = \overline{S}$.   
In other words, the set $A_S \cup B_{T}$ does not become a clique only if $S$ and $T$ are disjoint.
This characterization will be useful to construct a multi-expression in a moment.
This way, the cut between $\bigcup_{S \in \cS} A_S$ and $\bigcup_{S \in \cS} B_S$ of the same copy from a ``thickened matching'' becomes a ``thickened anti-matching'', i.e., there are all edges except for the ones with one endpoint in~$A_S$ and the other in $B_{\overline{S}}$ for $S \in \cS \cup \tildecS$.

\begin{figure}[t]
    \includegraphics{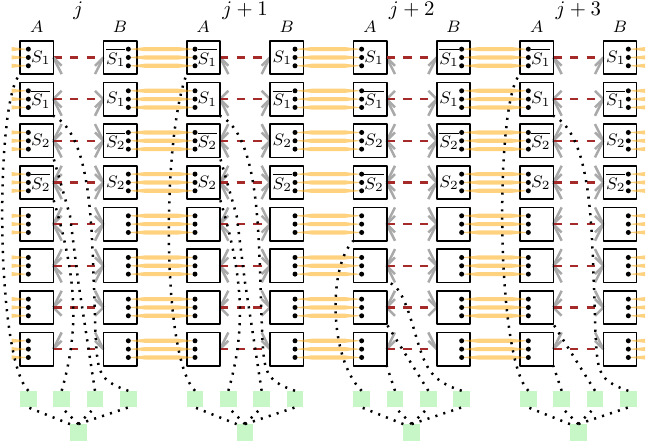}
    \centering
    \caption{Part of the lower-bound construction for four copies between $j$ and $j+3$.
    Between $A$ and $B$ of the same copy all edges are present except when having one endpoint in some $A_S$ and the other in $B_{\overline{S}}$.
    There is an $F'$-gadget between every vertex of $B_S$ and its counterpart in the next copy of $A_S$.
    Five $\gadgetIf$-gadgets handling the corresponding edge of $G$ are attached to each copy.
    Between~$U_{i_1}$ and $U_{i_2}$ with $i_1 \neq i_2$, there may be multiple edges of $G$.
    So there may be multiple copies (e.g., $j$ and $j+1$) where the $\gadgetIf$-gadgets are attached to the vertices of $A_S$ for the same sets~$S \in \cS \cup \tildecS$---the actual connections sketched with dotted lines differ for different edges though.
    } 
    \label{fig:construction-path}
\end{figure}

Crucially, we can now create the edges of the graph $G^*$ in a multi-expression using~$\cO(k)$ labels only.
Essentially, every vertex of the set $A_S$ or $B_S$ gets label set $\left\{q^{\la} \mid q \in S\right\}$ or~$\left\{q^{\lb} \mid q \in S\right\}$, respectively.
And we join $q^{\la}$ with $q^{\lb}$ for every $q \in [k]$. Now there is a complete bipartite graph between $A_S$ and $B_{T}$ if and only if there is a value $q \in S \cap T$, i.e., if~$T \neq \overline{S}$.
We also ensure that, vaguely speaking, the $F'$-gadgets between the $j$-th copy of~$B_S$ and the~$j+1$-st copy of $A_S$ do not belong to the same \emph{cut} of the expression.
Thus we can reuse a constant number of labels for all $F'$-gadgets.
At most five $\gadgetIf$-gadgets are attached to every copy and a constant number of reusable labels suffices to create those edges too.
 
\subsection{Construction}
The instance $(G^*, \budget)$ of \textsc{Max Cut} is defined as follows.
Let $m = \abs*{E}$ be the number of edges in $G$ and let $E = \{v_1w_1, \dots, v_m w_m\}$ be those edges.
We construct the graph $G^*$ as follows.
We start with two vertices $d_1$ and $d_2$ connected by an $F'$-gadget.
And we add a vertex $d_2'$ and connect it to $d_2$ by an $F$-gadget.
The aim of these gadgets is to ensure that in any optimal partition of $G^*$, the vertex $d_1$ is on the one side and the vertices $d_2$ and $d_2'$ are on the other side.
So the function of the vertices $d_1, d_2, d_2'$ is to be the reference vertices for the sides of a partition.
And we proceed as follows for every $j \in [m]$.

For every set $S \in \cS \cup \tildecS$, we add the sets 
\begin{align*}
    A_S(j) = \big\{a^i_S(j) \mid i \in [n]\big\} && \text{ and} && B_S(j) = \big\{b^i_S(j) \mid i \in [n]\big\}
\end{align*}
of $n$ new vertices each such that $A_S(j)$ induces a clique and $B_S(j)$ induces a clique.
For every set $S \in \cS$ and every index $i \in [n]$, we add an $F'$-gadget between the vertices $a_S^i(j)$ and~$a_{\overline{S}}^i(j)$.
Also for all pairs $S, T \in \cS \cup \tildecS$ satisfying $T \neq \overline S$ and all $i, i' \in [n]$, we add the edge~$a_S^i(j) b_{T}^{i'}(j)$. 
And we define the sets 
\begin{align*}
    A(j) = \bigcup_{S \in \cS \cup \tildecS} A_S(j), && A = \bigcup_{j \in [m]} A(j), && B(j) = \bigcup_{S \in \cS \cup \tildecS} B_S(j), && B = \bigcup_{j \in [m]} B(j).
\end{align*}

The pair $v_j w_j \in E$ is an edge of $G$. 
So let the indices $i_1(j) \neq i_2(j) \in [k']$ and $\alpha(j), \beta(j) \in [n]_0$ be such that $v_j = u_{i_1(j)}^{\alpha(j)}$ and $w_j = u_{i_2(j)}^{\beta(j)}$.
Further let $S_1(j) = \phi\big(i_1(j)\big) \in \cS$ and~$S_2(j) = \phi\big(i_2(j)\big) \in \cS$ be the sets from the family $\cS$ corresponding to the independent sets $U_{i_1(j)}$ and~$U_{i_2(j)}$ containing the vertices $v_j$ and $w_j$, respectively.
We add the set
\begin{align*}
    Z(j) = &\left\{z_{1}^{<\alpha(j)}(j) \mid \alpha(j) \neq 0\right\} \cup 
            \left\{z_{1}^{>\alpha(j)}(j) \mid \alpha(j) \neq n\right\}  \\
            \cup
            &\left\{z_{2}^{<\beta(j)}(j) \mid \beta(j) \neq 0\right\} \cup 
            \left\{z_{2}^{>\beta(j)}(j) \mid \beta(j) \neq n\right\}
\end{align*} 
of pairwise distinct new vertices.
Note that each of these four sets consists of at most one vertex: for example, the first set consists of a single vertex $z_{1}^{<\alpha(j)}(j)$ if $\alpha(j)$ is non-zero, and it is empty otherwise. 
So the set $Z(j)$ contains at most four vertices.
Also note that due to~$n'>1$, i.e., $n > 0$, the set $Z(j)$ consists of at least two vertices.

We add the gadget $\gadgetIf_{\abs*{Z(j)} - 1, \abs*{Z(j)}}\big(Z(j), d_2, d_2'\big)$ denoted by $\gadgetIf\big(Z(j)\big)$.
Because the vertices $d_2$ and $d_2'$ may be assumed to be on the same side of the partition of $G^*$ as argued before, this gadget ensures that at most $\abs*{Z(j)} - 1$ out of $\abs*{Z(j)}$ vertices of $Z(j)$ are on the different side (cf.\ \cref{lem:H-if-properties}), i.e., at least one of them is on the same side as $d_2$.
Informally speaking, this employs the distinction of the cases how at least one of the endpoints $v_j$ and $w_j$ of the current edge may be not selected.
To handle these cases, we add the following gadgets:
\begin{itemize}
    \item If $z_{1}^{<\alpha(j)}(j) \in Z(j)$, add $\gadgetIf_{\alpha(j)-1, n}\big(A_{S_1(j)}(j), d_2, z_{1}^{<\alpha(j)}(j)\big)$ denoted by $\gadgetIf\big(z_{1}^{<\alpha(j)}(j)\big)$. 
    \item If $z_{1}^{>\alpha(j)}(j) \in Z(j)$, add $\gadgetIf_{n-(\alpha(j)+1),n}\big(A_{\overline{S_1(j)}}(j), d_2, z_{1}^{>\alpha(j)}(j)\big)$ denoted by $\gadgetIf\big(z_{1}^{>\alpha(j)}(j)\big)$. 
    \item If $z_{2}^{<\beta(j)}(j) \in Z(j)$, add $\gadgetIf_{\beta(j)-1,n}\big(A_{S_2(j)}(j), d_2, z_{2}^{<\beta(j)}(j)\big)$ denoted by $\gadgetIf\big(z_{2}^{<\beta(j)}(j)\big)$.
    \item If $z_{2}^{>\beta(j)}(j) \in Z(j)$, add $\gadgetIf_{n-(\beta(j)+1),n}\big(A_{\overline{S_2(j)}}(j), d_2, z_{2}^{>\beta(j)}(j)\big)$ denoted by $\gadgetIf\big(z_{2}^{>\beta(j)}(j)\big)$.  
\end{itemize}
To understand the function of these gadgets, imagine that the vertex $z_{1}^{<\alpha(j)}(j)$ is on the same side as $d_2$.
Then the first gadget ensures that at most $\alpha(j)-1$ vertices of $A_{S_1(j)}(j)$ are placed on the different side.
In particular, this number is not equal to $\alpha(j)$ and the vertex $v_j$ is not selected as desired.
Further, if the vertex $z_{1}^{>\alpha(j)}(j)$ is on the same side as $d_2$, then the second gadget ensures that at most $n-(\alpha(j)+1)$ vertices of $A_{\overline{S_1(j)}}(j)$ are on the different side.
And because the partitions of $A_{S_1(j)}(j)$ and $A_{\overline{S_1(j)}}(j)$ are ``complementary'' due to~$F'$-gadgets, at most $\alpha(j)+1$ vertices of $A_{S_1(j)}(j)$ are on the different side of $d_2$, i.e.,~$v_j$ is, again, not selected.
The last two gadgets have an analogous function for the endpoint $w_j$.

We define the following upper bound on the number of edges crossed by any partition of these new $\gadgetIf$-gadgets:
\begin{align}
    \budget(j) = &\mcut\big(\gadgetIf_{\abs*{Z(j)} - 1,\abs*{Z(j)}}\big)  \nonumber \\
    +&\chi\Bigl(z_{1}^{<\alpha(j)}(j) \in Z(j)\Bigr) \cdot \mcut\big(\gadgetIf_{\alpha(j)-1,n}\big)  \nonumber \\
    +&\chi\Bigl(z_{1}^{>\alpha(j)}(j) \in Z(j)\Bigr) \cdot \mcut\big(\gadgetIf_{n-(\alpha(j)+1),n}\big)  \nonumber\\
    +&\chi\Bigl(z_{2}^{<\beta(j)}(j) \in Z(j)\Bigr) \cdot \mcut\big(\gadgetIf_{\beta(j)-1,n}\big)  \nonumber \\
    +&\chi\Bigl(z_{2}^{>\beta(j)}(j) \in Z(j)\Bigr) \cdot \mcut\big(\gadgetIf_{n-(\beta(j)+1),n}\big) \label{eq:def-budget-j}
\end{align}
where $\chi$ is the characteristic function defined as 
\[
    \chi\Bigl(x \in X\Bigr)
=
\begin{cases}
    1 & \text{if } x \in X \\
    0 & \text{otherwise}.
\end{cases}
\]
We will show that every optimal partition of the final graph $G^*$ needs to cross exactly $\budget(j)$ edges from these new $\gadgetIf$-gadgets.

Finally, for every $j \in [m-1]$, every $i \in [n]$, and every $S \in \cS \cup \tildecS$, we add the $F'$-gadget between $b_S^i(j)$ and $a_S^i(j+1)$.
This concludes the construction of the graph $G^*$. In the rest of this section, when we refer to a gadget of $G^*$, we ignore the $F'$-gadgets of $G^*$ that are included in $\gadgetIf$-gadgets. This will allow us to state that an optimal solution must satisfy all gadgets of $G^*$. 

We define the value $D$ used in the definition of the gadgets as
\[
    D = m \cdot \Bigl(4 k' \cdot {n \choose 2} + 2k' \cdot (2k'- 1) \cdot n^2\Bigr) > 1.
\]
Observe that the value $D$ is the number of edges of $G^*$ not contained in any gadget and those are precisely the edges in the set $E_{G^*}[A \cup B]$. 
First, there are cliques of form $A_S(j)$ and $B_S(j)$ with~$S \in \cS \cup \tildecS$ and $j \in [m]$.
And second, there is a complete bipartite graph between $A_S(j)$ and $B_{T}(j)$ for every $S,T \in \cS \cup \tildecS$ with $T \neq \overline{S}$. 
So we indeed have
\begin{equation}\label{eq:value-d}
    D = \abs*{E_{G^*}(A, B)}.
\end{equation}
Also note that the values $D$ and $C$ (cf.\ \eqref{eq:def-of-c}) are polynomial in $n' = n + 1$ and $k'$.
Later we will rely on this to bound the running time required for the construction of $G^*$.

We define the target value $\budget$ of our \textsc{Max Cut} instance as 
\begin{equation}\label{eq:def-budget}
    \budget = N \cdot \mcut(F') + \mcut(F) +\bigl(\sum_{j \in [m]} \budget(j)\bigr) + m \cdot L 
\end{equation}
where
\begin{align*}
    & L_1 = k' \cdot (2k'-2) \cdot n^2 && L_2 = 2k' \cdot n^2 \\ 
    & L = L_1 + L_2 && N = 1 + m \cdot k' \cdot n + (m-1) \cdot 2k' \cdot n.
\end{align*}
Observe that $N$ is precisely the number of $F'$-gadgets in $G^*$ not contained in larger gadgets: one is between $d_1$ and $d_2$, further $m \cdot k' \cdot n$ are between pairs of form $\big(a_S^i(j), a_{\overline{S}}^i(j)\big)$, and~$(m-1) \cdot 2k' \cdot n$ more are between pairs of form $\big(b_S^i(j), a_S^i(j+1)\big)$.
The definition of $\budget$ reflects 
that we aim at partitioning every gadget (not contained in a larger gadget) optimally and crossing $m \cdot L$ more edges.
As we show in the next lemma, for every $j \in [m]$, the value $L$ is the maximum number of edges crossed in the subgraph $G^*\big[A(j) \cup B(j)\big]$ by a hypothetical partition $(V_1, V_2)$ that satisfies all the gadgets of $G^*$. 
For this, we show show that $m \cdot L_1$ edges are automatically crossed once the $F'$-gadgets are satisfied, and crossing $m \cdot L_2$ further edges is equivalent to partitioning every clique $A_S(j) \cup B_S(j)$ optimally:
\begin{lemma}\label{lem:max-cut-L-edges}
    Let $j \in [m]$ and let $(V_1, V_2)$ be a partition of $G^*\big[A(j) \cup B(j)\big]$ satisfying all~$F'$-gadgets with both endpoints in $A(j)$. 
    And for every $S \in \cS$, let~$t_S = \abs*{A_S(j) \cap V_1}$.
    Then the following properties hold:
    \begin{enumerate}
        \item We have $\abs*{E_{G^*}(V_1, V_2)} \leq L$,
        \item The equality $\abs*{E_{G^*}(V_1, V_2)} = L$ applies if and only if for every set $S \in \cS$ we have $\abs*{B_S(j) \cap V_1} = n-t_S$ and $\abs*{B_{\overline{S}}(j) \cap V_1} = t_S$.
    \end{enumerate}
\end{lemma}
\begin{proof}
    In this proof, by ``all $F'$-gadgets'' we refer to all $F'$-gadgets with both endpoints in~$A(j)$.
    Because $A_S(j)$ consists of $n$ vertices, we have $\abs*{A_S(j) \cap V_2} = n-t_S$ for every $S \in \cS$.
    There is an $F'$-gadget between $a_S^i(j)$ and $a_{\overline{S}}^i(j)$ for every $i \in [n]$.
    An $F'$-gadget forces its endpoints to be on different sides of the partition (cf.\ \cref{lem:F-gadgets}).
    As all $F'$-gadgets are satisfied, we have $\abs*{A_{\overline{S}}(j) \cap V_2} = t_S$ and $\abs*{A_{\overline{S}}(j) \cap V_1} = n-t_S$.
    
    We define the edge sets
    \begin{align*}
        E_1 = &\big\{a_S^i(j) b_T^\ell(j) \mid i, \ell \in [n], S \in \cS, T \in (\cS \cup \tildecS) \setminus \{S, \overline{S}\}\big\} \\
        \cup&\big\{a_{\overline S}^i(j) b_T^\ell(j) \mid i, \ell \in [n], S \in \cS, T \in (\cS \cup \tildecS) \setminus \{S, \overline{S}\}\big\}
    \end{align*}
    and
    \[
        E_2 = \bigcup_{S \in \cS \cup \tildecS} \big\{uv \mid u, v \in A_S(j) \cup B_S(j), u \neq v\big\}.
    \]
    We claim that the sets $E_1$ and $E_2$ partition the set of edges of $G^*\big[A(j) \cup B(j)\big]$.
    This follows from the following properties.
    Let $S \in \cS \cup \tildecS$ be fixed.
    The vertices in $A_S(j)$ induce a clique and are non-adjacent to vertices in $A(j) \setminus A_S(j)$.
    Further, there are no edges between~$A_S(j)$ and~$B_{\overline{S}}(j)$.
    And the set $A_S(j) \cup B_{S}(j)$ forms a clique.
    Finally, for every set $T \in \cS \cup \tildecS$ with~$T \notin \{S, \overline{S}\}$, there is a complete bipartite graph between $A_S(j)$ and $B_T(j)$.
    So now to analyze the size of~$E_{G^*}(V_1, V_2)$ depending on $(V_1, V_2)$, we consider the edges in $E_1$ and $E_2$ separately.

    We claim that the partition $(V_1, V_2)$ crosses exactly $L_1$ edges of $E_1$.
    Consider arbitrary sets $S \in \cS$ and $T \in (\cS \cup \tildecS) \setminus \{S, \overline{S}\}$ and integers $i, \ell \in [n]$.
    Since $(V_1, V_2)$ satisfies $F'$-gadgets we have 
    \[
        \abs*{\{a_S^i(j), a_{\overline S}^i(j)\} \cap V_1} = \abs*{\{a_S^i(j), a_{\overline S}^i(j)\} \cap V_2} = 1.
    \]
    Therefore, precisely one of the edges $a_S^i(j) b_T^\ell(j)$ and $a_{\overline{S}}^i(j) b_T^\ell(j)$ belongs to $E_{G^*}(V_1, V_2)$, no matter what the side of the vertex $b_T^\ell(j)$ is. 
    With $\abs*{\cS} = \abs*{\tildecS} = k'$, we get 
    \[
        \abs*{E_{G^*}(V_1, V_2) \cap E_1} = k' \cdot (2k'-2) \cdot n^2 = L_1.
    \]

    We also claim that $(V_1, V_2)$ crosses at most $L_2$ edges of $E_2$.
    The subgraph~$G^*[E_2]$ is a vertex-disjoint union of $\abs*{\cS \cup \tildecS} = 2k'$ cliques on $2n$ vertices each.
    For a clique of size $2n$, optimal partitions are precisely those which put exactly $n$ vertices on each side, and such partitions cross $n^2$ edges of the clique.
    Thus, we have $\abs*{E_{G^*}(V_1, V_2) \cap E_2} \leq 2k' \cdot n^2 = L_2$ and that implies the first claim.
    
    Moreover, we have $\abs*{E_{G^*}(V_1, V_2)} = L$ if and only if every such clique is crossed optimally, i.e., for every $S \in \cS$, the equalities 
    \begin{align*}
        \abs*{\big(A_S(j) \cup B_S(j)\big) \cap V_1} = n && \text{and} && \abs*{\big(A_{\overline S}(j) \cup B_{\overline S}(j)\big) \cap V_1} = n
    \end{align*}
    are satisfied.
    Recall that $\abs*{A_S(j) \cap V_1} = t_S$ and $\abs*{A_{\overline S}(j) \cap V_1} = n - t_S$ hold.
    We conclude that $\abs*{E_{G^*}(V_1, V_2)} = L$ is equivalent to having $\abs*{B_S(j) \cap V_1} = n-t_S$ and $\abs*{B_{\overline S}(j) \cap V_1} = t_S$ concluding the proof of the second claim.
\end{proof}

\subsection{Correctness of this reduction}
Here we show that the constructed instance of \textsc{Max Cut} is equivalent to the original instance of \multicolored{}.
We start with the simpler direction:
\begin{lemma}
    If $G$ admits a multicolored independent set, then $G^*$ admits a partition $(V_1, V_2)$ with $\abs*{E_{G^*}(V_1, V_2)} \geq \budget$. 
\end{lemma}
\begin{proof}
    Suppose the graph $G = (U_1 \cup \dots \cup U_{k'}, E)$ admits a multicolored independent set~$\left\{u_1^{\gamma(1)}, \dots, u_{k'}^{\gamma(k')}\right\}$ where $\gamma(1), \dots, \gamma(k') \in [n]_0$.
    For every set $S \in \cS$, we define the value~$i_S = \phi^{-1}(S) \in [k']$.
    Then the independent set $U_{i_S}$ of $G$
    corresponds to the set $S$. 
    With this we construct the partition $(V_1, V_2)$ as follows.
    
    First, we put the vertex $d_1$ into $V_1$ and we put the vertices $d_2$ and $d_2'$ into $V_2$.
    This satisfies the two $F$- and $F'$-gadgets between those vertices (cf. \cref{lem:F-gadgets}). 
    So we extend the partition to an optimal partition of those gadgets.
    Then we partition the sets $A$ and $B$ as follows.
    For every $S \in \cS$, every $j \in [m]$, and every $i \in [n]$:     
    \begin{itemize}
        \item put $a_S^i(j)$ into $V_1$ if $i \leq \gamma(i_S)$ and into $V_2$ otherwise,
        \item put $a_{\overline S}^i(j)$ into $V_2$ if $i \leq \gamma(i_S)$ and into $V_1$ otherwise,
        \item put $b_S^i(j)$ into $V_2$ if $i \leq \gamma(i_S)$ and into $V_1$ otherwise,
        \item put $b_{\overline S}^i(j)$ into $V_1$ if $i \leq \gamma(i_S)$ and into $V_2$ otherwise.
    \end{itemize}
    Now every $F'$-gadget with both endpoints in $A$ is satisfied and we extend the partition to an optimal partition of this gadget (cf. \cref{lem:F-gadgets}). 
    Further, we have 
    \begin{align}
        \gamma(i_S) = \abs*{A_S(j) \cap V_1} = \abs*{B_{\overline S}(j) \cap V_1} && \abs*{A_{\overline{S}}(j) \cap V_1} = \abs*{B_S(j) \cap V_1} = n-\gamma(i_S) \label{eq:on-side-v-1}
    \end{align}
    for every $S \in \cS$ and $j \in [m]$.
    By \cref{lem:max-cut-L-edges} we then have 
    \[
        \abs*{E_{G^*}\bigl((A \cup B) \cap V_1, (A \cup B) \cap V_2\bigr)} = m\cdot L.
    \]
    Also we have $b_S^i(j) \in V_2$ if and only if $a_S^i(j+1) \in V_1$ for all $S \in \cS \cup \tildecS$, $j \in [m-1]$, and~$i \in [n]$. 
    So every $F'$-gadget with one endpoint in $A$ and the other in $B$ is satisfied, and we extend the partition to an optimal partition of this gadget (cf. \cref{lem:F-gadgets}).

    Then for every $j \in [m]$ we proceed as follows.
    Recall that $v_j w_j = u_{i_1(j)}^{\alpha(j)} u_{i_2(j)}^{\beta(j)} \in E$ an edge of $G$.
    And we have $S_1(j) = \phi\big(i_1(j)\big) \in \cS$ and~$S_2(j) = \phi\big(i_2(j)\big) \in \cS$ for the sets from~$\cS$ corresponding to the independent sets $U_{i_1(j)}$ and $U_{i_2(j)}$ of $G$ containing the vertices $v_j$ and~$w_j$, respectively.
    Since the set $\left\{u_1^{\gamma(1)}, \dots, u_{k'}^{\gamma(k')}\right\}$ is an independent set in $G^*$, we have 
    \[
        \Big(u_{i_1(j)}^{\gamma(i_1(j))}, u_{i_2(j)}^{\gamma(i_2(j))}\Big) \neq \Big(u_{i_1(j)}^{\alpha(j)}, u_{i_2(j)}^{\beta(j)}\Big),
    \]
    i.e., 
    \[
        \Big(\gamma\big(i_1(j)\big), \gamma\big(i_2(j)\big)\Big) \neq \Big(\alpha(j), \beta(j)\Big),
    \]
    i.e., at least one of the following cases occurs:
    \begin{align*}
         \gamma\big(i_1(j)\big) \leq \alpha(j)-1 && \lor && \gamma\big(i_1(j)\big) \geq \alpha(j)+1 &&\lor \\  \gamma\big(i_2(j)\big) \leq \beta(j)-1 && \lor && \gamma\big(i_2(j)\big) \geq \beta(j)+1 &&.
    \end{align*}
    By \eqref{eq:on-side-v-1}, this is equivalent to
    \begin{align*}
        \abs*{A_{S_1(j)}(j) \cap V_1} \leq \alpha(j)-1 && \lor && \abs*{A_{\overline{S_1(j)}}(j) \cap V_1} \leq n - \big(\alpha(j)+1\big) && \lor \\ 
        \abs*{A_{S_2(j)}(j) \cap V_1} \leq \beta(j)-1 && \lor && \abs*{A_{\overline{S_2(j)}}(j) \cap V_1} \leq n - \big(\beta(j)+1\big) && . && 
    \end{align*}
    So let $z$ be equal to $z_{1}^{<\alpha(j)}(j)$ / $z_{1}^{>\alpha(j)}(j)$ / $z_{2}^{<\beta(j)}(j)$ / $z_{2}^{>\beta(j)}(j)$ in the first / second / third / fourth case---we choose any of those if multiple cases occur.
    Because the first / second / third / fourth case can only occur if $\alpha(j) \neq 0$ / $\alpha(j) \neq n$ / $\beta(j) \neq 0$ / $\beta(j) \neq n$, we have~$z \in Z(j)$.
    We put $z$ into $V_2$ and the vertices in $Z(j) \setminus \{z\}$ into $V_1$.

    Now the following properties apply.
    First, we have $d_2, d_2' \in V_2$ and $\abs*{Z(j) \cap V_1} = \abs*{Z(j)} - 1$ so we extend the partition to an optimal partition of $\gadgetIf(Z(j))$ (cf.\ \cref{lem:H-if-properties}).
    Second, for every $z' \in Z(j) \setminus \{z\}$ we have $d_2 \in V_2$ and $z' \in V_1$ so we extend the partition to an optimal partition of $\gadgetIf(z')$ (cf.\ \cref{lem:H-if-properties}).
    Finally, consider the gadget $\gadgetIf(z)$.
    We have~$d_2 \in V_2$,~$z \in V_2$, and at most $\alpha(j)-1$ / $n - (\alpha(j)+1)$ / $\beta(j)-1$ / $n - (\beta(j)+1)$ vertices of $A_{S_1(j)}(j)$ / $A_{\overline{S_1(j)}}(j)$ / $A_{S_2(j)}(j)$ / $A_{\overline{ S_2(j)}}(j)$ are in $V_1$ depending on the case that occurs.
    By the choice of $z$ and the definition of $\gadgetIf(z)$, we are able to extend the partition to an optimal partition of this gadget (cf.\ \cref{lem:H-if-properties}).
    
    Altogether, $(V_1, V_2)$ satisfies one $F$-gadget, all $F'$-gadgets, and all $\gadgetIf$-gadgets, and we also have $\abs*{E_{G^*}\bigl((A \cup B) \cap V_1, (A \cup B) \cap V_2\bigr)} = m \cdot L$.
    As the edge sets of the corresponding subgraphs are disjoint, we get $\abs*{E_{G^*}(V_1, V_2)} \geq \budget$ as claimed.
\end{proof}

Now we prove the other direction:
\begin{lemma}
    If $G^*$ admits a partition $(V_1, V_2)$ with $\abs*{E_{G^*}(V_1, V_2)} \geq \budget$, then $G$ admits a multicolored independent set. 
\end{lemma}
\begin{proof}
    Without loss of generality assume that $d_1 \in V_1$ holds.
    For every $S \in \cS$ and $j \in [m]$, we define the values 
    \begin{align}
        t_S^j = \abs*{A_S(j) \cap V_1} \in [n]_0 && \text{and} && t_S = t_S^1. \label{eq:def-t-s-j}
    \end{align}
    The function $\phi$ is a bijection so that every integer $i \in [k']$ corresponds to a set $\phi(i)$ from the family $\cS$.
    In the rest of the proof, we argue that
    $\left\{u_1^{t_{\phi(1)}}, \dots, u_{k'}^{t_{\phi(k')}}\right\}$
    is an independent set in the graph~$G$.

    Every edge of $G^*$ belongs to either a gadget of $G^*$ (an $F$-gadget, an $F'$-gadget, or an~$\gadgetIf$-gadget) or to $E_{G^*}[A \cup B]$.
    First, suppose that the partition $(V_1, V_2)$ is suboptimal for some gadget $J$ of $G^*$.
    By the properties of the gadgets we have $\abs*{E_J(V_1, V_2)} \leq \mcut(J) - D^2$ (cf. \cref{lem:F-gadgets,lem:H-if-properties}).
    The value $D$ satisfies $D = \abs*{E_{G^*}[A \cup B]} > 1$ (cf.\ \eqref{eq:value-d}) so we get
    \begin{align*}
        \abs*{E_{G^*}(V_1, V_2)} \leq &\Bigl(\mcut(F) + N \cdot \mcut(F') + \bigl(\sum_{j\in[m]} \budget(j)\bigr) - D^2\Bigr) \\
        +&\abs*{E_{G^*}\big(V_1 \cap (A \cup B), V_2 \cap (A \cup B)\big)} \\
        <&\bigl(\budget - D^2\bigr) + \abs*{E_{G^*}\big(V_1 \cap (A \cup B), V_2 \cap (A \cup B)\big)} \\
        \leq &\bigl(\budget - D^2\bigr) + \abs*{E_{G^*}[A \cup B]} \leq \budget - D^2 + D < \budget
    \end{align*}
    contradicting the cardinality of $E_{G^*}(V_1, V_2)$.
    So the partition of $(V_1, V_2)$ is optimal for each of the above-mentioned gadgets. 
    First, this implies that $d_2, d_2' \in V_2$ holds.
    Second, we have
    \begin{align*}
        &\abs*{E_{G^*}\big(V_1 \cap (A \cup B), V_2 \cap (A \cup B)\big)} \\
        = &\abs*{E_{G^*}(V_1, V_2)} - \Bigl(\mcut(F) + N \cdot \mcut(F') + \bigl(\sum_{j\in[m]} \budget(j)\bigr)\Bigr) \\
        \geq &\budget - \Bigl(\mcut(F) + N \cdot \mcut(F') + \bigl(\sum_{j\in[m]} \budget(j)\bigr)\Bigr) \stackrel{\eqref{eq:def-budget}}{=} m \cdot L. 
    \end{align*}
    \cref{lem:max-cut-L-edges} implies that this value is also upper-bounded by $m \cdot L$, i.e., the equality applies, and moreover, for every set  $S \in \cS$ and every index $j \in [m]$, we have 
    \begin{align}
        \abs*{B_S(j) \cap V_1} = n-\abs*{A_S(j) \cap V_1} = n-t_S^j && \abs*{B_{\overline{S}}(j) \cap V_1} = \abs*{A_S(j) \cap V_1} = t_S^j. \label{eq:a-b-equalities}
    \end{align}
    Because every $F'$-gadget with both endpoints in $A$ is satisfied, we also have 
    \begin{equation}\label{eq:a-s-overline-s}
        \abs*{A_{\overline S}(j) \cap V_1} = n - \abs*{A_S(j) \cap V_1} = n - t_S^j.
    \end{equation}
    With this, we now show that the same selection is made by the partition $(V_1, V_2)$ in every selection gadget. 
    \begin{claim}\label{claim:equal-choices-on-j}
        Let $S \in \cS$, then for all $j \in [m]$ we have $t_S^j = t_S$.
    \end{claim}
    \begin{claimproof}
        By definition we have $t_S = t_S^1$.
        Thus it suffices to show that $t_S^j = t_S^{j+1}$ for every index~$j \in [m-1]$.
        So let $j \in [m-1]$ be arbitrary.
        The $F'$-gadgets with one endpoint in $A$ and the other in $B$ ensure that we have $a_S^i(j) \in V_1$ if and only if $b_S^i(j+1) \in V_2$ for every index~$i \in [n]$.
        Therefore, we have
        \[
            t_S^j = \abs*{A_S(j) \cap V_1} = \abs*{B_S(j+1) \cap V_2} = n - \abs*{B_S(j+1) \cap V_1} \stackrel{\eqref{eq:a-b-equalities}}{=} n-\big(n-t^{j+1}_S\big) = t^{j+1}_S
        \]
        and the claim follows.
    \end{claimproof}
    
    Now suppose that the set $\left\{u_1^{t_{\phi(1)}}, \dots, u_{k'}^{t_{\phi(k')}}\right\}$ is not an independent set in $G$.
    Then there exist indices $i_1 \neq i_2 \in [k']$ such that the vertices $u_{i_1}^{t_{\phi(i_1)}}$ and $u_{i_2}^{t_{\phi(i_2)}}$ are adjacent, i.e., we have~$u_{i_1}^{t_{\phi(i_1)}} u_{i_2}^{t_{\phi(i_2)}} \in E(G)$. 
    So let $j \in [m]$ be such that $\left\{u_{i_1}^{t_{\phi(i_1)}}, u_{i_2}^{t_{\phi(i_2)}}\right\} = \{v_j, w_j\}$ is the~$j$-th edge of $G$.
    Without loss of generality, we assume that $u_{i_1}^{t_{\phi(i_1)}} = v_j$ and $u_{i_2}^{t_{\phi(i_2)}} = w_j$ apply.
    Then the sets $S_1(j) = \phi(i_1)$ and $S_2(j) = \phi(i_2)$ from $\cS$ correspond to the independent sets~$U_{i_1}$ and $U_{i_2}$ of $G$, respectively. 
    The partition $(V_1, V_2)$ satisfies the gadget $\gadgetIf\big(Z(j)\big)$ and we have~$d_2, d_2' \in V_2$. 
    Thus, we have $\abs*{Z(j) \cap V_1} \leq \abs*{Z(j)} - 1$ (cf.\ \cref{lem:H-if-properties}), i.e., there exists a vertex $z \in Z(j) \cap V_2$.
    The gadget $\gadgetIf(z)$ is also satisfied by $(V_1, V_2)$ and we have $z \in V_2$ and~$d_2 \in V_2$ so one of the following cases applies by definition of $\gadgetIf(z)$ (cf.\ \cref{lem:H-if-properties}): 
    \begin{align*}
        \abs*{A_{S_1(j)}(j) \cap V_1} \leq \alpha(j) - 1 && \lor && \abs*{A_{\overline{S_1(j)}}(j) \cap V_1} \leq n - (\alpha(j) + 1) && \lor \\
        \abs*{A_{S_2(j)}(j) \cap V_1} \leq \beta(j) - 1 && \lor && \abs*{A_{\overline{S_2(j)}}(j) \cap V_1} \leq n - (\beta(j) + 1) &&
    \end{align*}
    depending on the vertex $z \in Z(j)$.
    By definition of $t^j_S$     together with equality~\eqref{eq:a-s-overline-s}, this is equivalent to
    \begin{align*}
        t^j_{S_1(j)} \leq t_{S_1(j)} - 1 && \lor && n-t^j_{S_1(j)} \leq n - (t_{S_1(j)} + 1) && \lor \\
        t^j_{S_2(j)} \leq t_{S_2(j)} - 1 && \lor && n-t^j_{S_2(j)} \leq n - (t_{S_2}(j) + 1) &&.
    \end{align*}
    By \cref{claim:equal-choices-on-j} we have $t_{S_1} = t_{S_1}^j$ and $t_{S_2} = t_{S_2}^j$ so this is equivalent to
    \begin{align*}
        t_{S_1} \leq t_{S_1} - 1 && \lor && t_{S_1} \geq t_{S_1} + 1 && \lor \\
        t_{S_2} \leq t_{S_2} - 1 && \lor && t_{S_2} \geq t_{S_2} + 1 &&.
    \end{align*}
    But this is a false expression.
    So $\{u_1^{t_{\phi(1)}}, \dots, u_{k'}^{t_{\phi(k')}}\}$ is indeed a multicolored independent set in $G$.
\end{proof}

\subsection{Construction of a linear multi-expression}

\subparagraph*{Sketch} The main idea behind the construction can be sketched as follows (see \cref{fig:max-cut-expression} for an illustration).
There is an iteration for every $j = 1, \dots, m+1$. 
In iteration $j$ we create all vertices of $A(j)$ and $B(j-1)$ (up to corner cases $j=1$ and $j=m+1$ where~$B(0)$ and~$A(m+1)$ do not exist), whose ordering is specified below.
During this process a vertex~$a_S^i(j)$ (resp.~$b_S^i(j-1)$) gets labels $\left\{q^{\la} \mid q \in S\right\}$ (resp.\ $\left\{q^{\lb} \mid q \in S\right\}$). 
And we maintain the invariant that the vertices of form $a_{T}(j-1)$ have label set $\left\{q^{\la}_{\lold} \mid q \in T\right\}$.
Thus, in the end of the $j$-th iteration we can correctly create the edges between $A(j-1)$ and~$B(j-1)$ by joining $q^{\la}_{\lold}$ with $q^{\lb}$ for every $q \in [k]$.
Then we forget $q^{\la}_{\lold}$ and $q^{\lb}$, and relabel~$q^{\la}$ to $q^{\la}_{\lold}$ so that the invariant holds after iteration $j$.

\begin{figure}[t]
    \centering
    \includegraphics[width=\linewidth]{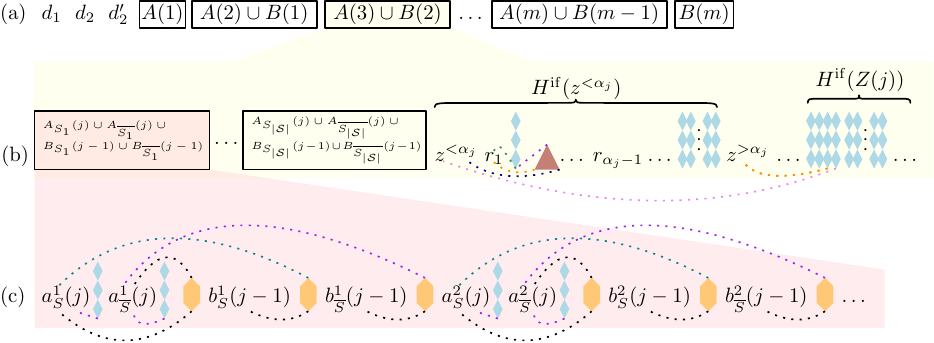}
    \caption{Order of the introduce-nodes in the constructed linear multi-expression. Gadget connections are sketched with dotted lines. Three blue diamonds depict a column of an $\gadgetIf$-gadget. An orange hexagon is an $F'$-gadget and a red triangle is a $T$-gadget. (a) High-level overview. (b) Order for fixed $j \in [m-1]$. (c) Order for fixed $j \in [m-1]$ and fixed $S \in \cS$.}
    \label{fig:max-cut-expression}
\end{figure}

In iteration $j$, we further iterate over the sets $S \in \cS$ and then for a fixed set $S$, we iterate over indices $i \in [n]$.
We put the vertices $a_S^i(j)$, $a_{\overline S}^i(j)$, $b_S(j-1)$, and $b_{\overline S}(j-1)$ almost consecutively in the expression (up to some gadget parts attached to these vertices). 
This ensures that for fixed $S$ and $j$, the vertices of $A_S(j) \cup B_S(j-1) \cup A_{\overline S}(j) \cup B_{\overline S}(j-1)$ are consecutive and a constant number of labels suffices to make each of~$A_S(j)$,~$B_S(j-1)$,~$A_{\overline S}(j)$,~and~$B_{\overline S}(j-1)$ a clique.
Namely, in iteration $i$, we use the \emph{current} labels~$\la$,~$\lb$,~$\labar$,~and~$\lbbar$ for vertices $a_S^i(j)$, $a_{\overline S}^i(j)$, $b^i_S(j-1)$, and $b_{\overline S}^i(j-1)$, respectively, and we ensure that the vertices~$a_S^{i'}(j)$, $a_{\overline S}^{i'}(j)$,~$b^{i'}_S(j-1)$,~and~$b_{\overline S}^{i'}(j-1)$ satisfying $i' < i$, contain \emph{old} labels~$\laold$,~$\lbold$,~$\labarold$,~and~$\lbbarold$.
Then joining each current label with the corresponding old label and then relabeling current labels to the old labels creates the desired cliques.

Our construction ensures that for most $F$- and $F'$-gadgets the endpoints are close to each other so that we can reuse the same labels $\labelf, \labelfprimeleft, \labelfprimeright, \labelfprimeleftedge$, and $\labelfprimerightedge$ to create those gadgets.
For the remaining $F$- and $F'$-gadgets a small number of ``private'' labels suffices.
Similarly, we can construct $\gadgetIf$-gadgets attached to each copy. 
For this we create every column of an $\gadgetIf$-gadget close to the vertex it is attached to.
Every column is basically a path of~$F$-gadgets so it can be created locally using one additional label $\lp$.
At most five~$\gadgetIf$-gadgets are attached to each copy.
So we use at most five current and five old labels~$\lone$,~$\ltwo$,~$\lthree$,~$\lfour$,~$\lfive$ and~$\loneold$,~$\ltwoold$,~$\lthreeold$,~$\lfourold$,~$\lfiveold$, respectively, to create the complete $2n$-partite graph on the columns of one gadget.
Finally, for fixed $j$, the elements of $Z(j)$ use four labels~$\lzone$,~$\lztwo$,~$\lzthree$, and~$\lzfour$.
Now we formalize these ideas.

\begin{lemma}
    The instance $(G^*, \budget)$ and a linear multi-expression of $G^*$ using $3k +\cO(1)$ labels can be computed from the instance $G = (U_1 \cup \dots \cup U_{k'}, E)$ in time polynomial in $n'$ and $k'$.
\end{lemma}
\begin{proof}
    Since we construct a linear expression it suffices to describe the order of operations.
    An introduce-node is always implicitly followed by a union-node.
    For two labels $\ell_1$ and $\ell_2$, by \emph{relabeling $\ell_1$ to $\ell_2$} we naturally refer to the operation $\rho_{\ell_1 \to \{\ell_2\}}$.
    
    The set of the labels is
    \begin{align*}
        &\big\{q^{\la}_{\lold}, q^{\la}, q^{\lb} \mid q \in [k]\big\} \cup \big\{\lwone, \lwtwo, \lwtwoprime, \labelf, \labelfprimeleft, \labelfprimeright, \labelfprimeleftedge, \labelfprimerightedge\big\} \cup \\ &\big\{\la, \labar, \lb, \lbbar, \laold, \labarold, \lbold, \lbbarold, \lp, \lr, \lone, \loneold, \ltwo, \ltwoold, \lthree, \lthreeold, \lfour, \lfourold, \lfive, \lfiveold, \lzone, \lztwo, \lzthree, \lzfour\big\},
    \end{align*}
    i.e., $3k + \cO(1)$ labels are used.
    
    As the graph $G^*$ contains multiple gadgets, we first explain how some gadgets are added to the expression.
    First, by \emph{adding the gadget vertices of $F(u, v)$} we simply mean introducing a vertex with label set $\{\labelf\}$ $C$ times.
    The label $\labelf$ can later be used to create the missing edges to the endpoints $u$ and $v$ and then be forgotten.
    Second, \emph{adding the gadget vertices of~$F'(u, v)$} means repeating the following sequence of operations $C$ times:
    \begin{enumerate}
        \item introduce a vertex with label set $\big\{\labelfprimeleft, \labelfprimeleftedge\big\}$,
        \item introduce a vertex with label set $\big\{\labelfprimeright, \labelfprimerightedge\big\}$,
        \item join labels $\labelfprimeleftedge$ and $\labelfprimerightedge$,
        \item forget label $\labelfprimeleftedge$,
        \item forget label $\labelfprimerightedge$.
    \end{enumerate}
    Observe that as a result of these operations, a matching on $C$ edges was created such that in every edge, one endpoint has label set $\big\{\labelfprimeleft\big\}$ and the other has label set $\big\{\labelfprimeright\big\}$. 
    These two labels can later be used to create the missing edges to the endpoints $u$ and $v$ of the gadget and then be forgotten. 

    Now we construct the desired linear multi-expression as follows (see \cref{fig:max-cut-expression} for an illustration).
    We introduce the vertices $d_1$, $d_2$, and $d_2'$ with label sets $\{\lwone\}$, $\{\lwtwo\}$, and $\{\lwtwoprime\}$, respectively.
    After that, we add the gadget vertices of $F'(d_1, d_2)$, then join $\lwone$ with $\labelfprimeleft$, forget~$\labelfprimeleft$, join $\lwtwo$ with $\labelfprimeright$, and forget $\labelfprimeright$.
    Then we add the gadget vertices of $F(d_2', d_2)$, join~$\labelf$ with $\lwtwoprime$, join~$\labelf$ with $\lwtwo$, and forget $\labelf$.
    
    To simplify the description for the corner cases $j = 1$ and $m+1$, we stick to the following convention: whenever in the following description, a non-existing vertex is referred to (for instance, when vertices of $B(0)$ and $A(m+1)$ and gadgets attached to them are introduced), we simply skip the corresponding operation.
    Wit this, for every index~$j = 1, \dots, m+1$, every set~$S \in \cS$, and every index~$i \in [n]$, i.e., using three nested loops with outermost being~$j$ and innermost being~$i$, we repeat the following operations.
    
    Introduce the vertex $a_S^i(j)$ with the label set $\big\{q^{\la} \mid q \in S\big\} \cup \{\la\}$.
    If $S = S_1(j)$ and~$\alpha(j) \neq 0$ apply, we proceed as follows to create the gadget~$\gadgetIf\big(z^{<\alpha(j)}_{1}(j)\big)$ (we refer to \cref{fig:h-if-gadget} for an illustration of an $\gadgetIf$-gadget).
    Here the vertices (e.g., column $p_{i,q}$) refer to those vertices of the gadget $\gadgetIf\big(z^{<\alpha(j)}_{1}(j)\big)$.
    \begin{enumerate}
        \item\label{item:column-one-f-gadget} Repeat the following sequence of operations~$D$ times: introduce a vertex with label set~$\{\lp, \lone\}$, join $\lp$ and $\labelf$, forget $\labelf$, add gadget vertices of $F$, join $\lp$ and $\labelf$, forget $\lp$.  
        
        One such repetition creates the next vertex $p_{i,q}$ of the column $P_{i}$ and creates the edges from~$p_{i,q}$ to the gadget vertices of $F(p_{i,q-1}, p_{i,q})$.
        If $q = 1$, this operation creates no edges, and otherwise the gadget vertices of $F(p_{i,q-1}, p_{i,q})$ are precisely the vertices currently containing the label $\labelf$. 
        Then it forgets $\labelf$ so that it can be reused, introduces the gadget vertices of $F(p_{i,q}, p_{i,q+1})$ (or of~$F\big(p_{i,q}, a_S^i(j)\big)$ if $p_{i,q}$ is the last vertex of the column, i.e., if~$q = D$), and finally creates the edges from $p_{i,q}$ to those vertices.
        And it forgets the label $\lp$ because $p_{i,q}$ now has all its incident edges.
        \item\label{item:column-attachment} Join $\la$ and $\labelf$ and forget $\labelf$.
        This step creates the missing edges of the gadget $F(p_{i,D}, a_S^i(j))$. 
        \item Join $\loneold$ and $\lone$. 
        This creates the edges from $P_{i}$ to all columns of the same $\gadgetIf$-gadget created so far.
        In the end, this ensures that the columns of this gadget induce a complete~$2n$-partite graph.
        \item\label{item:relabel-column-to-old} Relabel $\lone$ to $\loneold$. This way we will be able to join further columns of the same $\gadgetIf$-gadget with this one.
    \end{enumerate}
    And if $S = S_2^j$ and $\beta(j) \neq 0$, we proceed analogously to create the gadget~$\gadgetIf\big(z^{<\beta(j)}_{2}(j)\big)$ with the only difference that label $\lone$ becomes $\ltwo$ and label $\loneold$ becomes $\ltwoold$. 

    Similarly, we introduce the vertex $a_{\overline S}^i(j)$ with the label set $\big\{q^{\la} \mid q \in \overline S\big\} \cup \{\labar\}$.
    If one of the gadgets $\gadgetIf\big(z^{>\alpha(j)}_{1}(j)\big)$ and~$\gadgetIf\big(z^{>\beta(j)}_{2}(j)\big)$ touches the vertices in the set~$A_{\overline{S}}(j)$, to attach the corresponding column to $a_{\overline S}^i(j)$, we proceed as in \cref{item:column-one-f-gadget}-\cref{item:relabel-column-to-old} with the following two differences.
    First, in step \cref{item:column-attachment}, we join with $\labar$ instead of $\la$.
    And second, instead of label~$\lone$ and~$\loneold$, we use labels~$\lthree$ and~$\lthreeold$ or~$\lfour$ and~$\lfourold$, respectively, depending on whether~$S = S_1(j)$ or $S = S_2(j)$ applies.
    After that, we add the gadget vertices of $F'\big(a_S^i(j), a_{\overline S}^i(j)\big)$, join $\la$ and~$\labelfprimeleft$, forget~$\labelfprimeleft$, join $\labar$ and $\labelfprimeright$, and forget~$\labelfprimeright$.
    
    No $\gadgetIf$-gadgets are attached to the vertices $b_S^i(j-1)$ and $b_{\overline{S}}^i(j-1)$ so their creation is simpler, namely as follows.
    Introduce the vertex $b_S^i(j-1)$ with label set $\big\{q^{\lb} \mid q \in S\big\} \cup \{\lb\}$, add the gadget vertices of $F'\big(b_S^i(j-1), a_S^i(j)\big)$, join $\lb$ and $\labelfprimeleft$, forget $\labelfprimeleft$, join $\la$ and $\labelfprimeright$, and forget $\labelfprimeright$.
    Then analogously proceed for $b_{\overline S}^i(j-1)$ with the only difference that $\lb$ becomes $\lbbar$ and $\la$ becomes $\labar$.
    After that we:
    \begin{enumerate}
        \item join $\la$ and $\laold$, relabel $\la$ to $\laold$,
        \item join $\labar$ and $\labarold$, relabel $\labar$ to $\labarold$,
        \item join $\lb$ and $\lbold$, relabel $\lb$ to $\lbold$,
        \item join $\lbbar$ and $\lbbarold$, relabel $\lbbar$ to $\lbbarold$.
    \end{enumerate}
    Here joining $\la$ and $\laold$ creates the edges from $a^i_S(j)$ to every vertex $a^{i'}_S(j)$ with $i' \in [i-1]$ so that $A_S(j)$ becomes a clique in the end.
    The remaining three operations have an analogous purpose for $A_{\overline{S}}(j)$, $B_S(j)$, and $B_{\overline{S}}(j)$.
    If $i = n$ applies, i.e., we added all vertices of~$A_S(j)$,~$A_{\overline{S}}(j)$,~$B_S(j)$, and~$B_{\overline{S}}(j)$ and created the corresponding cliques, we forget each of~$\laold$,~$\labarold$,~$\lbold$,~$\lbbarold$.
    
    The remaining operations are only executed if $i = n$ and $S$ is the \emph{last} element of $\cS$, i.e., we are in the last iteration of the loop for the sets~$S \in \cS$ and in the last iteration of the loop for the indices~$i \in [n]$.
    For every index $q \in [k]$, we join $q^{\la}_{\lold}$ with $q^{\lb}$, forget label $q^{\la}_{\lold}$, forget label $q^{\lb}$, and then relabel $q^{\la}$ to $q^{\la}_{\lold}$. 
    We recall that two vertices of form $a_S^i(j-1)$ and $b_T^{i'}(j-1)$ are adjacent if and only if the sets $S$ and $T$ share some element $q \in [k]$.
    So this step correctly creates the edges between $A(j-1)$ and $B(j-1)$.
    The relabeling ensures that now only the vertices of $A(j)$ contain the labels from $\big\{q^{\la}_{\lold} \mid q \in [k]\big\}$ and in the next iteration the edges between the sets $A(j)$ and $B(j)$ will be created.

    Now we finish the construction of the $\gadgetIf$-gadgets touching the $j$-th copy.
    For each of the gadgets of form $\gadgetIf(z)$ with $z \in Z(j)$, so far we created the first $n$ columns and attached them to the correct vertices.
    So it remains to add the missing $n$ columns to each of them and correctly attach those columns as well as to construct the gadget $\gadgetIf\big(Z(j)\big)$.

    If $\alpha(j) \neq 0$, we proceed as follows.
    We introduce the vertex $z^{<\alpha(j)}_{1}(j)$ with label set $\{\lzone\}$.
    For every index~$t \in \big[n-\alpha(j)\big]$, repeat the following sequence of operations.
    Here the vertices (e.g.,~$r_t$) refer to those vertices of the gadget $\gadgetIf\big(z^{<\alpha(j)}_{1}(j)\big)$.
    \begin{enumerate}
        \item Introduce the vertex $r_t$ with label set $\{\lr\}$.
        \item Repeat $D-1$ times: introduce a vertex with label set $\{\lp, \lone\}$, join $\lp$ and $\labelf$, forget $\labelf$, add gadget vertices of $F$, join $\lp$ and $\labelf$, forget $\lp$. Similarly as before, this step essentially creates the column $P_{n+t}$ together with the $F$-gadgets between its vertices.
        The main difference is that this time instead of an $F$-gadget, a $T$-gadget is attached to the vertex~$p_{n+t, D}$.
        For this reason, this step is only repeated $D-1$ times and in the next steps, we create the vertex $p_{n+t, D}$ and the $T$-gadget attached to it.
        \item Introduce the vertex $p_{n+t, D}$ with label set $\{\lp, \lone\}$, join $\lp$ and $\labelf$, forget $\labelf$. 
        The next three steps create the gadget $T\big(p_{n+t,D}, r_t, z^{<\alpha(j)}_{1}(j)\big)$.
        \item Add gadget vertices of $F'(p_{n+t, D}, r_t)$, join $\lp$ and $\labelfprimeleft$, forget $\labelfprimeleft$, join $\lr$ and $\labelfprimeright$, forget $\labelfprimeright$.
        \item Add gadget vertices of $F'\big(p_{n+t, D}, z^{<\alpha(j)}_{1}(j)\big)$, join $\lzone$ and $\labelfprimeleft$, forget $\labelfprimeleft$, join $\lp$ and $\labelfprimeright$, forget $\labelfprimeright$, forget $\lp$.
        \item Add gadget vertices of $F'\big(r_t,z^{<\alpha(j)}_{1}(j)\big)$, join $\lr$ and $\labelfprimeleft$, forget $\labelfprimeleft$, join $\lzone$ and $\labelfprimeright$,  forget~$\labelfprimeright$.
        \item Add gadget vertices of $F(r_t, d_2)$, join $\lr$ and $\labelf$, join $\labelf$ and $\lwtwo$, forget $\labelf$, forget $\lr$.
        \item Join $\lone$ and $\loneold$, relabel $\lone$ to $\loneold$. This step is, as before, to create the complete $2n$-partite graph on the columns.
    \end{enumerate}
    Note that forgetting labels $\lp$ and $\lr$ in each repetition allows reusing them for the next column.
    Now the gadget $\gadgetIf\big(z^{<\alpha(j)}_{1}(j)\big)$ is still missing $\alpha(j)$ further columns each attached to no vertex, we create them by repeating $\alpha(j)$ times the following sequence of operations.
    \begin{enumerate}
        \item Repeat $D-1$ times: introduce a vertex with label set $\{\lp, \lone\}$, join $\lp$ and $\labelf$, forget $\labelf$, introduce gadget vertices of $F$, join $\lp$ and $\labelf$, forget $\lp$. 
        \item Introduce a vertex with label set $\{\lp, \lone\}$, join $\lp$ and $\labelf$, forget $\labelf$, forget $\lp$.
        \item Join $\lone$ and $\loneold$, relabel $\lone$ to $\loneold$.
    \end{enumerate}
    The subgraph $H$ of $\gadgetIf(z^{<\alpha(j)}_{1}(j))$ is finished now so we forget label $\loneold$.
    
    Next we create the column of $\gadgetIf(Z(j))$ attached to $z^{<\alpha(j)}_{1}(j)$ using the following operations.
    Repeat $D$ times: introduce a vertex with label set $\{\lp, \lfive\}$, join $\lp$ and $\labelf$, forget~$\labelf$, introduce gadget vertices of $F$, join $\lp$ and $\labelf$, forget $\lp$. 
    Then join $\labelf$ and $\lzone$, forget $\labelf$, and forget $\lzone$ because all adjacencies of $z^{<\alpha(j)}_{1}(j)$ are already created.
    After that, join $\lfive$ and $\lfiveold$ and relabel $\lfive$ to $\lfiveold$.

    Now in case $\alpha(j) \neq n$ / $\beta(j) \neq 0$ / $\beta(j) \neq n$, we repeat the above process to create the gadget $\gadgetIf(z)$ for each of at most three remaining vertices $z \in Z(j) \setminus \big\{z^{<\alpha(j)}_{1}(j)\big\}$. 
    The only difference is that we use labels $\ltwo, \ltwoold, \lztwo$ / $\lthree, \lthreeold, \lzthree$ / $\lfour, \lfourold, \lzfour$, respectively, instead of~ $\lone$,~$\loneold$,~$\lzone$.
    Now it only remains to finish the construction of the gadget $\gadgetIf\big(Z(j)\big)$.
    
    For this, we first, we create $2n-\big(2 \cdot \abs*{Z(j)}-1\big)$ missing columns of this gadget not attached to any vertex by repeating the following sequence $2n-\big(2 \cdot \abs*{Z(j)}-1\big)$ times.
    Repeat~$D-1$ times: introduce a vertex with label set $\{\lp, \lfive\}$, join $\lp$ and $\labelf$, forget $\labelf$, introduce gadget vertices of $F$, join $\lp$ and $\labelf$, forget $\lp$. 
    Then introduce a vertex with label set $\{\lp, \lfive\}$, join $\lp$ and $\labelf$, forget $\labelf$, forget $\lp$.
    Join $\lfive$ and $\lfiveold$, relabel $\lfive$ to $\lfiveold$.

    And finally, we create the last $n - \big(\abs*{Z(j)} - 1\big)$ columns of this gadget to which $T$-gadgets are attached similarly to as described before for the gadget $\gadgetIf\big(z^{<\alpha(j)}_{1}(j)\big)$.
    Namely, we repeat the following for every index~$t \in \Big[n - \big(\abs*{Z(j)} - 1\big)\Big]$:
    \begin{enumerate}
        \item Introduce the vertex $r_t$ with label set $\{\lr\}$.
        \item Repeat $D-1$ times: introduce a vertex with label set $\{\lp, \lfive\}$, join $\lp$ and $\labelf$, forget $\labelf$, introduce gadget vertices of $F$, join $\lp$ and $\labelf$, forget $\lp$. 
        \item Introduce the vertex $p_{n+t,D}$ with label set $\{\lp, \lfive\}$, join $\lp$ and $\labelf$, forget $\labelf$. 
        \item Introduce gadget vertices of $F'(p_{n+t,D},r_t)$, join $\lp$ and $\labelfprimeleft$, forget $\labelfprimeleft$, join $\lr$ and $\labelfprimeright$, forget~$\labelfprimeright$.
        \item Introduce gadget vertices of $F'(p_{n+t,D}, d_2')$, join $\lwtwoprime$ and $\labelfprimeleft$, forget $\labelfprimeleft$, join $\lp$ and $\labelfprimeright$, forget~$\labelfprimeright$, forget $\lp$.
        \item Introduce gadget vertices of $F'(r_t, d_2')$. Join $\lr$ and $\labelfprimeleft$, forget $\labelfprimeleft$, join $\labelfprimeright$ and $\lwtwoprime$, forget $\labelfprimeright$.
        \item Introduce gadget vertices of $F(r_t, d_2)$. Join $\lr$ and $\labelf$, join $\labelf$ and $\lwtwo$, forget $\labelf$, forget $\lr$.
        \item Join $\lfive$ and $\lfiveold$, relabel $\lfive$ to $\lfiveold$.
    \end{enumerate}
    Then forget $\lfiveold$. This concludes the construction of the multi-expression of $G^*$.

    Now we bound the running time of the process.
    The values $C$ and $D$ used in the definition of the gadgets are polynomial in $n'$ and $k'$. 
    The number $m$ of edges of $G$ is upper-bounded by ${{k'} \choose 2} \cdot (n')^2$ and we have $\abs*{\cS} = k'$, $k \in \Theta(\log{k'})$, and $n = n'-1$. 
    Going through all combinations of~$j \in [m+1]$, $S \in \cS$, and $i \in [n]$ thus takes a polynomial in $n'$ and $k'$ number of iterations.
    A straightforward check of the operations together with the polynomial bounds on $C$ and $D$ imply that one iteration takes time polynomial in $n'$ and $k'$.
    The graph $G^*$ and the value $\budget$ can be computed in the same time along this process.
\end{proof}
Now we are finally ready to prove \cref{thm:maxcut-mcw-lb}.
\begin{proof}[Proof of \cref{thm:maxcut-mcw-lb}]
    Suppose there is a computable function $f$ and an algorithm $\cA$ that gets an instance $(\widehat G, \widehat{\budget})$ and a multi-$\widehat{k}$-expression of a graph $\widehat G$ and solves the \textsc{Max Cut} problem in time $f(\widehat{k}) \cdot \widehat{n}^{2^{o(\widehat{k})}}$ where $\widehat n$ denotes the number of vertices in $\widehat G$.
    
    Then we can solve the \multicolored{} problem as follows.
    Let~$G = (U_1 \cup \dots \cup U_{k'}, E)$ be a graph with pairwise disjoint independent sets $U_1, \dots, U_{k'}$ of size~$n'$ each. 
    First, as described in this section in time polynomial in $n'$ and $k'$ we compute an equivalent instance $(G^*, \budget^*)$ of \textsc{Max Cut} together with a multi-$k^*$-expression of $G^*$ where~$k^*$ satisfies~$k^* \in \cO(\log{k'})$.
    In particular, the number $n^*$ of vertices of $G^*$ is polynomial in $n'$ and $k'$.
    And we run $\cA$ to solve this instance.
    Running $\cA$ then takes time
    \[
        f(k^*) \cdot (n^*)^{2^{o(k^*)}} \leq f(k^*) \cdot (n^*)^{2^{o(\log{k'})}} \leq g(k') \cdot (n')^{2^{o(\log{k'})}} \leq g(k') \cdot (n')^{o(k')}
    \]
    where $g$ is some computable function.
    So altogether the time spent on solving the given instance of \multicolored{} is also upper-bounded by $g(k') \cdot (n')^{o(k')}$ which by \cref{thm:multicolored-lb} contradicts ETH.
\end{proof}

Fomin et al.~\cite{FominGLS14} provided an $n^{\cO(k)}$ algorithm for \textsc{Max Cut} working on $k$-expressions, i.e., parameterized by clique-width:
\begin{theorem}[Subsection 4.2 in~\cite{FominGLS14}]
    Given a graph $G$ on $n$ vertices and a $k$-expression of $G$, in time $n^{\cO(k)}$ the \textsc{Max Cut} problem can be solved.
\end{theorem}

Because it is straightforward to transform a multi-$k$-expression into a $2^k$-expression of the same graph~\cite{Furer17}, we obtain the following algorithm for \textsc{Max Cut} for multi-clique-width implying the tightness of our lower bound under ETH:
\begin{theorem}
    Given a graph $G$ on $n$ vertices and a multi-$k$-expression of $G$, in time $n^{\cO(2^k)}$ we can solve the \textsc{Max Cut} problem.
\end{theorem}

\section{Algorithm for Hamiltonian Cycle}\label{sec:hamiltonian-cycle}
A \emph{Hamiltonian cycle} in a graph $G$ is a cycle visiting each vertex of $G$ precisely once.
In the \textsc{Hamiltonian Cycle} problem, given a graph $G$, we are asked whether $G$ contains a Hamiltonian cycle.

In this section, we show that \textsc{Hamiltonian Cycle} can be solved in time $n^{\cO(k)}$ on graphs provided with multi-$k$-expressions. 
The crux behind the algorithm can be sketched as follows. 
A partial solution of \textsc{Hamiltonian Cycle} is usually a path packing of a subgraph created in the expression so far.
Next, every vertex is incident with precisely two edges in any Hamiltonian cycle of the final graph~$G$.  
So for every endpoint of a path in such a path packing, it suffices to guess one label  ``responsible'' for the creation of the missing incident edge of the sought Hamiltonian cycle later.
Here, for a path consisting of a single vertex, we treat this vertex as two independent endpoints.
With this idea in hand, we are able to adapt the algorithm by Bergougnoux et al.~\cite{BergougnouxKK20} working on clique-expressions to multi-expressions.

Let us remind the reader that clique-expressions are essentially multi-expressions where no vertex holds multiple labels simultaneously.
To solve the problem, Bergougnoux et al.~\cite{BergougnouxKK20} study the unique auxiliary multigraph of a partial solution, i.e., a path packing of the current graph. 
In this multigraph, the vertex set is equal to the label set $[k]$ of the provided clique-expression. And for every path, say with endpoints $u$ and $v$, of the path packing, there is one edge in this multigraph between the labels of $u$ and $v$---this edge may, in general, be a loop.
Even though the number of possible auxiliary multigraphs can be as large as~$n^{\Theta(k^2)}$---indeed, for every pair of labels the multiplicity of the corresponding edge can vary between $0$ and $n$---Bergougnoux et al.~\cite{BergougnouxKK20} apply the technique of representative sets to argue that what really matters about such an auxiliary multigraph is its degree sequence, i.e., the degree of every label, and the partition of labels into connected components.
This reduces the number of relevant partial solutions to be at most~$n^{\cO(k)} \cdot k^{\cO(k)} = n^{\cO(k)}$.

In our algorithm for multi-$k$-expressions, every path packing now yields, in general, multiple auxiliary graphs. 
Namely, for every path and each of its endpoints, any of its up to $k$ labels can potentially be later used at some join-node to create the missing incident edge.
Here, for a path on exactly one vertex, we treat this vertex as two endpoints and thus also get two labels for this path.
Crucially, the number of types of auxiliary multigraphs based on their degree sequences and connected components remains bounded by $n^{\cO(k)}$.
In the remainder of this section, we work out the details of this idea and prove the correctness.
Let us remark for clarification that, as Bergougnoux et al.~\cite{BergougnouxKK20}, we first focus on solving the \textsc{Hamiltonian Path} problem, i.e., finding a path visiting every vertex exactly once, and in the very end, show how \textsc{Hamiltonian Cycle} reduces to it.
We rely on the notation introduced by Bergougnoux et al.~\cite{BergougnouxKK20} as much as possible.
To simplify the following definitions, from now on the graph $G$ on $n$ vertices and a multi-$k$-expression of $G$ are fixed.
In this section, all multigraphs have vertex set $[k]$.
Two multigraphs $A$ and $B$ are \emph{equal} and we write~$A = B$, if for every pair of labels~$a, b \in [k]$, the multiplicity of $\{a, b\}$ is equal in $A$ and~$B$.

For a node $x$ of the multi-expression, a \emph{path packing} $\cP$ of $x$ is a spanning, i.e., containing all vertices of $G_x$, set of vertex-disjoint paths, possibly of length $0$ in the graph~$G_x$.
We use~$E(\cP)$ to denote the set of edges used by the path packing~$\cP$.
Because any path packing contains all vertices of $G_x$, with a slight abuse of notation we sometimes treat $\cP$ as the graph~$\big(V_x, E(\cP)\big)$.
Unless specified otherwise, with a \emph{path} in $\cP$ we always refer to a maximal path of~$\cP$, i.e., to a connected component of the graph~$\big(V_x, E(\cP)\big)$.
A \emph{label choice} $\phi$ of~$\cP$ (in~$x$) is a mapping $\phi \colon \cP \to {[k] \choose 1} \cup {[k] \choose 2}$ such that the following is satisfied for every path~$P \in \cP$.
Let $u$ and $v$ be the endpoints of $P$ (with $u = v$ if and only if $P$ has length $0$) and let the labels~$i, j \in [k]$ be such that $\{i, j\} = \phi(\cP)$ (with possibly $i = j$).
Then we have 
\begin{itemize}
    \item either $i \in \lab_x(u)$ and $j \in \lab_x(v)$
    \item or $i \in \lab_x(v)$ and $j \in \lab_x(u)$.
\end{itemize} 
In simple words, for every endpoint of a path in $\cP$, the mapping $\phi$ remembers one label.
And if the path has length 0, i.e., both endpoints are the same vertex, we treat this vertex as two independent endpoints.
Note that if some endpoint of a path in $\cP$ has an empty label set in $\lab_x$, then the path packing has no label choice.

\begin{figure}[t]
    \includegraphics{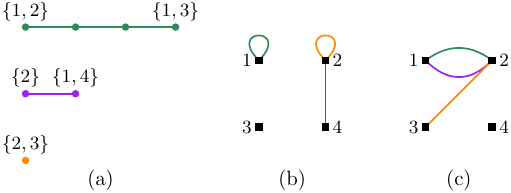}
    \centering
    \caption{(a) A path packing of a graph together with label sets of endpoints of the paths and (b), (c) two examples of auxiliary multigraphs for it.}
    \label{fig:ham-cycle-aux}
\end{figure}

A partial solution of a node $x$ is a pair $(\cP, \phi)$ such that $\cP$ is a path packing of $x$ and~$\phi$ is a label choice of $\cP$ in $x$.
And the set $\Pi(x)$ is defined as the set of all partial solutions of $x$.
The \emph{auxiliary multigraph} $\aux(\cP, \phi)$ of $(\cP, \phi)$ has vertex set $[k]$ and for all labels~$i, j \in [k]$, the multiplicity of $\{i,j\}$ in $\aux(\cP, \phi)$ is equal to $\abs*{\phi^{-1}\big(\{i,j\}\big)}$.
In simple words, for every path in $\cP$, the multigraph contains one edge with endpoints being the labels to which this path is mapped by $\phi$~(see \cref{fig:ham-cycle-aux} for an illustration).
We remind the reader that such multigraphs are allowed to have loops. 
Following the convention by Bergougnoux et al.~\cite{BergougnouxKK20}, we refer to (the edges of) such multigraphs as \emph{red}.
This allows us to speak about \emph{red-blue Eulerian trails} in a moment.
We also define the set 
\[
    \aux(x) = \big\{\aux(\cP, \phi) \mid (\cP, \phi) \in \Pi(x)\big\}.
\]

For two multigraphs $A$ and $B$, we write $A \simeq B$ if for every label~$i \in [k]$, we have~$\deg_A(i) = \deg_B(i)$ and the sets of connected components of $A$ and $B$ coincide.
For a set $\cA$ of multigraphs, the operation $\rreduce(\cA)$ returns a subset of $\cA$ that contains precisely one element from each equivalence class of $\cA / \simeq$.
This output has a bounded size and can be computed efficiently as shown by Bergougnoux et al.~\cite {BergougnouxKK20}:
\begin{lemma}[Lemma 2 in \cite{BergougnouxKK20}]
\label{lem:size-of-repr}
    For every subset~$\cA \subseteq \aux(x)$, we have $\abs*{\rreduce(\cA)} \leq n^k \cdot 2^{k \cdot (\log(k)+1)}$ and we can moreover compute $\rreduce(\cA)$ in time $\mathcal{O}\big(\abs*{\cA} \cdot n \cdot k^2 \cdot \log(n \cdot k)\big)$.
\end{lemma}

Let $C$ be a multigraph each of whose edges is colored red or blue and such that $C$ contains at least one edge.
A \emph{red-blue Eulerian trail} of $C$ is a closed (i.e., the first and the last vertex coincide) walk visiting each edge of $C$ precisely once such that red and blue edges alternate on this walk, in particular, the first and the last edge of this walk have different colors.
Note that if $C$ admits a red-blue Eulerian trail, then it contains at least one blue and at least one red edge.
In our applications, red edges will represent paths of partial solutions created so far while blue edges correspond to hypothetic paths possibly created later in the expression. 

For two multigraphs $A$ and $B$, we use $A \uplus B$ to denote the edge-disjoint union of $A$ and~$B$, i.e., the multigraph such that for all labels~$i, j \in [k]$, the multiplicity of~$\{i, j\}$ in~$A \uplus B$ is equal to the sum of the multiplicities of~$\{i, j\}$ in~$A$ and~$B$.

Now let $\cA$ and $\cB$ be two families of multigraphs such that the edges of those multigraphs are colored red.
We write~$\cA \lesssim \cB$ if for every multigraph $M$ on the vertex set $[k]$ whose edges are colored blue, the following holds: if there is a multigraph~$B \in \cB$ such that~$B \uplus M$ admits a red-blue Eulerian trail, then there also exists a multigraph~$A \in \cA$ such that~$A \uplus M$ admits a red-blue Eulerian trail.
In this case, we also say that the family~$\cA$ \emph{represents} the family~$\cB$.
Observe that $\lesssim$ is a transitive relation.
Bergougnoux et al.~\cite {BergougnouxKK20} showed that the operation~$\rreduce$ returns a representative, namely:
\begin{lemma}[Lemma 6 in \cite{BergougnouxKK20}]
\label{lemma:reduce-representation}
    For every family $\cA \subseteq \aux(x)$, it holds that $\rreduce_H(\cA) \lesssim \cA$.
\end{lemma}

Let $r$ denote the root of the multi-expression.
Similarly to Bergougnoux et al.~\cite{BergougnouxKK20}, we will later show that computing a set $\cA_r$ representing the family $\aux(r)$ suffices to solve the problem (cf.\ \cref{lem:hc-root-checking}).
So in the following we focus on the computation, for every node $x$ of the multi-expression, of a bounded-size set~$\cA_x$ of red auxiliary graphs satisfying~$\cA_x \subseteq \aux(x)$ and~$\cA_x \lesssim \aux(x)$.
For an introduce-node, there is only one auxiliary multigraph whose computation is straightforward.
\begin{lemma}\label{lem:hc-leaf}
	Let $(G_x, \lab_x) = v\langle i \rangle$ for some vertex $v$ and label $i \in [k]$. 
	Then the family~$\aux(x)$ consists of a single multigraph whose only edge is the loop at the vertex $i$.
\end{lemma}

\begin{proof}
	Because the graph $G_x$ contains a single vertex $v$, the unique path packing of $x$ has a single path consisting of the vertex $v$ only.
	The label set of $v$ in $\lab_x$ is the set $\{i\}$ so the unique label choice of this path packing maps this path to the set $\{i\}$.
	So the lemma follows.
\end{proof}

In the next four lemmas, we show how the desired set $\cA_x$ for a node $x$ can be computed from such sets of the children of $x$.
If $x$ forgets a label $i$, it suffices to simply discard the multigraphs in which the degree of $i$ is non-zero.
We show that this maintains representation.
\begin{lemma}\label{lem:hc-forget-label}
	Let $(G_x, \lab_x) = \rho_{i \to \emptyset}(G_y, \lab_y)$ for some label $i \in [k]$. 
	Further, let a subset~$\cA_y \subseteq \aux(y)$ satisfy $\cA_y \lesssim \aux(y)$.
	We define the set $\cA_x$ as
	\[
		\cA_x = \big\{A \in \cA_y \mid \deg_A(i) = 0\big\}.
	\]
	Then we have $\cA_x \subseteq \aux(x)$ and $\cA_x \lesssim \aux(x)$.
\end{lemma}

\begin{proof}
	We start by showing the inclusion $\cA_x \subseteq \aux(x)$.
	For this, consider a multigraph~$A \in \cA_x$.
	Due to $A \in \cA_y \subseteq \aux(y)$, we have $A = \aux(\cP, \phi)$ for some partial solution $(\cP, \phi)$ of~$y$.
    Because the graphs~$G_x$ and~$G_y$ are equal, $\cP$ is also a path packing of $x$.
	The property~$\deg_A(i) = 0$ implies that $\phi$ maps no path in $\cP$ to any set containing $i$.
    Furthermore, for every label $j \neq i \in [k]$, if a vertex of $G_y$ contains a label $j$ in $\lab_y$, then it also does in $\lab_x$.
    So $\phi$ is a label choice of $\cP$ in $x$ too, the pair~$(\cP, \phi)$ is thus a partial solution of~$x$, and we have~$A = \aux(\cP, \phi) \in \aux(x)$ so the claim follows.
	
	It remains to prove the representation $\cA_x \lesssim \aux(x)$.
    First, observe that by definition of a forget-label-operator, any partial solution of $x$ is a partial solution of $y$ as well so we have~$\aux(x) \subseteq \aux(y)$. 
	Now consider an arbitrary partial solution $(\cP, \phi)$ of $x$ (i.e., a red multigraph~$\aux(\cP, \phi) \in \aux(x)$) and a blue multigraph $M$ and a such that $\aux(\cP, \phi) \uplus M$ admits a red-blue Eulerian trail.
	Because no vertex of $G_x$ contains a label~$i$ in $\lab_x$, we have~$\deg_{\aux(\cP, \phi)}(i) = 0$.
    This implies~$\deg_{M}(i) = 0$.
    Indeed, if some blue edge of $M$ contained an endpoint $i$, then in the red-blue Eulerian trail, some consequent red edge of $\aux(\cP, \phi)$ with an endpoint in $i$ would exist too.
	Now, due to $\aux(\cP,\phi) \in \aux(x) \subseteq \aux(y)$ and~$\cA_y \lesssim \aux(y)$, there exists a red multigraph~$A^* \in \cA_y$ such that~$A^* \uplus M$ also admits a red-blue Eulerian trail.
    Due to~$\deg_{M}(i) = 0$, we analogously get $\deg_{A^*}(i) = 0$ as for any blue edge visited by the red-blue Eulerian trail, there is a consequent red edge with the same endpoint.    
    Hence,~$A^* \in \cA_x$ holds.
	This proves~$\cA_x \lesssim \aux(x)$.
\end{proof}

To handle add-label nodes, we define the following notation.
Let~$i \neq j \in [k]$ be two labels.
We define the mapping $R^{i \to j} \colon [k] \to 2^{[k]}$ as 
$R^{i \to j}(i) = \{i, j\}$ and $R^{i \to j}(\ell) = \{\ell\}$ for every label~$\ell \neq i \in [k]$.
For a multigraph $A$, we define the family $A^{i \to j}$ consisting of all multigraphs~$A'$ for which there exists a bijection $\psi \colon E(A) \to E(A')$ with the following property.
For every edge $\big(e, \{a, b\}\big)$ of $A$, let $\big(e', \{a', b'\}\big) = \psi\big(e, \{a, b\}\big)$, then we have
\begin{align}\label{eq:bijection-psi-case-1}
    e' = e && a' \in R^{i \to j}(a) && b' \in R^{i \to j}(b)
\end{align}
or 
\begin{align}\label{eq:bijection-psi-case-2}
    e' = e && a' \in R^{i \to j}(b) && b' \in R^{i \to j}(a).
\end{align}

Now if a node $x$ of the multi-expression adds label $j$ to label $i$, for every endpoint of a path in a path packing mapped to $i$, we can decide whether it remains mapped to $i$ or ``switches'' to $j$.
So on the level of some red multigraph~$A$, for every endpoint $i$ of some edge of~$A$, we can freely decide whether this endpoint remains in~$i$ or switches to~$j$. 
The definition of the family $A^{i \to j}$ then reflects all possibilities for such switches.
Also observe that even though deciding on the switches of each single edge of $A$ naively yields the running time of~$2^{\Theta(n)}$, we may compute the family $A^{i \to j}$ in time $n^{\cO(k)}$ using the following observation.
For each of the~$k$ labels $a \in [k]$, there are only $\cO(n^2)$ possibilities for the number of edges with the endpoints in $i$ and $a$ that switch.
Indeed, if $a \neq i$, then up to $n$ edges may switch the endpoint $i$ to $j$.
And if $a = i$, then up to $n$ edges may switch only one endpoint $i$ to $j$ and keep the other and up to $n$ further edges may switch both endpoints to $j$.
We now show that this process maintains representation.
The proof uses some ideas by Bergougnoux et al.~\cite{BergougnouxKK20} for a relabel-node $\rho_{i \to j}$ in clique-expressions.
But in our case it gets more involved since an arbitrary number of endpoints of paths may remain mapped to $i$.

\begin{lemma}\label{lem:hc-add-label}
	Let $(G_x, \lab_x) = \rho_{i \to \{i,j\}}(G_y, \lab_y)$ for some $i \neq j \in [k]$.
	Further, let a subset~$\cA_y \subseteq \aux(y)$ satisfy~$\cA_y \lesssim \aux(y)$.
	We define the set $\cA_x$ as
	\[
		\cA_x = \bigcup_{A \in \cA_y} A^{i \to j}.
	\]
	Then we have $\cA_x \subseteq \aux(x)$ and $\cA_x \lesssim \aux(x)$.
\end{lemma}
\begin{proof}
	We start with the inclusion~$\cA_x \subseteq \aux(x)$.
	For this, consider a multigraph $A \in \cA_y$ and a multigraph $A' \in A^{i \to j}$.
    Then there exists a bijection $\psi \colon E(A) \to E(A')$ satisfying equality~\eqref{eq:bijection-psi-case-1} or~\eqref{eq:bijection-psi-case-2} for the image every edge $\big(e, \{a, b\}\big)$ of $A$.
    Consider a partial solution~$(\cP, \phi)$ of $y$ with~$\aux(\cP, \phi) = A$.
    Because the graphs $G_x$ and $G_y$ are equal, $\cP$ is also a path packing of $x$.
    Now we define the label choice $\phi'$ of $\cP$ in $x$ as follows.
    First, to every path~$P \in \cP$, we bijectively assign an edge~$\big(e_P, \{a_P, b_P\}\big)$ of the multigraph~$A = \aux(\cP, \phi)$ satisfying~$\phi(P) = \{a_P, b_P\}$.
    Note that such a bijection exists by definition of the auxiliary multigraph~$\aux(\cP, \phi)$.
    Then we define the image $\phi'(P) = \{a_P', b_P'\}$  where 
    \[
        \big(e_P, \{a_P', b_P'\}\big) = \psi\big(e_P, \{a_P, b_p\}\big)
    \]
    for every path $P \in \cP$.
    We now argue that this is indeed a label choice for $\cP$ in $x$.
    Let again~$P \in \cP$ be an arbitrary path. And let $u$ and $v$ be the endpoints of $P$ with $a_P \in \lab_y(u)$ and $b_P \in \lab_y(v)$, their existence follows from $\phi(P) = \{a_P, b_P\}$.
    Because one of equalities~\eqref{eq:bijection-psi-case-1} or \eqref{eq:bijection-psi-case-2} is satisfied, we may assume without loss of generality that $a_P' \in R^{i \to j}(a_P)$ and~$b_P' \in R^{i \to j}(b_P)$ hold.
    By definition of the mapping~$R^{i \to j}$, if $a_P = i$, then $a_P'$ is either $i$ or $j$, and if~$a_P \neq i$, then~$a_P' = a_P$ holds.
    By definition of an add-label-node, the property $a_P \in \lab_y(u)$ then implies $a_P' \in \lab_x(u)$.
    And a symmetric argument implies~$b_P' \in \lab_x(v)$.
    So~$\phi'$ is indeed a path labeling of~$\cP$ in~$x$, i.e., we have~$(\cP, \phi) \in \Pi(x)$.
    Furthermore, by the choice of $\psi$ and the construction of $\phi'$, we have~$A' = \aux(\cP, \phi')$, i.e., $A' \in \aux(x)$.
    Indeed, every path $P$ of $\cP$ bijectively corresponds to the edge $\big(e_P, \{a_P', b_P'\}\big) = \psi\big(e_P, \{a_P, b_P\}\big)$ of $A'$.
	This implies the inclusion $\cA_x \subseteq \aux(x)$.
    
	Now we prove the representation~$\cA_x \lesssim \aux(x)$.
	For this, consider a red multigraph~$A \in \aux(x)$ and a blue multigraph $M$ such that $A \uplus M$ admits a red-blue Eulerian trail.
	Let~$(\cP, \phi)$ be a partial solution of $x$ with $A = \aux(\cP, \phi)$.
	We define the label choice $\phi'$ of $\cP$ in $y$ obtained from $\phi$ as follows.
    Let $P \in \cP$ be a path of the path packing $\cP$, let $\{c,d\} = \phi(P)$, and let $u$ and $v$ be the endpoints of $P$ such that $c \in \lab_x(u)$ and $d \in \lab_x(v)$.
    Let us emphasize that these elements as well as several objects below depend on the path $P$ but we omit the subscript $P$ for the sake of readability.
    We define the labels $c'$ and $d'$ as
    \begin{align*}
        c' = 
        \begin{cases}
            i & \text{if } c = j \text{ and } i \in \lab_y(u) \\
            c & \text{otherwise }
        \end{cases}
        &&
        d' = 
        \begin{cases}
            i & \text{if } d = j \text{ and } i \in \lab_y(u) \\
            d & \text{otherwise }
        \end{cases}
        .
    \end{align*}
    And we define the image~$\phi'(P) = \{c', d'\}$
    and we repeat the procedure for every path $P \in \cP$.
    By definition of an add-label node, the constructed mapping~$\phi'$ is a label choice of $\cP$ in $y$.
    Indeed, the label $c$ can only belong to the set~$\lab_x(u) \setminus \lab_y(u)$ if $c = j$ and $i \in \lab_y(u)$, and the analogous property applies to the vertex~$v$ and label~$d$.
    So we have $(\cP, \phi') \in \aux(y)$.
    Let now $A' = \aux(\cP, \phi') \in \aux(y)$ (see \cref{fig:ham-cycle-add-label} for an illustration).

    \begin{figure}[t]
    \includegraphics{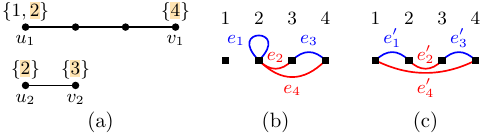}
    \centering
    \caption{Add-label operation $(G_x, \lab_x) = \rho_{1 \to \{1,2\}}(G_y, \lab_y)$. (a) A path packing $\cP$ of the node~$x$ together with a label choice in $x$ (in orange). (b) In red: the corresponding auxiliary graph $A$, in blue: a multigraph $M$ such that $A \uplus M$ admits a red-blue Eulerian trail $e_1, e_2, e_3, e_4$. (c) In the labeling function~$\lab_y$, the vertex~$u_1$ certainly contained label $1$ but not neccesarily label $2$ while the vertex~$u_2$ had label $2$ in $\lab_y$ already. So in red: an auxiliary multigraph $A'$ of $\cP$ where the corresponding label choice in $y$ maps $u_1$ to $1$ and remaining values stay unchanged. The endpoint $1$ of the edge $e_4'$ is marked. So when constructing the blue multigraph $M'$ from $M$ the edge~$\big(e_4, \{2,4\}\big)$ is replaced by the edge~$\big(e_4', \{1,4\}\big)$. Now the multigraph $A' \uplus M'$ admits a red-blue Eulerian trail $e'_1, e'_2, e'_3, e'_4$.}
    \label{fig:ham-cycle-add-label}
    \end{figure}

    Similarly to the first part of the proof, to every path in $\cP$ we bijectively assign an edge of~$A$ relying on the the label choice $\phi$, and we bijectively assign an edge of~$A'$ relying on the the label choice $\phi'$.
    In the following, speaking of an edge of $A$ or $A'$ \emph{corresponding} to a path of $\cP$, we refer to those bijections. 
    With this, there is a natural pairing of edges (and pairs of endpoints of those edges) of $A$ and $A'$ given by the paths in $\cP$.
    Namely, for every path~$P \in \cP$, using the notation above, there is a corresponding edge, say~$\big(e, \{c, d\}\big)$, in~$A$ and a corresponding edge, say~$\big(e', \{c', d'\}\big)$, in $A'$.
    We pair the edges~$e$ and~$e'$.
    Further, we pair the endpoint $c$ (resp.\ $d$) of $e$ with the endpoint $c'$ (resp.\ $d'$) of $e'$. 
    We say that an endpoint~$c$ (resp.\ $d$) of the edge $e$ of $A$ is \emph{marked} if we have $c' \neq c$ (resp.\ $d' \neq d$).
    Observe that an endpoint of an edge of $A$ can only be marked if it is equal to $j$.
    We emphasize that the label $j$ may be marked as an endpoint of one edge in $A$ and not marked as an endpoint of a different edge.
    Similarly, for a loop at $j$ in $A$, the label~$j$ may be marked as one endpoint of this loop and not marked as the other endpoint.
    
    Let now $T$ be a red-blue Eulerian trail of $A \uplus M$ that exists by assumption.
    Let~$e$ be an arbitrary blue edge of~$M$.
    And let $a_1$ and $a_2$ be the endpoints of $e$ so that in $T$ the edge $e$ is traversed from $a_1$ to $a_2$.
    Further, let $e_1$ be the red edge of $A$ preceding $e$ in $T$, and let $e_2$ be the red edge of $A$ following $e$ in $T$. 
    We define the labels $a_1'$ and $a_2'$ as
    \[
        a_i' =
        \begin{cases}
            i & \text{if $a_i$ is a marked endpoint of $e_i$} \\
            a_i & \text{otherwise }
        \end{cases}
    \]
    for every index $i \in [2]$.
    We replace the edge $e$ of $M$ by an edge, say $e'$, with endpoints~$a_1'$ and~$a_2'$, repeat this process for every edge of $M$, and define~$M'$ to be the arising blue multigraph.
    For the arguments below, we will use the following pairing.
    We pair every edge~$e$ of~$M$ with the edge~$e'$ of~$M'$ and we pair the endpoint~$a_1$ (resp.~$a_2$) of~$e$ with the endpoint~$a_1'$ (resp.~$a_2'$) of~$e'$---this pairs the edges of~$M$ and~$M'$ and their endpoints. 
    And we say that $a_1'$ (resp.\ $a_2'$) is a \emph{marked} endpoint of $e'$ if $a_1' \neq a_1$ (resp.\ $a_2' \neq a_2$).
    Observe that an endpoint of an edge in $M'$ can only be marked if it is equal to $i$.

    Now the multigraph $A' \uplus M'$ admits a red-blue Eulerian trail $T'$ constructed by starting with $T$ and replacing each edge by its pair, i.e., every edge of $A$ is replaced by its pair in~$A'$ and every edge of $M$ is replaced by its pair in $M'$.
    By construction of $M'$, the sequence~$T'$ remains a closed walk.
    Indeed, by construction, we replace some visits of the label $j$ by the label $i$ in such a way that such a replacement either occurs for both edges involved in this visit, or for none of the two edges.
    Clearly, it also remains red-blue and visits every edge of~$A' \uplus M'$ precisely once.

    So due to $A' \in \aux(y)$ and $\cA_y \lesssim \aux(y)$, there exists a red multigraph~$A^* \in \cA_y$ such that~$A^* \uplus M'$ also admits a red-blue Eulerian trail, say $T^*$.
    We will now show how to construct a red multigraph $\widetilde A \in \cA_x$ from $A^*$ and $T^*$ such that $\widetilde A \uplus M$ admits a red-blue Eulerian trail, this then concludes the proof of $\cA_x \lesssim \aux(x)$.
    Vaguely speaking, the construction of $\widetilde A$ could be imagined as the reversion of the construction of $M'$.

    Let $e^*$ be an arbitrary edge of $A^*$ and let $a^*_1$ and $a^*_2$ be the endpoints of $e^*$ so that in $T^*$ the edge $e^*$ is traversed from $a_1^*$ to $a_2^*$.
    Further, let $e'_1$ be the edge of $M'$ preceding $e^*$ in $T^*$ and let 
    $e'_2$ be the edge of $M'$ following $e^*$ in $T^*$. 
    We define the vertices $\widetilde{a}_1$ and $\widetilde{a}_2$ as
    \[
        \widetilde{a}_i =
        \begin{cases}
            j & \text{if $a^*_i$ is a marked endpoint of $e'_i$} \\
            a^*_i & \text{otherwise }
        \end{cases}
    \]
    for every index~$i \in [2]$.
    We now replace the edge $e^*$ of $A^*$ by an edge, say $\widetilde{e}$, with endpoints~$\widetilde{a}_1$ and~$\widetilde{a}_2$.
    Repeating this process for every edge of $A^*$, we obtain a red multigraph $\widetilde{A}$.
    We pair the edge $e^*$ of $A^*$ with the edge $\widetilde{e}$ of $\widetilde{A}$, and we pair the endpoint $a^*_1$ (resp.\ $a^*_2$) of $e^*$ with the endpoint $\widetilde{a}_1$ (resp.\ $\widetilde{a}_2$) of $\widetilde{e}$. 

    Because only the label $i$ can be a marked endpoint of some edge in $M'$, to obtain the multigraph $\widetilde{A}$, we take a multigraph $A^* \in \cA_y$ and for some of its edges, replace an endpoint~$i$ (possibly both in case of a loop) by an endpoint $j$.
    So the multigraph $\widetilde{A}$ belongs to the family~$(A^*)^{i \to j} \subseteq \cA_x$ by definition.
    And to obtain a red-blue Eulerian trail $\widetilde T$ of $\widetilde A \uplus M'$, we start with $T^*$ and replace each edge by its pair, i.e., every edge of $A^*$ is replaced by its pair in~$\widetilde A$ and every edge of $M'$ is replaced by its pair in $M$.
    Recall that to obtain the multigraph~$M$ from~$M'$, we take every marked endpoint (being in label~$i$) of every edge of $M'$ and replace it by~$j$.
    So by construction of $\widetilde A$, the sequence $\widetilde T$ indeed remains a closed walk.
    Clearly, it remains red-blue and visits every edge of $\widetilde A \uplus M$ precisely once.
    This concludes the proof of $\cA_x \lesssim \aux(x)$.
\end{proof}

If the node $x$ forms a union of two vertex-disjoint graphs, the partial solutions for node $x$ are simply given by pairwise unions of partial solutions of the children. 
And this, in fact, also maintains representation as shown in the next lemma.
The proof closely follows the proof for clique-expressions by Bergougnoux et al.~\cite{BergougnouxKK20}.
\begin{lemma}\label{lem:hc-union-node}
	Let $(G_x, \lab_x) = (G_{y_1}, \lab_{y_1}) \oplus (G_{y_2}, \lab_{y_2})$ where $G_{y_1}$ and $G_{y_2}$ are vertex-disjoint non-empty graphs.
	Further, let a subset $\cA_{y_1} \subseteq \aux(y_1)$ satisfy $\cA_{y_1} \lesssim \aux(y_1)$ and let a subset~$\cA_{y_2} \subseteq \aux(y_2)$ satisfy $\cA_{y_2} \lesssim \aux(y_2)$.
	We define the set $\cA_x$ as
	\[
		\cA_x = \{A_1 \uplus A_2 \mid A_1 \in \cA_{y_1}, A_2 \in \cA_{y_2}\}.
	\]
	Then we have $\cA_x \subseteq \aux(x)$ and $\cA_x \lesssim \aux(x)$.
\end{lemma}

\begin{proof}
	By definition of a union-node, the graph $G_x$ is a disjoint union of $G_{y_1}$ and $G_{y_2}$ and every vertex of $G_{y_1}$ and $G_{y_2}$ preserves its label from $\lab_{y_1}$ and $\lab_{y_2}$, respectively, in~$\lab_x$.
	Consider a pair $A_1 \in \aux(y_1)$ and $A_2 \in \aux(y_2)$.
	And consider partial solutions~$(\cP_1, \phi_1)$ of $y_1$ and $(\cP_2, \phi_2)$ of $y_2$ satisfying $\aux(\cP_1, \phi_1) = A_1$ and $\aux(\cP_2, \phi_2) = A_2$.
	Then the pair~$(\cP_1 \cup \cP_2, \phi_1 \cup \phi_2)$
    is trivially a partial solution of $x$ and we have $\aux(\cP_1 \cup \cP_2, \phi_1 \cup \phi_2) = A_1 \uplus A_2$ implying $A_1 \uplus A_2 \in \aux(x)$.
	On the other hand, consider a multigraph $A \in \aux(x)$ and a partial solution $(\cP, \phi)$ of~$x$ with $A = \aux(\cP, \phi)$. 
    Observe that every path of $\cP$ is contained in precisely one of $G_{y_1}$ and $G_{y_2}$.
    So for every index~$i \in [2]$, let $\cP_i$ denote the restriction of $\cP$ to the paths contained in the graph~$G_{y_i}$.
    Then $\cP_i$ is a path packing of $y_i$ and the pair $(\cP_{i}, \phi_{|_{\cP_i}})$
    is a partial solution of $y_i$, i.e., $\aux(\cP_i, \phi_{|_{\cP_i}}) \in \aux(y_i)$.
    Further,~$\cP$ is a disjoint union of $\cP_1$ and $\cP_2$ so we have~$\aux(\cP_1, \phi_{|_{\cP_1}}) \uplus \aux(\cP_2, \phi_{|_{\cP_2}}) = \aux(\cP, \phi) = A$.
	This proves the equality
	\[
		\big\{A_1 \uplus A_2 \mid A_1 \in \aux(y_1), A_2 \in \aux(y_2)\big\} = \aux(x).
	\]
	In particular, this implies the inclusion $\cA_x \subseteq \aux(x)$.
	
	Now we prove the representation~$\cA_x \lesssim \aux(x)$.
    For this, consider an arbitrary red multigraph $A \in \aux(x)$ and a blue multigraph $M$ such that $A \uplus M$ admits a red-blue Eulerian trail, say $T$.   
    By above equality we have $A = A_1 \uplus A_2$ for some multigraphs~$A_1 \in \aux(y_1)$ and $A_2 \in \aux(y_2)$.
     First, we show that there exists a red multigraph $A_1^* \in \cA_{y_1}$ such that~$(A_1^* \uplus A_2) \uplus M$ also admits a red-blue Eulerian trail.
	   For this, consider the blue multigraph $M'$ such that every maximal subwalk of $T$ that does not contain edges from~$A_1$ bijectively corresponds to an edge in $M'$ with the same endpoints as in this subwalk.
     Then every edge of $A_2 \uplus M$ is contained in exactly one such subwalk.
	 Note that because $G_{y_1}$ is a non-empty graph, the multigraph~$A_1$ contains at least one edge so no such subwalk is the whole $T$.
     Also note that by maximality, both end-edges of any such subwalk are some blue edges from $M$---this is because red-blue edges alternate on $T$ and every edge of~$A_1$ and every edge of $A_2$ is red.
	   Now the multigraph~$A_1 \uplus M'$ admits a red-blue Eulerian trail~$T'$ obtained from $T$ by replacing each such subwalk by the corresponding edge from~$M'$.
	 Because the family~$\cA_{y_1}$ represents~$\aux(y_1)$, there exists a red multigraph~$A_1^* \in \cA_{y_1}$ such that~$A_1^* \uplus M'$ admits a red-blue Eulerian trail~$T^*$.
	 By replacing every edge of $M'$ in $T^*$ by the corresponding subwalk of $T$, we obtain a red-blue Eulerian trail $\widetilde{T}$ of $A_1^* \uplus A_2 \uplus M$.
	 Now, analogously, we construct the blue multigraph~$M''$ to have an edge for every maximal subwalk of $\widetilde{T}$ that does not contain edges of $A_2$.
	   A symmetric argument using the representation of $\aux(y_2)$ by~$\cA_{y_2}$ implies that there is a multigraph~$A_2^* \in \cA_{y_2}$ such that $A_1^* \uplus A_2^* \uplus M$ admits a red-blue Eulerian trail. 
       Due to $A_1^* \in \cA_{y_1}$ and $A_2^* \in \cA_{y_2}$, we get $A_1^* \uplus A_2^*\in \cA_x$ by definition.
       So $\cA_x$ represents~$\aux(x)$ as claimed. 
\end{proof}

Now suppose that we want to add an edge between a vertex containing label $i$ and another vertex containing label $j$ to some path packing in such a way that it remains a path packing.
On the level of auxiliary graphs, this corresponds to taking two distinct edges, one with endpoints $i$ and some label $a$ and the other with endpoints $j$ and some label $b$, and replacing them with a single edge with endpoints $a$ and $b$.
At a join-node, at most $n-1$ edges can be added to a path packing this way without breaking the acyclicity.
However, applying the described ``combine-operation'' $n-1$ times might explode the size of the computed set of auxiliary multigraphs.
To avoid this, after every single application of this combination, we will apply the operation $\rreduce$.
In the following lemma, we will show that this approach maintains representation.
But first, we set up the required notation.

Let $i \neq j \in [k]$ be two labels.
For a red multigraph $A$, we define the family $A + \{i,j\}$ of red multigraphs in two steps as follows:
\begin{enumerate}
    \item First, the family~$A + \{i,j\}$ contains~$A$ as a subset.
    \item Second, let $a, b \in [k]$ be two labels and let $\big(e, \{a,i\}\big)$ and $\big(e', \{b,j\}\big)$ be two distinct, i.e., satisfying~$e \neq e'$, edges of $A$.
    Then the multigraph obtained from $A$ by removing the edges $e$ and $e'$ and adding a new edge $\big(e'', \{a,b\}\big)$ belongs to $A + \{i,j\}$.
\end{enumerate}
Let us emphasize that we allow any of $e$ and $e'$ to be a loop.
And we also allow $e$ and $e'$ to have the same set of endpoints: in this case, both edges have endpoints $i$ and $j$ and the multiplicity of $\{i,j\}$ in $A$ is thus at least two. 
It is important that $e$ and $e'$ are distinct: this way, we remove two edges and replace them with a new one.
For the joining sketched above, this corresponds to the fact that we never connect two endpoints of the same path (i.e., ``close a cycle too early'')  but rather make a longer path consisting of two old paths and one new edge.

And for a family $\cA$ of red multigraphs, we define the family $\cA + \{i,j\}$ simply as 
\[
    \cA +  \{i,j\} = \bigcup_{A \in \cA} A +  \{i,j\}.
\]
We also define the family $\cA + 0 \cdot \{i,j\} = \cA$ and for a positive integer $q$, we define the family 
\[
    \cA + q \cdot \{i,j\} = \Bigl(\cA + (q-1) \cdot \{i,j\}\Bigr) +  \{i,j\}.
\]
This reflects starting with the family $\cA$ and applying the operation $+\{i,j\}$ $q$ times.

\begin{lemma}\label{lem:hc-join-node}
	Let $(G_x, \lab_x) = \eta_{i,j}(G_y, \lab_y)$ for some labels~$i \neq j \in [k]$ where $G_y$ is a non-empty graph containing no vertex $v$ with $\{i,j\} \subseteq \lab_y(v)$.
    Then it holds that \[\aux(y) + (n-1) \cdot \{i,j\} = \aux(x).\]
    
	Further, let a subset $\cA_y \subseteq \aux(y)$ satisfy $\cA_y \lesssim \aux(y)$.
    Let $\cA_x^0 = \cA_y$ and for every integer~$q \in [n-2]_0$, let \[\cA_x^{q+1} = \rreduce\big(\cA^q_x + \{i,j\}\big).\]
	We define the set $\cA_x$ as $\cA_x = \cA_x^{n-1}$.
	Then we have $\cA_x \subseteq \aux(x)$ and $\cA_x \lesssim \aux(x)$.
\end{lemma}
\begin{proof}
    The proof of this lemma relies on the proof of the analogous lemma by Bergougnoux et al.~\cite{BergougnouxKK20}. However, for multi-clique-width we need to incorporate non-unique label choices into it. 
    For the sake of completeness, we provide all details.
	We start by proving the following claim:
    \begin{claim}\label{claim:join-one-step-representation}
        For any subset $\cA \subseteq \aux(x)$, we have $\cA + \{i,j\} \subseteq \aux(x)$.
    \end{claim}
    \begin{claimproof}
        Consider a multigraph $A' \in \cA + \{i,j\}$.
        We show that $A' \in \aux(x)$ holds.
        If $A' \in \cA$, the claim follows immediately.
        Otherwise, there exists a multigraph $A \in \cA$ and two distinct edges $\big(e_1, \{a,i\}\big)$ and $\big(e_2, \{b, j\}\big)$ for some labels~$a, b \in [k]$ such that the multigraph~$A'$ is obtained from $A$ by removing the edges $e_1$ and $e_2$ and adding a new edge $\big(e_3, \{a,b\} \big)$.
        We have~$A \in \cA \subseteq \aux(x)$ so let $(\cP, \phi)$ be a partial solution of $x$ with $\aux(\cP, \phi) = A$.
        Then the path packing $\cP$ of $x$ contains two distinct paths, say $P_1$ and $P_2$, with endpoints $v_1,u_1$ and $u_2,v_2$, respectively, such that $a \in \lab_{x}(v_1)$, $i \in \lab_{x}(u_1)$, $j \in \lab_{x}(u_2)$, and $b \in \lab_{x}(v_2)$.
        By definition of a join-node, the graph $G_x$ contains the edge $u_1 u_2$.
        Let $\cP'$ be obtained from~$\cP$ by removing the paths $P_1$ and $P_2$ and adding a new path defined as the concatenation of the path $P_1$, the edge $u_1 u_2$, and the path $P_2$.
        Note that this new path has endpoints $v_1$ and $v_2$ and~$\cP'$ is a path packing of $x$.
        We define the label choice $\phi'$ of $\cP'$ in $x$ so that~$\phi'(P) = \phi(P)$ if $P \in \cP$ and $\phi'(P) = \{a,b\}$ otherwise, i.e., if $P$ is the new path.
        This implies that $(\cP', \phi')$ is a partial solution of $x$, i.e., $(\cP', \phi') \in \Pi(x)$.      
        Finally, observe that we have $A' = \aux(\cP', \phi')$ by construction of $\cP$ and $\phi'$.
        So we get $A' \in \aux(x)$ concluding the proof.
    \end{claimproof}

    Now we prove the first equality of the lemma.
    First, because $G_y$ is a subgraph of $G_x$ on the same vertex set and every vertex has the same label set in $\lab_y$ and $\lab_x$, every partial solution of $y$ is also a partial solution of $x$ implying $\aux(y) \subseteq \aux(x)$.
    So by applying claim~\cref{claim:join-one-step-representation}~$n-1$ times, we get $\aux(y) + (n-1) \cdot \{i,j\} \subseteq \aux(x)$.
    We now show that the other inclusion applies too:
    \begin{claim}
        It holds that $\aux(y) + (n-1) \cdot \{i,j\} \supseteq \aux(x)$.
    \end{claim}
    \begin{claimproof}
        We define the set 
        \[
            \cE = \big\{u v \mid u, v \in V_x, i \in \lab_{x}(u), j \in \lab_{x}(v)\big\}.
        \]
        By precondition on the join-operator $\eta_{i,j}$, no vertex of $G_x$ contains both labels $i$ and $j$ in~$\lab_y = \lab_x$ and we have $E_x \setminus E_y \subseteq \cE \subseteq E_x$ by definition of this operator.
        Now consider a multigraph~$A \in \aux(x)$ and let $(\cP, \phi)$ be a partial solution of $x$ with $A = \aux(\cP, \phi)$.
        Let~$q$ be the number of edges from $\cE$ used by the paths in $\cP$.
        Then we have $0 \leq q \leq n-1$ because the paths of ~$\cP$ form an acyclic subgraph of $G_x$.
        We prove by induction over the value~$q$ that $A$ belongs to $\aux(y) + q \cdot \{i,j\}$ which then implies the claim because~$\aux(y) + q \cdot \{i,j\} \subseteq \aux(y) + (n-1) \cdot \{i,j\}$ holds.
        If $q = 0$, then no edge of $\cP$ belongs to $E_x \setminus E_y$ so~$(\cP, \phi)$ is a partial solution of $y$ and therefore $A \in \aux(y) = \aux(y) + 0 \cdot \{i,j\}$.
        
        Otherwise, we have $q > 0$. Then we fix a path $P$ in $\cP$ that contains some edge $u_1 u_2 \in \cE$ with $i \in \lab_{x}(u_1)$ and $j \in \lab_{x}(u_2)$.
        Further let $P_1$ and $P_2$ be two paths (possibly of zero length) obtained from $P$ by the removal of the edge $u_1 u_2$ and let the vertices~$v_1$ and $v_2$ be such that~$v_1, u_1$ and $u_2,v_2$ are the endpoints of $P_1$ and $P_2$, respectively.
        Then $v_1$ and $v_2$ are the endpoints of $P$. 
        Finally, let the labels~$a, b \in [k]$ be such that $\phi(P) = \{a,b\}$,~$a \in \lab_{x}(v_1)$, and~$b \in \lab_{x}(v_2)$.
        Let $\cP'$ be obtained from $\cP$ by removing the path $P$ and adding the paths~$P_1$ and $P_2$. 
        Then~$\cP'$ is a path packing of $x$ using $q-1$ edges of $\cE$.
        Further, we define a label choice~$\phi'$ of~$\cP'$ in~$x$ so that $\phi'(P) = \phi(P)$ if $P$ belongs to $\cP$, and $\phi'(P_1) = \{a,i\}$, $\phi'(P_2) = \{j, b\}$.
        And let $A' = \aux(\cP', \phi') \in \aux(x)$.
        By induction, we have $A' \in \aux(y) + (q-1) \cdot \{i,j\}$.
        The multigraph $A$ is obtained from $A'$ by removing one edge with endpoints~$a$ and~$i$ and another edge with endpoints~$j$ and $b$ and adding a new edge with endpoints~$a$ and~$b$, i.e., we have~$A \in A' + \{i,j\}$.
        Altogether, we get that
        \[
            A \in \bigl(\aux(y) + (q-1) \cdot \{i,j\}\bigr) + \{i,j\} = \aux(y) + q \cdot \{i,j\}
        \]
        as claimed. This concludes the proof.
    \end{claimproof}
    So we indeed get the equality $\aux(y) + (n-1) \cdot \{i,j\} = \aux(x)$.
    
    Further, the inclusions $\cA_x^0 = \cA_y \subseteq \aux(y) \subseteq \aux(x)$ and the inductive application of \cref{claim:join-one-step-representation} as well as the fact that $\rreduce$ returns a subset of the input set imply that~$\cA_x \subseteq \aux(x)$ 
    holds.
    So we remain with the representation.
    For this, we prove the following claim:
    \begin{claim}\label{claim:join-a-b-one-step-representation}
        Let $\cA, \cB \subseteq \aux(x)$ be such that $\cA \lesssim \cB$.
        Then we have $\cA + \{i,j\}  \lesssim \cB + \{i,j\}$.
    \end{claim}
    \begin{claimproof}
        The proof goes analogously to a similar claim by Bergougnoux et al.~\cite{BergougnouxKK20}.
        However, their claim is about partial solutions while we only argue on the level of auxiliary graphs. 
        So we provide the complete proof for the sake of completeness.
        Consider an arbitrary red multigraph $A \in \cB + \{i,j\}$ and a blue multigraph $M$ such that the multigraph $A \uplus M$ admits a red-blue Eulerian trail, say $T$.
        The property $\cB \subseteq \aux(x)$ and \cref{claim:join-one-step-representation} imply that we have~$A \in \aux(x)$.
        Because the graph~$G_y$ is non-empty, the graph~$G_x$ is also non-empty and the multigraph~$A$ contains at least one edge. 
        Therefore, $M$ also contains at least one edge. 
        If $A$ belongs to $\cB$, then due to $\cA \lesssim \cB$, there exists a multigraph $A' \in \cA \subseteq \cA + \{i,j\}$ such that $A' \uplus M$ admits a red-blue Eulerian trail as desired.

        Otherwise, there exists a red multigraph $A^* \in \cB$ and two distinct edges $\big(e_1, \{a, i\}\big)$ and~$\big(e_2, \{j, b\}\big)$ of $A^*$ for some labels~$a, b \in [k]$ such that $A$ is obtained from $A^*$ by removing the edges $e_1$ and $e_2$ and adding a new edge $\big(e_3, \{a, b\}\big)$.
        Let the blue multigraph $M^*$ be obtained from $M$ by adding a new edge $\big(e^*, \{i, j\}\big)$.
        In particular, $M^*$ now contains at least two edges.
        The multigraph $A^* \uplus M^*$ admits a red-blue Eulerian trail obtained from $T$ by replacing the red edge $\big(e_3, \{a,b\}\big)$ of $A$ by a red-blue-red sequence $\big(e_1, \{a, i\}\big), \big(e^*, \{i,j\}\big), \big(e_2, \{j, b\}\big)$.
        Due to~$\cA \lesssim \cB$, there exists a red multigraph $\widetilde A \in \cA$ such that $\widetilde A \uplus M^*$ also admits a red-blue Eulerian trail, sat $\widetilde T$.
        Let $\big(\widetilde{e}_1, \{a', i\}\big)$ and $\big(\widetilde{e}_2, \{j, b'\}\big)$ be the edges of $\widetilde{A}$ preceding and following, respectively, in $\widetilde T$ the edge $(e^*, \{i,j\})$ of $M^*$  for some $a', b' \in [k]$.
        Because~$M^*$ contains at least two edges, the edges $\widetilde{e}_1$ and $\widetilde{e}_2$ are distinct. 
        So let the multigraph $A'$ be obtained from $\widetilde A$ by removing the edges $\big(\widetilde{e}_1, \{a', i\}\big)$ and $\big(\widetilde{e}_2, \{j, b'\}\big)$ and adding a new edge $\big(e', \{a', b'\}\big)$.
        First, we have $A' \in \cA + \{i, j\}$.
        And second, $A' \uplus M$ admits a red-blue Eulerian trail: to obtain it, we start with $\widetilde T$ and replace the red-blue-red sequence~$\big(\widetilde{e}_1, \{a', i\}\big), \big(e^*, \{i,j\}\big), \big(\widetilde{e}_2, \{j, b'\}\big)$ by a single red edge $\big(e', \{a',b'\}\big)$ of $A'$.
        This concludes the proof of $\cA + \{i,j\} \lesssim \cB + \{i,j\}$.
    \end{claimproof}
    Now we put everything together to prove that the family~$\cA_x$ represents~$\aux(x)$.
    First, we have~$\cA_x^0 = \cA_y \lesssim \aux(y) = \aux(y) + 0 \cdot \{i,j\}$ by assumption.
    Further, if $\cA^q_x \lesssim \aux(y) + q \cdot \{i,j\}$ holds for some integer~$q$, then by \cref{claim:join-a-b-one-step-representation}, we also have 
    \[
        \cA^q_x + \{i,j\} \lesssim \Bigl(\aux(y) + q \cdot \{i,j\}\Bigr) + \{i,j\} = \aux(y) + (q+1) \cdot \{i,j\}.
    \]
    By \cref{lemma:reduce-representation} (i.e., operator~$\rreduce$s returns a representative), we have
    \[
        \cA_x^{q+1} = \rreduce\big(\cA_x^q + \{i,j\}\big) \lesssim \cA_x^q + \{i,j\}. 
    \]
    Because $\lesssim$ is transitive, the last two relations imply
    \[
        \cA_x^{q+1} \lesssim \aux(y) + (q+1) \cdot \{i,j\}.
    \]
    So by induction over $q$, we obtain 
    \[
        \cA_x = \cA_x^{n-1} \lesssim \aux(y) + (n-1) \cdot \{i,j\}.
    \]
    With the equality~$\aux(y) + (n-1) \cdot \{i,j\} = \aux(x)$, the last claim of the lemma follows.
\end{proof}

Now before putting everything together, we provide a procedure to check, at the root of the multi-expression, for the existence of a Hamiltonian path between an arbitrary pair of adjacent vertices.
The statement and the proof are analogous to Bergougnoux et al.~\cite{BergougnouxKK20}:
\begin{lemma}\label{lem:hc-root-checking}
    Let $(G_x, \lab_x)$ be a multi-labeled graph, let $u, v$ be two vertices of $G_x$, and let~$\ell_u \neq \ell_v$ be two labels such that the following holds for every vertex $w$ of $G_x$:
    \[
        \lab_{x}(w) \cap \{\ell_u, \ell_v\} =
        \begin{cases}
            \{\ell_u\} & \text{ if } w = u \\
            \{\ell_v\} & \text{ if } w = v \\
            \emptyset & \text{ otherwise}
        \end{cases}
        .
    \]
    Further, let a subset $\cA_x \subseteq \aux(x)$ satisfy $\cA_x \lesssim \aux(x)$.
    Then the graph~$G_x$ contains a Hamiltonian path with endpoints $u$ and $v$ if and only if $\cA_x$ contains a multigraph whose edge set consists of only one edge and this edge has endpoints $\ell_u$ and $\ell_v$.
\end{lemma}
\begin{proof}
    First, suppose that the family~$\cA_x$ contains a multigraph $A$ whose edge set consists of only one edge and this edge has endpoints $\ell_u$ and $\ell_v$.
    Due to $\cA_x \subseteq \aux(x)$, we have~$A \in \aux(x)$.
    Let $(\cP, \phi) \in \Pi(x)$ be a partial solution of $x$ with $A = \aux(\cP, \phi)$.
    By the properties of $A$, the path packing $\cP$ of $x$ consists of a single path (which is therefore a Hamiltonian cycle of $G_x$) with endpoints $u'$ and $v'$ satisfying $\ell_u \in \lab_{x}(u')$ and $\ell_v \in \lab_{x}(v')$.
    By the preconditions of the lemma we get $u' = u$ and $v' = v$, i.e., $G_x$ contains a Hamiltonian path with endpoints $u$ and~$v$.

    For the other direction, suppose that $G_x$ contains a Hamiltonian path $P$ with endpoints~$u$ and~$v$.
    For a path packing $\cP$ of $x$ that consists of $P$ only and the label choice $\phi$ of $\cP$ in $x$ mapping $P$ to $\{\ell_u, \ell_v\}$, let $A = \aux(\cP, \phi) \in \aux(x)$.
    Then the red multigraph $A$ contains a single edge and its endpoints are $\ell_u$ and $\ell_v$.
    Now consider a blue multigraph $M$ whose edge set consists of a single edge with endpoints are $\ell_u$ and $\ell_v$.
    Then trivially $A \uplus M$ admits a red-blue Eulerian trail.
    Due to $\cA_x \lesssim \aux(x)$, the family~$\cA_x$ contains a red multigraph $A'$ such that $A' \uplus M$ also admits a red-blue Eulerian trail.
    The multigraph $M$ has a single edge so in fact,~$A$ is the unique red multigraph such that $A \uplus M$ admits a red-blue Eulerian trail.
    So we get $A = A' \in \cA_x$ as claimed concluding the proof.
\end{proof}

Now we are ready to provide the desired algorithm.
\thmHC*
\begin{proof}
    First, by \cref{lem:mcw-special-relabels} we may assume that every relabel-node either adds a label or forgets a label and every introduce-node creates a vertex with precisely one label.
    Furthermore, the number of nodes in the expression is polynomial in $n$. 
    If the graph $G$ contains at most two vertices, the instance is trivial.    
    Otherwise, for every edge $uv$ of $G$ we proceed as follows to check if $G$ admits a Hamiltonian cycle with endpoints $u$ and $v$.
    As observed by Bergougnoux et al.~\cite{BergougnouxKK20}, the graph $G$ admits a Hamiltonian cycle if and only if this is true for at least one edge $uv$: indeed, every Hamiltonian cycle is a union of some edge $uv$ and a Hamiltonian path between $u$ and $v$.
    So let now $uv$ be a fixed edge of $G$.
    Let $\ell_u = k+1$ and $\ell_v = k+2$ be two new labels. 
    We slightly adapt the provided multi-expression as follows.
    Let $u\langle i_u\rangle $ be the node introducing the vertex~$u$ (for some label~$i_u \in [k]$).
    We replace it by an introduce-node~$u\langle \ell_u\rangle$ followed by a relabel-node $\rho_{\ell_u \to \{\ell_u, i_u\}}$.
    Similarly, we replace the node $v\langle i_v\rangle $ introducing the vertex~$v$ (for some label~$i_v \in [k]$) by $v\langle \ell_v\rangle$ followed by $\rho_{\ell_v \to \{\ell_v, i_v\}}$.
    The resulting expression is a multi-$(k+2)$-expression of the same graph $G$.
    And for the label sets in $\lab_r$ at the root~$r$ of this expression, we have
    \[
        \lab_{r}(w) \cap \{\ell_u, \ell_v\} =
        \begin{cases}
            \{\ell_u\} & \text{ if } w = u \\
            \{\ell_v\} & \text{ if } w = v \\
            \emptyset & \text{ otherwise}
        \end{cases}
    \]
    for every vertex $w$ of the graph~$G_r = G$, i.e., the precondition of \cref{lem:hc-root-checking} is satisfied.
    Furthermore, the new multi-expression still satisfies the properties stated in \cref{lem:mcw-special-relabels} and the number of nodes in the expression increased by two, i.e., remains polynomial in $n$.
    We define~$k' = k+2$ to be the number of labels used by this expression.
    All multigraphs below thus have the vertex set $[k']$.
    By \cref{lem:hc-root-checking}, it suffices to compute a set $\cA \subseteq \aux(r)$ with~$\cA \lesssim \aux(r)$.

    For this, we process the constructed multi-$k'$-expression of $G$ in a bottom-up manner and for every node $x$, we compute a set $\cA_x \subseteq \aux(x)$ satisfying $\cA_x \subseteq \aux(x)$ whose size is upper-bounded as $\abs*{\cA_x} \leq n^{k'} \cdot 2^{k' \cdot (\log(k')+1)} \in n^{\cO(k)}$.
    Furthermore, we argue that $\cA_x$ can be computed in time $n^{\cO(k')} = n^{\cO(k)}$ for any node $x$.
    
    If $x$ is a leaf-node, we compute the family $\cA_x$ of size one using \cref{lem:hc-leaf}.
    This family can clearly be computed in time polynomial in $n$.

    If $x$ is a $\rho_{i \to \emptyset}$-node for some label~$i \in [k']$ with a child $y$, we compute a family $\cA_x$ from a family $\cA_y$ using \cref{lem:hc-forget-label}. 
    By construction, it then holds that $\abs*{\cA_x} \leq \abs*{\cA_y}$ so the desired upper-bound on the cardinality of $\abs*{\cA_x}$ still applies.
    Furthermore, this family can trivially computed from $\abs*{\cA_x}$ in time $n^{\cO(k')}$ by checking all elements of $\cA_y$.

    If $x$ is a $\rho_{i \to \{i,j\}}$-node for some labels~$i \neq j \in [k']$ with a child $y$, we first compute the set~$\cA_x' = \bigcup_{A \in \cA_y} A^{i \to j}$ from~$\cA_y$.
    By~\cref{lem:hc-add-label}, the set $\cA_x'$ satisfies $\cA_x' \subseteq \aux(x)$ and~$\cA'_x \lesssim \aux(x)$.
    Then we compute $\cA_x = \rreduce(\cA_x')$. 
    By \cref{lemma:reduce-representation} and transitivity of~$\lesssim$, the family $\cA_x$ then also represents the family~$\aux(x)$.
    And we also have
    \[
        \cA_x = \rreduce(\cA'_x) \subseteq \cA'_x \subseteq \aux(x).
    \]
    The application of $\rreduce$ already ensures the bound on the size of $\cA_x$.
    We now bound the running time of this computation.
    For labels~$a, b \in [k']$ and a multigraph $A$, we define the value~$\#_A^{\{a,b\}}$ as the multiplicity of $\{a,b\}$ in $A$.
    Now fix an arbitrary multigraph $A \in \aux(y)$.
    Every multigraph is uniquely determined by the multiplicity of every pair $a, b \in [k']$ of labels.
    So by definition, the family~$A^{i \to j}$ consists of all multigraphs~$A'$ for which there exist integer values $q_1, \dots, q_{i-1}, q_{i+1}, \dots, q_{k'}, q_{\operatorname{loop}}^1, q_{\operatorname{loop}}^2 \in \bN$ such that
    \begin{align*}
        q_a \leq \#_A^{\{a, i\}} && q_{\operatorname{loop}}^1 + q_{\operatorname{loop}}^2 \leq \#_A^{\{i\}}
    \end{align*} 
    hold for every label $a \in [k'] \setminus \{i\}$ and for every pair $\{a, b\} \subseteq [k']$ of labels, we have
    \begin{itemize}
        \item $\#_{A'}^{\{a, b\}} = \#_{A}^{\{a, b\}}$ if $\{a, b\} \cap \{i, j\} = \emptyset$---because none of the labels $i$ and $j$ is involved, the multiplicity does not change from $A$ to $A'$,
        \item $\#_{A'}^{\{a, i\}} = \#_{A}^{\{a, i\}} - q_a$ and $\#_{A'}^{\{a, j\}} = \#_{A}^{\{a, j\}} + q_a$ if $a \notin \{i, j\}$---this reflects that $q_a$ edges with endpoints $a$ and $i$ become edges with endpoints $a$ and $j$,
        \item $\#_{A'}^{\{i\}} = \#_{A}^{\{i\}} - q_{\operatorname{loop}}^1 - q_{\operatorname{loop}}^2$---this reflects that $q_{\operatorname{loop}}^1$ loops at $i$ become edges with endpoints~$i$ and~$j$ and $q_{\operatorname{loop}}^2$ loops at $i$ become loops at $j$,
        \item $\#_{A'}^{\{i, j\}} = \#_{A}^{\{i, j\}} - q_j + q_{\operatorname{loop}}^1$---reflects that $q_j$ edges with endpoints $i$ and~$j$ become loops at~$j$ and $q_{\operatorname{loop}}^1$ loops at $i$ become edges with endpoints $i$ and~$j$ matching the previous item, 
        \item $\#_{A'}^{\{j\}} = \#_{A}^{\{j\}} + q_{\operatorname{loop}}^2 + q_j$---
        reflects that $q_{\operatorname{loop}}^2$ loops at $i$ become loops at $j$ and $q_j$ edges with endpoints $i$ and $j$ become loops at $j$ matching the previous two items.
    \end{itemize}    
    Observe that the multigraph~$A$ in $\aux(y)$ has at most $n$ edges, this is simply because every path packing of $y$ contains at most $n$ paths.
    Thus, each of the values $q_1, \dots, q_{i-1}, q_{i+1}, \dots, q_{k'}$, $q_{\operatorname{loop}}^1, q_{\operatorname{loop}}^2$ satisfying the above properties is upper-bounded by $n$.
    So our algorithms considers all $n^{\cO(k')}$ options $q_1, \dots, q_{i-1}, q_{i+1}, \dots, q_{k'}, q_{\operatorname{loop}}^1, q_{\operatorname{loop}}^2$ and for each of them, in time polynomial in $n$, it constructs the corresponding multigraph $A' \in A^{i \to j}$ determined by the above five items.
    Iterating over all multigraphs $A \in \cA_y$ yields the total running time of $\abs*{\cA_y} \cdot n^{\cO(k')} = n^{\cO(k')}$.
    In particular, the cardinality of $\cA'_x$ is at most $n^{\cO(k')}$.
    By \cref{lem:size-of-repr}, the set $\cA_x = \rreduce(\cA'_x)$ is then computed in time $n^{\cO(k')}$.

    Now let $x$ be an $\oplus$-node with children $y_1$ and $y_2$.
    We compute a family $\cA'_x$ from the families $\cA_{y_1}$ and $\cA_{y_2}$ as defined in \cref{lem:hc-union-node} and then apply $\rreduce$ to $\cA_x'$ to obtain $\cA_x$.
    The application of $\rreduce$ ensures that the size of $\cA_x$ is bounded as desired and representation is maintained.
    We now bound the running time required for this computation.
    For every pair~$A_1 \in \cA_{y_1}$ and $A_2 \in \cA_{y_2}$, both of those multigraphs have at most $n$ edges so we can compute the edge-disjoint union of $A_1$ and $A_2$ in time polynomial in $n$.
    Using the bound on the sizes of $\cA_{y_1}$ and $\cA_{y_2}$, all such pairs can be computed in time $n^{\cO(k')}$ and the resulting family has cardinality $n^{\cO(k')}$. 
    Finally, by \cref{lem:size-of-repr} the application of $\rreduce$ to this family takes time $n^{\cO(k')}$. 

    Finally, let $x$ be an $\eta_{i,j}$-node with a child $y$, we compute the family $\cA_x$ from $\cA_y$ using \cref{lem:hc-join-node} as follows.
    Note that by definition of a multi-expression, the graph $G_y$ contains no node containing both labels $i$ and $j$ in $\lab_y$ so the precondition of \cref{lem:hc-join-node} is satisfied.
    We start with~$\cA_x^0 = \cA_y$.
    And then for every index~$q \in [n-1]$, we compute the family~$\cA_x^q = \rreduce\big(\cA_x^{q-1} + \{i, j\}\big)$.
    In the end we return $\cA_x = \cA_x^{n-1}$.
    The application of $\rreduce$ in each step ensures that the size of $\cA^q_x$ is bounded by $n^{k'} \cdot 2^{k'\cdot (\log(k')+1)} \in n^{\cO(k')}$ for every $q \in [n-1]$ (cf.~\cref{lem:size-of-repr}).
    In particular, this implies that the size of $\cA_x$ is bounded by this value as desired. 
    So it remains to prove that the computation can be carried out in $n^{\cO(k')}$.
    For this, it suffices to show that for every $q \in [n-1]$, the family $\cA_x^q$ can be computed from $\cA_x^{q-1}$ in time~$n^{\cO(k')}$. 
    So let $q \in [n-1]$ be fixed.
    First, consider an arbitrary multigraph $A \in \cA_x^{q-1}$.
    Because no path packing of $x$ contains more than~$n$ paths, the multigraph $A$ contains at most~$n$ edges.
    Thus, we can construct the family~$A + \{i,j\}$ by considering at most $\cO(n^2)$ pairs of distinct edges of $A$, checking whether one of them has an endpoint in $i$ and the other has an endpoint is $j$, and in the affirmative case replacing the two edges by a corresponding single one.
    For a fixed multigraph $A\in \cA_x^{q-1}$, this can trivially be done in time polynomial in $n$.
    In particular, the size of $A + \{i,j\}$ is bounded by $\cO(n^2)$.
    Doing so for every $A \in \cA_x^{q-1}$, we obtain the family $\cA_x^{q-1}+\{i,j\}$ of size~$\cO\big(\abs*{\cA_x^{q-1}}  \cdot n^2\big) \in n^{\cO(k')}$ that can be computed in time $\abs*{\cA_x^{q-1}} \cdot n^{\cO(1)}$.
    Now the family $\cA_x^q$ is computed by applying $\rreduce$ to $\cA_x^{q-1}+\{i,j\}$.
    By \cref{lem:size-of-repr}, the family $\cA_x^q$ is computed in time $\cO\big(\abs*{\cA_x^{q-1}} \cdot n^2 \cdot n \cdot (k')^2 \cdot \log(n \cdot k')\big) \in n^{\cO(k')}$ as claimed.

    The number of nodes in the multi-$k'$-expression of $G$ is polynomial in $n$ as argued in the beginning.
    The above arguments imply that on each node of the expression we spend time~$n^{\cO(k')} = n^{\cO(k)}$.
    Altogether, for the root-node $r$ of the constructed multi-$k'$-expression, we construct the desired family $\cA_r$ in time $n^{\cO(k)}$.
    Repeating the process for every edge $uv$ of $G$ increases the running time by a factor of $\cO(n^2)$ concluding the proof.
\end{proof} 

Fomin et al.~\cite{FominGLSZ19} showed the following lower bound for clique-width:
\begin{theorem}[Theorem 2 in \cite{FominGLSZ19}]
    Let $G$ be an $n$-vertex graph given together with a $k$-expression of $G$. 
    Then the \textsc{Hamiltonian Cycle} problem cannot be solved in time $f(k) \cdot n^{o(k)}$ for any computable function $f$ unless the ETH fails.
\end{theorem}

Since any $k$-expression of a graph 
can trivially be transformed into a multi-$k$-expression of the same graph~\cite{Furer17}, the lower bound transfers to multi-clique-width showing that our algorithm is tight under ETH.
\begin{theorem}
    Let $G$ be an $n$-vertex graph given together with a multi-$k$-expression of~$G$. 
    Then the \textsc{Hamiltonian Cycle} problem cannot be solved in time $f(k) \cdot n^{o(k)}$ for any computable function $f$ unless the ETH fails.
\end{theorem}

\section{Algorithm for Edge Dominating Set}
\label{sec:eds}
For a graph $G$, a set $Q \subseteq E(G)$ of edges is called an \emph{edge dominating set} of $G$ if the induced subgraph~$G\big[V(G) \setminus V(Q)\big]$ is edgeless, i.e., every edge of $G$ shares at least one endpoint with some edge in $Q$, i.e., the set $V(Q)$ of endpoints of $Q$ is a vertex cover of $G$.
In this section, we show how provided a multi-$k$-expression creating a graph $G$ on $n$ vertices and an integer~$\cardinality$, in time $n^{\cO(k)}$ we can decide whether $G$ admits an edge dominating set of size at most $\cardinality$, i.e., solve the \textsc{Edge Dominating Set} problem.
To simplify the following definitions, from now on the graph~$G$ on~$n$ nodes and a multi-$k$-expression of $G$ are fixed.
We will rely on the following folklore equivalent formulation of the problem.
We provide a proof for the sake of completeness.
\begin{lemma}\label{lem:eds-equivalent-formulation}
    A graph $G$ admits an edge dominating set of size at most $\cardinality$ if and only if there exists a pair $(S, M)$ with the following properties:
    \begin{enumerate}
        \item the set $S$ is a vertex cover of $G$,
        \item ths set $M$ is a matching in $G$ satisfying $V(M) \subseteq S$,
        \item\label{item:vertex-cover-matching-cardinality} and we have $\abs*{M} + \abs*{S \setminus V(M)} \leq \cardinality$.
    \end{enumerate}
\end{lemma}

\begin{figure}[h]
    \includegraphics{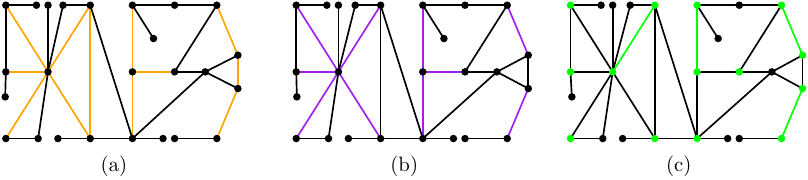}
    \centering
    \caption{(a) An orange edge dominating set of a graph. (b) A purple minimal edge dominating set $Q$. (c) Green vertex cover $V(Q)$ and green matching $M$ with one edge per connected component of $G[Q]$ satisfy $\abs*{Q} = \abs*{V(Q)} - \abs*{M}$.}
    \label{fig:eds-equivalence}
\end{figure}

\begin{proof}
    First, let $Q$ be an edge dominating set of $G$ with $\abs*{Q} \leq \cardinality$ (see \cref{fig:eds-equivalence} for an illustration).
    We may assume that~$Q$ is an inclusion-minimal edge dominating set of~$G$.
    Then the subgraph $G[Q]$ induced by $Q$ forms a vertex-disjoint union of stars on at least two vertices each.
    Indeed, if this subgraph contains a simple path of length three, say on vertices~$u$,~$v$,~$w$,~$y$, then the set~$Q \setminus \{vw\}$ is still an edge dominating set; and if it contains a triangle, say on vertices~$u$,~$v$,~$w$, then the set~$Q \setminus \{vw\}$ is still an \eds contradicting the minimality in both cases.
    So let $C$ denote the set of connected components of the graph~$G[Q]$.
    Further, let $S = V(Q)$ be the set of endpoints of the edges in~$Q$ and let~$M$ be a set containing one edge from every connected component in $C$.
    Then $S$ is a vertex cover of $G$ and $M$ is a matching whose endpoints all belong to $S$.
    We have
    \[
        \abs*{M} + \abs*{S \setminus V(M)} = \abs*{M} + \big(\abs*{S} - 2\abs*{M}\big) = \abs*{S} - \abs*{M} = \abs*{S} - \abs*{C} = \abs*{Q} \leq \cardinality.
    \]
    Here, the penultimate equality holds since the forest $G[Q]$ with $\abs*{S}$ vertices and $\abs*{C}$ connected components contains $\abs*{S} - \abs*{C}$ edges.
    This proves one direction of the claim.

    Now let $(S, M)$ be a pair satisfying the three items of the lemma.
    We may assume that $S$ contains no vertex isolated in $G$ as such vertices can be discarded from $S$ without violating any of the three properties.
    For every vertex $v \in S \setminus V(M)$, choose an arbitrary edge $e_v \in E(G)$ incident with $v$.
    We define the set $Q$ as 
    \[
        Q = M \cup \big\{e_v \mid v \in S \setminus V(M)\big\}.
    \]
    First, the set $Q$ is an \eds of $G$ because the set $V(Q)$ of its endpoints contains the vertex cover $S$ as a subset.
    Second, we have
    \[
        \abs*{Q} = \abs*{M} + \abs*{\big\{e_v \mid v \in S \setminus V(M)\big\}} \leq \abs*{M} + \abs*{S \setminus V(M)} \leq \cardinality
    \]
    where the last equality holds by Item~\ref{item:vertex-cover-matching-cardinality} concluding the proof of the other direction.
\end{proof}

Relying on this equivalence, we now define the partial solutions.
For a node $x$ of the expression, a \emph{partial solution} of $x$ is a pair $(S, M)$ such that $S$ is a vertex cover of the graph~$G_x$ and $M$ is a matching of $G_x$ satisfying $V(M) \subseteq S$.

The crucial idea behind bounding the number of relevant footprints (defined in a moment) is that for every vertex in the set~$S \setminus V(M)$, it suffices to ``guess'' one label that will later lead to the creation of an incident edge added to $M$ (if such an edge will exist at all).
We formalize this idea using label choices.
A \emph{label choice} for the partial solution $(S, M)$ of $x$ is a mapping $\phi \colon S \setminus V(M) \to [k]$ such that $\phi(v) \in \lab_x(v)$ for every vertex $v \in S \setminus V(M)$.
Thus,~$\phi$ maps every vertex of $S$ ``not matched'' by $M$ to one of its labels.
Note that if~$S \setminus V(M)$ contains a vertex whose label set is empty in $\lab_x$, then no label choice for $(S, M)$ exists.

For a partial solution $(S, M)$ of $x$ and a label choice $\phi$ of $(S, M)$, the \emph{footprint} (in $x$) is the triple $(I, \psi, \ell)$ such that the following properties apply.
\begin{enumerate}
    \item The set $I \subseteq [k]$ with $I = \bigcup_{v \in V_x \setminus S} \lab_x(v)$ remembers all labels of vertices in $V_x \setminus S$. It is used to ensure at join-nodes that $V_x \setminus S$ remains an independent set.
    \item The mapping $\psi \colon [k] \to [n]_0$ with $\psi(i) = \abs*{\phi^{-1}(i)}$ for all $i \in [k]$ remembers how many vertices are mapped to each label from $[k]$.
    First, $\psi$ is used to create matching edges at join-nodes and second, it is relevant to keep track of the size of the corresponding edge dominating set.
    \item And the value $\ell \in \bN$ with $\abs*{M} = \ell$ remembers the cardinality of the matching.
\end{enumerate}
Note that even though a partial solution may have multiple footprints, the footprint becomes unique once we fix the label choice.
The set $\cS_x$ of \emph{footprints of $x$} is defined as the set of footprints of all partial solutions of $x$ and all their label choices.
Our algorithm will then follow from the following five lemmas which provide the computation of footprints for every type of nodes.
For a leaf-node, the arising graph consists of a single node so we can compute the footprints directly.
\begin{lemma}\label{lem:eds-leaf}
    Let $(G_x, \lab_x) = v\langle i \rangle$ for some vertex $v$ and some label $i \in [k]$.
    Then there are precisely two footprints of the node~$x$, namely $\cS_x = \Bigl\{\big(\emptyset, \psi_1, 0\big), \big(\{i\}, \psi_2, 0\big)\Bigr\}$ with
    \begin{enumerate}
        \item $\psi_1(i) = 1$ and $\psi_1(j) = 0$ for all $j \neq i \in [k]$,
        \item and $\psi_2(j) = 0$ for all $j \in [k]$.
    \end{enumerate}
\end{lemma}
\begin{proof}
    The graph $G_x$ contains only one vertex whose label set is $\{i\}$ and it contains no edges. 
    Thus, there are precisely two partial solutions of $x$.
    First, it is $\big(\{v\}, \emptyset\big)$ with the unique label choice mapping $v$ to $i$.
    Since $V_x = \{v\}$, there are no vertices in $V_x \setminus \{v\}$. 
    Thus, the only footprint of this partial solution is $(\emptyset, \psi_1, 0)$.
    And second, the partial solution $(\emptyset, \emptyset)$ whose unique label choice is the empty mapping.
    The only vertex in $V_x \setminus \emptyset$ is $v$.
    So the only footprint of this partial solution is $\big(\{i\}, \psi_2, 0)$.
\end{proof}

When a node of the expression forgets a label $i$, the label choices mapping some vertex to $i$ are ``invalidated'' so we discard the corresponding footprints.
\begin{lemma}\label{lem:eds-remove-label}
    Let $(G_x, \lab_x) = \rho_{i \to \emptyset}(G_y, \lab_y)$ for some label $i \in [k]$ and a multi-$k$-labeled graph~$(G_y, \lab_y)$.
    Then the footprints of the node~$x$ are given as follows:
    \[
        \cS_x = \Bigl\{\big(I \setminus \{i\}, \psi, \ell\big) \mid (I, \psi, \ell) \in \cS_y, \psi(i) = 0\Bigr\}.
    \]
\end{lemma}
\begin{proof}
    The unlabeled graphs $G_x$ and $G_y$ are equal and only the label sets differ so that we have $\lab_x(v) = \lab_y(v) \setminus \{i\}$ for every vertex $v \in V_x = V_y$.
    Hence, the sets of partial solutions of~$x$ and $y$ are equal and only the label choices of those partial solutions differ.

    So let first $(I, \psi, \ell) \in \cS_y$ be a footprint of $y$ satisfying $\psi(i) = 0$.
    Then there exist a partial solution $(S, M)$ and a label choice $\phi$ of $y$ with this footprint. 
    In particular, we then have~$\phi^{-1}(i) = \emptyset$.
    The pair $(S, M)$ is a partial solution of $x$ as well.
    Further, for every vertex~$v \in S \setminus V(M)$, the properties $\phi(v) \in \lab_y(v)$ and $\phi(v) \neq i$ imply that we also have~$\phi(v) \in \lab_x(v)$.
    So the mapping $\phi$ is also a label choice of $(S, M)$ in $x$.
    Finally, we have
    \[  
        \bigcup_{v \in V_x \setminus S} \lab_x(v) = \Bigl(\bigcup_{v \in V_y \setminus S} \lab_y(v)\Bigr) \setminus \{i\} = I \setminus \{i\}.
    \]
    So the tuple~$\big(I \setminus \{i\}, \psi, \ell\big)$ is indeed a footprint of $x$, i.e.,~$\big(I \setminus \{i\}, \psi, \ell\big) \in \cS_x$.

    For the other direction, let $(\widetilde I, \widetilde \psi, \widetilde \ell) \in \cS_x$ be a footprint of $x$.
    And let $(S, M)$ be a partial solution and let $\phi$ be a label choice of $x$ with this footprint.
    Then $(S, M)$ is a partial solution of $y$ and $\phi$ its label choice in $y$: indeed, we have $\phi(v) \in \lab_x(v) \subseteq \lab_y(v)$ for every vertex~$v \in S \setminus V(M)$.
    Because no vertex in $G_x$ contains label $i$, we have $\phi^{-1}(i) = \emptyset$ and therefore,~$\widetilde \psi(i) = 0$.
    So let $(I, \psi, \ell) \in \cS_y$ be the footprint of $(S, M)$ and $\phi$ in $y$.
    We then have~$\widetilde \psi = \psi$ and~$\widetilde \ell = \ell$ simply by definition.
    We also have $\widetilde I = \bigcup_{v \in  V_x \setminus S} \lab_x(v)$ and~$I = \bigcup_{v \in  V_x \setminus S} \lab_y(v)$.
    Then by definition of the graph $G_x$, we have $\widetilde I = I \setminus \{i\}$.
\end{proof}

When a node of the expression adds a label $i$ to a label $j$, a label choice may now decide, for every vertex mapped to $i$ whether it stays mapped to $i$ or ``switches'' to $j$.
Since for the footprints only the number of such nodes matters, it suffices to branch on the number of vertices, for which such a switch occurs.
\begin{lemma}\label{lem:eds-add-label}
    Let $(G_x, \lab_x) = \rho_{i \to \{i, j\}}(G_y, \lab_x)$ for some labels~$i \neq j \in [k]$ and a multi-$k$-labeled graph~$(G_y, \lab_y)$.
    Then the footprints of the node $x$ are given as follows:
    \begin{align*}
        \cS_x = \Biggl\{\Big(\widetilde I, \psi\big[i \mapsto \psi(i) - r, j \mapsto \psi(j) + r\big], \ell\Big) \mid & (I, \psi, \ell) \in \cS_y, \widetilde I =
        \begin{cases}
            I & \text{if } i \notin I \\     
            I \cup \{j\} & \text{if } i \in I,
        \end{cases} \\
        &r \in \big[\psi(i)\big]_0 
        \Biggr\}.
    \end{align*}
\end{lemma}
\begin{proof}
    The unlabeled graphs $G_x$ and $G_y$ are equal and only the label sets differ so that we have 
    \[
        \lab_x(v) = 
        \begin{cases}
            \lab_y(v) \cup \{j\} & \text{if } i \in \lab_y(v) \\
            \lab_y(v) & \text{otherwise}
        \end{cases} 
    \]
    for every vertex $v \in V_x = V_y$.
    Hence, the sets of partial solutions of $x$ and $y$ are equal and only the label choices of those partial solutions differ.
    
    First, let $(I, \psi, \ell) \in \cS_y$ be a footprint of $y$ and let $r \in \big[\psi(i)\big]_0$ be an integer.
    Further, let~$(S, M)$ be a partial solution and let $\phi$ be a label choice of $(S, M)$ in $y$ with this footprint.
    In particular, we have $\abs*{\phi^{-1}(i)} = \psi(i) \geq r$.
    Let now $v_1, \dots, v_r$ be $r$ pairwise distinct vertices from $\phi^{-1}(i)$.
    Then $(S, M)$ is a partial solution of $x$ as well and $\widetilde \phi := \phi[v_1 \mapsto j, \dots, v_r \mapsto j]$ is a label choice for $(S, M)$ in $x$.
    Indeed, for every vertex $v \in S \setminus \big(V(M) \cup \{v_1, \dots, v_r\}\big)$, we have~$\widetilde\phi(v) = \phi(v) \in \lab_y(v) \subseteq \lab_x(v)$, and we also have $\widetilde \phi(v_q) = j \in \lab_x(v_q)$ due to~$i = \phi(v_q) \in \lab_y(v_q)$ for every $q \in [r]$.
    Let $(\widetilde I, \widetilde \psi, \widetilde \ell) \in \cS_x$ be the footprint of $(S, M)$ and $\widetilde \phi$ in $x$.
    We trivially have~$\ell = \widetilde \ell$.
    Further, we have $\widetilde \psi = \psi\big[i \mapsto \psi(i) - r, j \mapsto \psi(j) + r\big]$ by construction.
    Finally, we have $I = \bigcup_{v \in V_x \setminus S} \lab_y(v)$ and $\widetilde I = \bigcup_{v \in V_x \setminus S} \lab_x(v)$.
    By definition of $\lab_x$, we then have $\widetilde I =
    \begin{cases}
        I & \text{if } i \notin I \\     
        I \cup \{j\} & \text{if } i \in I
    \end{cases}$.

    For the other direction, let $(\widetilde I, \widetilde \psi, \widetilde \ell) \in \cS_x$ be a footprint of $x$.
    And let $(S, M)$ be a partial solution and $\widetilde \phi$ a label choice of $(S, M)$ in $x$ with this footprint.
    Let $r$ be the cardinality of the set $\big\{v \in \widetilde \phi^{-1}(j) \mid i \in \lab_y(v)\big\}$ and let $v_1, \dots, v_r$ be the elements of this set.
    So the vertices~$v_1, \dots, v_r$ are precisely the vertices mapped to $j$ by $\phi$ such that their label set in $\lab_y$ (in fact, in both $\lab_y$ and $\lab_x$) contains label $i$.
    We define the mapping~$\phi = \widetilde\phi[v_1 \mapsto i, \dots, v_r \mapsto i]$.
    The pair $(S, M)$ is also a partial solution of $y$ and we now argue that $\phi$ is a label choice for~$(S, M)$ in $y$.
    For this consider an arbitrary vertex $v \in S \setminus V(M)$.
    \begin{itemize}
        \item If~$\widetilde \phi(v) \neq j$ (and in particular,~$v \notin \{v_1, \dots, v_r\}$), then~$j \neq \widetilde \phi(v) \in \lab_x(v)$ implies $ \phi(v) = \widetilde \phi(v) \in \lab_y(v)$---indeed, from $\lab_y$ to $\lab_x$ the label set of a vertex can only gain label~$j$.
        \item If $\widetilde \phi(v) = j$ and $i \notin \lab_y(v)$ (and in particular,~$v \notin \{v_1, \dots, v_r\}$), then~$\widetilde \phi(v) \in \lab_x(v)$ implies $\phi(v) = \widetilde \phi(v) \in \lab_y(v)$---indeed, the label set of a vertex can change from $\lab_y$ to~$\lab_x$ only if it contained label~$i$.
        \item Finally, if $\widetilde \phi(v) = j$ and $i \in \lab_y(v)$ (i.e., if $v \in \{v_1, \dots, v_r\}$), then we have $\phi(v) = i \in \lab_y(v)$.
    \end{itemize}
    So $\phi$ is indeed a label choice for the partial solution $(S, M)$ in $y$.
    Let now $(I, \psi, \ell) \in \cS_y$ be the footprint of $(S, M)$ and $\phi$ in $y$.
    We then have $\widetilde I = \bigcup_{v \in  V_x \setminus S} \lab_x(v)$ and $I = \bigcup_{v \in  V_x \setminus S} \lab_y(v)$.
    By definition of $\lab_x$, we then have 
    \[
        \widetilde I =
        \begin{cases}
            I & \text{if } i \notin I \\     
            I \cup \{j\} & \text{if } i \in I
        \end{cases}.
    \]
    Next, by construction, we have~$\widetilde \psi = \psi\big[i \mapsto \psi(i) - r, j \mapsto \psi(j) + r\big]$.
    In particular, we have~$0 \leq r \leq \psi(i)$, i.e., $r \in \big[\psi(i)\big]_0$.
    And we also trivially have~$\widetilde \ell = \ell$.
\end{proof}

For a disjoint union of two graphs, the set of partial solutions is simply given by pairwise unions of partial solutions of the two graphs.
And the analogous property applies to label choices.
So for the footprints, we can combine all such pairs as follows.
\begin{lemma}\label{lem:eds-union}
    Let $(G_x, \lab_x) = (G_{y_1}, \lab_{y_1}) \oplus (G_{y_2}, \lab_{y_2})$ where $G_{y_1}$ and $G_{y_2}$ are vertex-disjoint graphs.
    Then the footprints of the node $x$ are given as follows:
    \[
        \cS_x = \Bigl\{\big(I_1 \cup I_2, i \mapsto \psi_1(i) + \psi_2(i), \ell_1 + \ell_2\big) \mid (I_1, \psi_1, \ell_1) \in \cS_{y_1}, (I_2, \psi_2, \ell_2) \in \cS_{y_2} \Bigr\}.
    \]
\end{lemma}
\begin{proof}
        First, let $(I_1, \psi_1, \ell_1) \in \cS_{y_1}$ be a footprint of $y_1$ and let $(I_2, \psi_2, \ell_2) \in \cS_{y_2}$ be a footprint of $y_2$.
        And for every index~$a \in [2]$, let $(S_a, M_a)$ be a partial solution of $y_a$ and let~$\phi_a$ be a label choice of $(S_a, M_a)$ in $y_a$ with footprint $(I_a, \psi_a, \ell_a)$.
        Then $(S_1 \cup S_2, M_1 \cup M_2)$ is a partial solution of $x$: because $G_x$ is a vertex-disjoint union of $G_{y_1}$ and $G_{y_2}$, the set~$S_1 \cup S_2$ is a vertex cover of $G_x$ and the set $M_1 \cup M_2$ is a matching in $G_x$.
        Furthermore, since the vertices from $G_{y_1}$ and $G_{y_2}$ preserve their labels from $\lab_{y_1}$ and $\lab_{y_2}$, respectively, in $\lab_x$, the mapping~$\phi_1 \cup \phi_2$ is a label choice for $(S_1 \cup S_2, M_1 \cup M_2)$ in $x$. 
        Let $(\widetilde I, \widetilde \psi, \widetilde \ell) \in \cS_x$ be the footprint of~$(S_1 \cup S_2, M_1 \cup M_2)$ and $\phi_1 \cup \phi_2$ in $x$.
        First, we have 
        \[
            \widetilde I = \bigcup_{v \in V_x \setminus (S_1 \cup S_2)} \lab_x(v) = \left(\bigcup_{v \in V_{y_1} \setminus S_1} \lab_{y_1}(v)\right) \cup \left(\bigcup_{v \in V_{y_2} \setminus S_2} \lab_{y_2}(v)\right) = I_1 \cup I_2
        \]
        because $S_1 \subseteq V_{y_1}$, $S_2 \subseteq V_{y_2}$, and the vertex sets $V_{y_1}$ and $V_{y_2}$ are disjoint. 
        Second, we have~$\widetilde \ell = \ell_1 + \ell_2$ because $M_1$ and $M_2$ are disjoint.
        And finally, by construction we have 
        \[
            \widetilde \psi(i) = \abs*{(\phi_1 \cup \phi_2)^{-1}(i)} = \abs*{\phi_1^{-1}(i)} + \abs*{\phi_2^{-1}(i)} = \psi_1(i) + \psi_2(i)
        \]
        for every $i \in [k]$.

        For the other direction, consider a footprint $(\widetilde I, \widetilde \psi, \widetilde \ell) \in \cS_x$ of $x$.
        And let $(\widetilde S, \widetilde M)$ be a partial solution of $x$ and let $\widetilde \phi$ be a label choice of $(\widetilde S, \widetilde M)$ in $x$ with this footprint.
        For every index~$a \in [2]$, let $S_a = S \cap V_{y_a}$, $M_a = M \cap E_{y_a}$, and $\phi_a = \widetilde\phi_{|_{S_a \setminus V(M_a)}}$.
        Because the graphs~$G_{y_1}$ and $G_{y_2}$ are vertex-disjoint, the set~$S_a$ is a vertex cover of $G_{y_a}$ and the set~$M_a$ is a matching in $G_{y_a}$ with $V(M_a) \subseteq S_a$ for every $a \in [2]$, i.e., $(S_a, M_a)$ is a partial solution of the node~$y_a$. 
        And since the labels of the vertices from $G_{y_1}$ and $G_{y_2}$ are preserved from $\lab_{y_1}$ and $\lab_{y_2}$, respectively, in~$\lab_x$, the mapping $\phi_a$ is a label choice for the partial solution~$(S_a, M_a)$ in $y_a$ for every index~$a \in [2]$.
        So let $(I_a, \psi_a, \ell_a) \in \cS_{y_a}$ be the footprint of $(S_a, M_a)$ and $\phi_a$ in~$y_a$.
        The following properties hold because $G_x$ is a vertex-disjoint union of the graphs $G_{y_1}$ and $G_{y_2}$.
        First, we have
        \[
            \widetilde I = \bigcup_{v \in V_x \setminus \widetilde S} \lab_x(v) = \left(\bigcup_{v \in V_{y_1} \setminus S_1} \lab_{y_1}(v)\right) \cup \left(\bigcup_{v \in V_{y_2} \setminus S_2} \lab_{y_2}(v)\right) = I_1 \cup I_2.
        \]
        Second, 
        \[
            S \setminus V(M) = \big(S_1 \setminus V(M_1)\big) \cup \big(S_2 \setminus V(M_2)\big)
        \]
        holds so we also have
        \[
            \widetilde \psi(i) = \abs*{\widetilde \phi^{-1}(i)} = \abs*{\phi_1^{-1}(i)} + \abs*{\phi_2^{-1}(i)} = \psi_1(i) + \psi_2(i)
        \]
        for every $i \in [k]$.
        And finally we have $\widetilde \ell = \abs*{\widetilde M} = \abs*{M_1} + \abs*{M_2} = \ell_1 + \ell_2$.
\end{proof}

When labels $i$ and $j$ are joined, we can add some new edges with endpoints in $S \setminus V(M)$ to $M$ to obtain a new partial solution.
For every set of edges that can possibly be added this way, there is a label choice which maps one endpoint of every edge to $i$ the other to~$j$.
So the number $r$ of such edges is bounded by the number of vertices mapped to each of labels~$i$ and~$j$.
Formally, we have the following relation.
\begin{lemma}\label{lem:eds-join}
    Let $(G_x, \lab_x) = \eta_{i, j}(G_y, \lab_y)$ for some labels $i \neq j \in [k]$ where the graph~$G_y$ contains no vertex~$v$ with $\{i,j\} \subseteq \lab_y(v)$.
    Then the footprints of the node~$x$ are given as follows:
    \begin{align*}
        \cS_x = \biggl\{\Big(I, \psi\big[i \mapsto \psi(i) - r, j \mapsto \psi(j) - r\big], \ell+r\Big) \mid &(I, \psi, \ell) \in \cS_y, \{i,j\}\not\subseteq I, \\ 
        &r \in \Big[\min\big\{\psi(i), \psi(j)\big\}\Big]_0\biggr\}.
    \end{align*}
\end{lemma}
\begin{proof}
    Let $(I, \psi, \ell) \in \cS_y$ be a footprint of $y$ satisfying $\{i, j\} \not\subseteq I$ and let $r$ be an integer satisfying 
    $r \in \Big[\min\big\{\psi(i), \psi(j)\big\}\Big]_0$.
    Further, let $(S, M)$ be a partial solution and let $\phi$ be a label choice of $(S, M)$ in $y$ with this footprint.
    Since 
    \[
        \{i, j\} \not\subseteq I = \bigcup_{v \in V_x \setminus S} \lab_y(v),
    \]
    the set~$E_x \setminus E_y$ contains no edge with both endpoints in $V_x \setminus S$ by definition of a join-node.
    Further, the graph $G_y$ is a subgraph of~$G_x$.
    Together with the fact that $S$ is a vertex cover of~$G_y$, this implies that~$S$ is also a vertex cover of $G_x$.
    Because $G_x$ is a subgraph of $G_y$, the set~$M$ is also a matching of the graph~$G_x$. 
    Due to 
    \[
        r \leq \min\big\{\psi(i), \psi(j)\big\} = \min\big\{\phi^{-1}(i), \phi^{-1}(j)\big\},
    \]
    there exist pairwise distinct vertices $u_1, \dots, u_r \in S \setminus V(M)$ with $\phi(u_1) = \dots = \phi(u_r) = i$.
    In particular, we have $i \in \bigcap_{z \in [r]} \lab_x(u_z)$.
    Similarly, there exist pairwise distinct vertices $v_1, \dots, v_r \in S \setminus V(M)$ with $\phi(v_1) = \dots = \phi(v_r) = j$.
    In particular, we have $j \in \bigcap_{z \in [r]} \lab_x(v_z)$.
    We define the set 
    \[
        Y = \big\{u_z v_z \mid z \in [r]\big\}.
    \]
    First, by definition of a join-node, we have $Y \subseteq E_x$.
    Second, by the precondition of the join-operator $\eta_{i,j}$, no vertex in $G_y$ (and hence, in $G_x$ too) contains both labels $i$ and $j$.
    This, in particular, implies that we have $u_z \neq v_q$ for all $z, q \in [r]$ and therefore, the set~$Y$ is a matching in the graph~$G_x$.
    Finally, by the choice of the vertices~$u_1, \dots, u_r$ and~$v_1, \dots, v_r$, the sets $M$ and $Y$ have disjoint sets of endpoints so the set~$M \cup Y$ is a matching in $G_x$.
    Altogether,~$(S, M \cup Y)$ is a partial solution and $\phi_{|_{S \setminus (V(M) \cup V(Y))}}$ is label choice of $(S, M \cup Y)$ in $x$.
    So let $(\widetilde I, \widetilde \psi, \widetilde \ell) \in \cS_x$ be the footprint of $(S, M \cup Y)$ and $\phi_{|_{S \setminus (V(M) \cup V(Y))}}$ in~$x$.
    First, we have $\widetilde I = I$ since the label set of every vertex is the same in $\lab_y$ and $\lab_x$.
    Second, by the choice of the vertices $u_1, \dots, u_r$ and $v_1, \dots, v_r$, we have
    \[
        \widetilde \psi(i) = \abs*{\phi_{|_{S \setminus (V(M) \cup V(Y))}}^{-1}(i)} = \abs*{\phi^{-1}(i)} - r = \psi(i) - r, 
    \]
    the analogous equality holds for $j$, and we have
    \[
        \widetilde \psi(z) = \abs*{\phi_{|_{S \setminus (V(M) \cup V(Y))}}^{-1}(z)} = \abs*{\phi^{-1}(z)} = \psi(z)
    \]
    for all $z \in [k] \setminus \{i, j\}$.
    And finally, we have $\widetilde \ell = \abs*{M \cup Y} = \abs*{M} + \abs*{Y} = \ell + r$.

    For the other direction, consider a footprint $(\widetilde I, \widetilde \psi, \widetilde \ell) \in \cS_x$ of $x$.
    And let $(\widetilde S, \widetilde M)$ be a partial solution and let $\widetilde \phi$ be its label choice in $x$ with this footprint.
    Then $\widetilde S$ is also a vertex cover of the subgraph $G_y$ of $G_x$.
    First, suppose for the sake of contradiction that~$\{i, j\} \subseteq \widetilde I$ holds.
    Then there exist vertices~$u, v \in V_x \setminus \widetilde S$ such that $i \in \lab_x(u)$ and $j \in \lab_x(v)$.
    By the precondition on the join-operator~$\eta_{i,j}$ we have $u \neq v$.
    And by definition of this join, we then have~$uv \in E_x$ contradicting the fact that $\widetilde S$ is a vertex cover of $G_x$.
    So we have $\{i, j\} \not\subseteq \widetilde I$.
    Next we define the set 
    \[
        Y = \big\{uv \in \widetilde M \mid i \in \lab_y(u), j \in \lab_y(v)\big\}
    \]
    and let $r = \abs*{Y}$ be its cardinality.
    Then by definition of the join-operator~$\eta_{i,j}$, all elements of~$\widetilde M \setminus Y$ belong to $E_y$ and therefore, $(\widetilde S, \widetilde M \setminus Y)$ is a partial solution of the node $y$.
    We have
    \[
        \widetilde S \setminus \big(V(\widetilde M \setminus Y)\big) = \big(\widetilde S \setminus V(\widetilde M)\big) \cup V(Y)
    \]
    and the labeling functions $\lab_x$ and $\lab_y$ coincide.
    So consider the label choice $\phi$ of $(\widetilde S, \widetilde M \setminus Y)$ in $y$ defined as 
    \[
        \phi(v) =
        \begin{cases}
            \widetilde \phi(v) & \text{if } v \in S \setminus V(\widetilde M) \\
            i & \text{if } v \in V(Y) \text{ and } i \in \lab_y(v) \\
            j & \text{if } v \in V(Y) \text{ and } j \in \lab_y(v) \\
        \end{cases}
    \]
    for every vertex $v \in \big(\widetilde S \setminus V(\widetilde M)\big) \cup V(Y)$.
    Recall that no vertex of $G_y$ contains both labels~$i$ and~$j$ in $\lab_y$ so $\phi$ is well-defined.
    And now let $(I, \psi, \ell) \in \cS_y$ be the footprint of $(\widetilde S, \widetilde M \setminus Y)$ and the label choice $\phi$ in $y$.
    First, we trivially have $\widetilde I = I$ and $\widetilde \ell = \ell+r$.
    And second, we have $\widetilde \psi(i) = \psi(i) - r$, $\widetilde \psi(j) = \psi(j) - r$,
    and $\widetilde \psi(z) = \psi(z)$ for all $z \in [k] \setminus \{i,j\}$ by construction of~$\phi$.
    In particular, we have $r \in \Big[\min\big\{\psi(i), \psi(j)\big\}\Big]_0$.
\end{proof}

With those lemmas in hand we are ready to provide the desired algorithm:
\thmEDS*
\begin{proof}
    First, by \cref{lem:mcw-special-relabels} we may assume that every relabel-node either adds a label or forgets a label and every introduce-node creates a vertex with precisely one label.
    Furthermore, the number of nodes in the expression is polynomial in $n$. 
    Now we slightly change the expression by adding a new label $\star \coloneqq k+1$ and right after each introduce-node, say $v \langle i \rangle$ for some element~$v$ and some label $i \in [k]$, we add an add-label-node $\rho_{i \to \{i, \star\}}$.
    We denote this expression by~$\xi$.
    The properties from \cref{lem:mcw-special-relabels} remain satisfied in $\xi$.
    And $\xi$ is a multi-$(k+1)$-expression of the same graph~$G$ such that for the root $r$ of this expression, every vertex of the graph $G_r = G$ has a non-empty label set in $\lab_r$.
    The algorithm is a bottom-up dynamic programming along $\xi$ that for every node $x$, computes the set $\cS_x$ of all footprints of all partial solutions and their label choices in $x$.
    For the leaf-computation, we use \cref{lem:eds-leaf}.
    And for every inner node $x$, we compute the set~$\cS_x$ from the footprints of its child / children using \cref{lem:eds-remove-label}, \cref{lem:eds-add-label}, \cref{lem:eds-union}, or \cref{lem:eds-join} depending on the type of the node $x$.
    At the root of the multi-expression $\xi$, the following equivalence holds.
    \begin{claim}
        The graph $G$ admits an edge dominating set of size at most $\cardinality$ if and only if for the root $r$ of $\xi$, the set $\cS_r$ contains a footprint $(I, \psi, \ell)$ with $I \subseteq [k+1]$, $\psi \colon [k+1] \to [n]_0$, and $\ell \in \bN$ such that $\ell + \sum_{i \in [k] \cup \{\star\}} \psi(i) \leq \cardinality$.
    \end{claim}
    \begin{claimproof}
        First, if $G$ admits an edge dominating set of size at most $\cardinality$, then by \cref{lem:eds-equivalent-formulation}, there exists a pair $(S, M)$ such that $S$ is a vertex cover of $G = G_r$, $M$ is a matching in~$G$ satisfying~$V(M) \subseteq S$, and we have $\abs*{M} + \abs*{S \setminus V(M)} \leq \cardinality$.
        Then $(S, M)$ is a partial solution of $r$.
        Let now $\phi \colon S \setminus V(M) \to [k] \cup \{\star\}$ be a mapping that maps every vertex from~$V(M) \setminus S$ to $\star$.
        By construction of $\xi$, the function $\phi$ is a label choice for~$(S, M)$ in~$r$.
        For the footprint~$(I, \psi, \ell) \in \cS_r$ of $(S, M)$ and $\phi$ in $r$, we then in particular have~$\psi(\star) = \abs*{S \setminus V(M)}$,~$\psi(i) = 0$ for all $i \in [k]$, and $\ell = \abs*{M}$.
        So we get 
        \[
            \ell + \sum_{i \in [k] \cup \{\star\}} \psi(i) = \abs*{M} + \abs*{S \setminus V(M)} \leq \cardinality.
        \]

        On the other hand, if the set~$\cS_r$ contains a footprint $(I, \psi, \ell)$ such that $\ell + \sum_{i \in [k] \cup \{\star\}} \psi(i) \leq \cardinality$, there exists a partial solution $(S, M)$ of $r$ and a label choice $\phi \colon S \setminus V(M) \to [k] \cup \{\star\}$ of~$(S, M)$ in $r$ such that $\psi(i) = \abs*{\phi^{-1}(i)}$ for all $i \in [k] \cup \{\star\}$ and $\ell = \abs*{M}$.
        Then the set~$S$ is a vertex cover of $G_r = G$, the set~$M$ is a matching in $G_r = G$ satisfying $V(M) \subseteq S$, and we have
        \[
            \abs*{M} + \abs*{S \setminus V(M)} = 
            \ell + \sum_{i \in [k] \cup \{\star\}} \abs*{\phi^{-1}(i)} = 
            \ell + \sum_{i \in [k] \cup \{\star\}} \psi(i) \leq \cardinality.
        \]
        By \cref{lem:eds-equivalent-formulation}, the graph $G$ then admits an edge dominating set of size at most $\cardinality$.
    \end{claimproof}
    So after computing the set $\cS_r$, we iterate through its elements and check whether it contains a footprint  $(I, \psi, \ell)$ satisfying $\ell + \sum_{i \in [k] \cup \{\star\}} \psi(i) \leq \cardinality$.

    It remains to bound the running time of the algorithm.
    First, for every node $x$ of the multi-$k+1$-expression~$\xi$ of $G$, the graph $G_x$ contains at most $n$ vertices so there are at most
    \[
        2^{k+1} \cdot (n+1)^{k+1} \cdot \left(\left\lceil\frac{n}{2}\right\rceil + 1\right) = n^{\cO(k)}
    \]
    footprints in~$\cS_x$.
    Indeed, for a subset $I$ of $[k+1]$ there are $2^{k+1}$ options.
    For a mapping~$\psi \colon [k+1] \to [n]_0$, there are $\abs*{[n]_0}^{k+1} = (n+1)^{k+1}$ options.
    And for the cardinality $\ell$ of a matching~$M$ in $G_x$, there are at most $\left(\left\lceil\frac{n}{2}\right\rceil + 1\right)$ options because this cardinality lies between $0$ and~$\lceil\frac{n}{2}\rceil$.
    For every node of $\xi$, after finishing the procedure for this node, in time $n^{\cO(k)}$ we remove the duplicates of computed footprints to ensure that the number of kept footprints never exceeds~$n^{\cO(k)}$.
    The leaf computation via \cref{lem:eds-leaf} can be done in time $\cO(k)$.
    The forget-label computation via \cref{lem:eds-remove-label} can clearly be done in time $n^{\cO(k)}$.
    The add-label computation via \cref{lem:eds-add-label} can be done in time~$n^{\cO(k)}$ as there are $\cO(n)$ options for the integer~$r$.
    The union computation can also be handled in time~$n^{\cO(k)}$ because $n^{\cO(k)} \cdot n^{\cO(k)} \in n^{\cO(k)}$.
    And finally, we can also process a join-node in time~$n^{\cO(k)}$ because, again, there are $\cO(n)$ options for the integer~$r$.
    Since the number of nodes in the expression $\xi$ is polynomial in~$n$, this yields the total running time of $n^{\cO(k)}$.
\end{proof}

Fomin et al.~\cite{FominGLS14} showed the following lower bound for clique-width:

\begin{theorem}[Theorem 5.1 in \cite{FominGLS14}]
    Let $G$ be an $n$-vertex graph given together with a $k$-expression of $G$. 
    Then the \textsc{Edge Dominating Set} problem cannot be solved in time~$f(k) \cdot n^{o(k)}$ for any computable function $f$ unless the ETH fails.
\end{theorem}

Since any $k$-expression of a graph 
can trivially be transformed into a multi-$k$-expression of the same graph~\cite{Furer17}, the lower bound transfers to multi-clique-width showing that our algorithm is tight under ETH.
\begin{theorem}
    Let $G$ be an $n$-vertex graph given together with a multi-$k$-expression of $G$ and let $\cardinality$ be an integer. 
    Then the \textsc{Edge Dominating Set} problem cannot be solved in time $f(k) \cdot n^{o(k)}$ for any computable function $f$ unless the ETH fails.
\end{theorem}

\section{Conclusion}

We have answered three open questions left by Chekan and Kratsch~\cite{ChekanK23} about the fine-grained complexity of \textsc{Max Cut}, \textsc{Hamiltonian Cycle}, and \textsc{Edge Dominating Set} with respect to multi-clique-width. For \textsc{Max Cut} we ruled out time $n^{2^{o(k)}}$ assuming ETH, while a matching algorithm can be obtained by using that multi-clique-width $k$ implies clique-width is at most $2^k$. In contrast, for \textsc{Hamiltonian Cycle} and \textsc{Edge Dominating Set} we were able to give algorithms with the same time bound $n^{\cO(k)}$ as for clique-width, which matches known lower bounds that extend to multi-clique-width. It would be interesting to study further problems that are XP (but likely not FPT) and have tight bounds with respect to clique-width and to determine such tight bounds relative to multi-clique-width and fusion-width. Are there problems that separate two or even all three of these parameters?

The same question can also be posed for problems that are FPT for clique-width. So far, for a small selection of problems, the same tight bounds were obtained for all three parameters~\cite{Furer14,Furer17,ChekanK23}. It would similarly be interesting to find separations between them.

The fact that multi-clique-width generalizes treewidth (and clique-width) with only a constant increase in value arguably makes it an attractive parameter for fine-grained parameterized complexity. This would be much strengthened if there were a better way to approximate it, say in FPT time. So far, only a doubly-exponential approximation follows from its relation to clique-with. Can this be improved?

\bibliographystyle{plain}
\bibliography{ref}

@inproceedings{BergougnouxKN23,
  author       = {Benjamin Bergougnoux and
                  Tuukka Korhonen and
                  Jesper Nederlof},
  editor       = {Petra Berenbrink and
                  Patricia Bouyer and
                  Anuj Dawar and
                  Mamadou Moustapha Kant{\'{e}}},
  title        = {Tight Lower Bounds for Problems Parameterized by Rank-Width},
  booktitle    = {40th International Symposium on Theoretical Aspects of Computer Science,
                  {STACS} 2023, Hamburg, Germany, March 7-9, 2023},
  series       = {LIPIcs},
  volume       = {254},
  pages        = {11:1--11:17},
  publisher    = {Schloss Dagstuhl - Leibniz-Zentrum f{\"{u}}r Informatik},
  year         = {2023},
  url          = {https://doi.org/10.4230/LIPIcs.STACS.2023.11},
  doi          = {10.4230/LIPICS.STACS.2023.11},
  bibsource    = {dblp computer science bibliography, https://dblp.org}
}

@article {BodlaenderCKN15,
	AUTHOR = {Bodlaender, Hans L. and Cygan, Marek and Kratsch, Stefan and
	Nederlof, Jesper},
	TITLE = {Deterministic single exponential time algorithms for
	connectivity problems parameterized by treewidth},
	JOURNAL = {Information and Computation},
	VOLUME = {243},
	YEAR = {2015},
	PAGES = {86--111},
	ISSN = {0890-5401},
	MRCLASS = {05C85 (68Q25)},
	MRNUMBER = {3352777},
	DOI = {10.1016/j.ic.2014.12.008},
	URL = {http://dx.doi.org/10.1016/j.ic.2014.12.008},
}

@inproceedings{GreilhuberSW25,
  author       = {Jakob Greilhuber and
                  Philipp Schepper and
                  Philip Wellnitz},
  editor       = {Olaf Beyersdorff and
                  Michal Pilipczuk and
                  Elaine Pimentel and
                  Kim Thang Nguyen},
  title        = {Residue Domination in Bounded-Treewidth Graphs},
  booktitle    = {42nd International Symposium on Theoretical Aspects of Computer Science,
                  {STACS} 2025, March 4-7, 2025, Jena, Germany},
  series       = {LIPIcs},
  volume       = {327},
  pages        = {41:1--41:20},
  publisher    = {Schloss Dagstuhl - Leibniz-Zentrum f{\"{u}}r Informatik},
  year         = {2025},
  url          = {https://doi.org/10.4230/LIPIcs.STACS.2025.41},
  doi          = {10.4230/LIPICS.STACS.2025.41},
  bibsource    = {dblp computer science bibliography, https://dblp.org}
}

@article{abs-2506-01645,
  author       = {Jakob Greilhuber and
                  D{\'{a}}niel Marx},
  title        = {The Price of Being Partial: Complexity of Partial Generalized Dominating
                  Set on Bounded-Treewidth Graphs},
  journal      = {CoRR},
  volume       = {abs/2506.01645},
  year         = {2025},
  url          = {https://doi.org/10.48550/arXiv.2506.01645},
  doi          = {10.48550/ARXIV.2506.01645},
  eprinttype    = {arXiv},
  eprint       = {2506.01645},
  bibsource    = {dblp computer science bibliography, https://dblp.org}
}

@article{FominGLS14,
  author    = {Fedor V. Fomin and
               Petr A. Golovach and
               Daniel Lokshtanov and
               Saket Saurabh},
  title     = {{Almost Optimal Lower Bounds for Problems Parameterized by Clique-Width}},
  journal   = {{SIAM} J. Comput.},
  volume    = {43},
  number    = {5},
  pages     = {1541--1563},
  year      = {2014},
  url       = {https://doi.org/10.1137/130910932},
  doi       = {10.1137/130910932},
  bibsource = {dblp computer science bibliography, https://dblp.org}
}

@inproceedings{ChekanK23,
  author       = {Vera Chekan and
                  Stefan Kratsch},
  editor       = {J{\'{e}}r{\^{o}}me Leroux and
                  Sylvain Lombardy and
                  David Peleg},
  title        = {{Tight Algorithmic Applications of Clique-Width Generalizations}},
  booktitle    = {48th International Symposium on Mathematical Foundations of Computer
                  Science, {MFCS} 2023, August 28 to September 1, 2023, Bordeaux, France},
  series       = {LIPIcs},
  volume       = {272},
  pages        = {35:1--35:15},
  publisher    = {Schloss Dagstuhl - Leibniz-Zentrum f{\"{u}}r Informatik},
  year         = {2023},
  url          = {https://doi.org/10.4230/LIPIcs.MFCS.2023.35},
  doi          = {10.4230/LIPICS.MFCS.2023.35},
  bibsource    = {dblp computer science bibliography, https://dblp.org}
}

@article{BergougnouxKK20,
  author       = {Benjamin Bergougnoux and
                  Mamadou Moustapha Kant{\'{e}} and
                  O{-}joung Kwon},
  title        = {An Optimal {XP} Algorithm for Hamiltonian Cycle on Graphs of Bounded
                  Clique-Width},
  journal      = {Algorithmica},
  volume       = {82},
  number       = {6},
  pages        = {1654--1674},
  year         = {2020},
  url          = {https://doi.org/10.1007/s00453-019-00663-9},
  doi          = {10.1007/S00453-019-00663-9},
  bibsource    = {dblp computer science bibliography, https://dblp.org}
}

@article{FominGLSZ19,
  author    = {Fedor V. Fomin and
               Petr A. Golovach and
               Daniel Lokshtanov and
               Saket Saurabh and
               Meirav Zehavi},
  title     = {Clique-width {III:} Hamiltonian Cycle and the Odd Case of Graph Coloring},
  journal   = {{ACM} Trans. Algorithms},
  volume    = {15},
  number    = {1},
  pages     = {9:1--9:27},
  year      = {2019},
  url       = {https://doi.org/10.1145/3280824},
  doi       = {10.1145/3280824},
  bibsource = {dblp computer science bibliography, https://dblp.org}
}

@article{LokshtanovMS18,
  author       = {Daniel Lokshtanov and
                  D{\'{a}}niel Marx and
                  Saket Saurabh},
  title        = {Known Algorithms on Graphs of Bounded Treewidth Are Probably Optimal},
  journal      = {{ACM} Trans. Algorithms},
  volume       = {14},
  number       = {2},
  pages        = {13:1--13:30},
  year         = {2018},
  url          = {https://doi.org/10.1145/3170442},
  doi          = {10.1145/3170442},
  bibsource    = {dblp computer science bibliography, https://dblp.org}
}

@inproceedings{Furer14,
  author    = {Martin F{\"{u}}rer},
  editor    = {Alberto Pardo and
               Alfredo Viola},
  title     = {A Natural Generalization of Bounded Tree-Width and Bounded Clique-Width},
  booktitle = {{LATIN} 2014: Theoretical Informatics - 11th Latin American Symposium,
               Montevideo, Uruguay, March 31 - April 4, 2014. Proceedings},
  series    = {Lecture Notes in Computer Science},
  volume    = {8392},
  pages     = {72--83},
  publisher = {Springer},
  year      = {2014},
  url       = {https://doi.org/10.1007/978-3-642-54423-1\_7},
  doi       = {10.1007/978-3-642-54423-1\_7},
  bibsource = {dblp computer science bibliography, https://dblp.org}
}

@inproceedings{Furer17,
  author    = {Martin F{\"{u}}rer},
  editor    = {Christos H. Papadimitriou},
  title     = {Multi-Clique-Width},
  booktitle = {8th Innovations in Theoretical Computer Science Conference, {ITCS}
               2017, January 9-11, 2017, Berkeley, CA, {USA}},
  series    = {LIPIcs},
  volume    = {67},
  pages     = {14:1--14:13},
  publisher = {Schloss Dagstuhl - Leibniz-Zentrum f{\"{u}}r Informatik},
  year      = {2017},
  url       = {https://doi.org/10.4230/LIPIcs.ITCS.2017.14},
  doi       = {10.4230/LIPIcs.ITCS.2017.14},
  bibsource = {dblp computer science bibliography, https://dblp.org}
}

@article{CyganKN18,
  author       = {Marek Cygan and
                  Stefan Kratsch and
                  Jesper Nederlof},
  title        = {Fast Hamiltonicity Checking Via Bases of Perfect Matchings},
  journal      = {J. {ACM}},
  volume       = {65},
  number       = {3},
  pages        = {12:1--12:46},
  year         = {2018},
  url          = {https://doi.org/10.1145/3148227},
  doi          = {10.1145/3148227},
  bibsource    = {dblp computer science bibliography, https://dblp.org}
}

@article{FockeMINSSW25,
  author       = {Jacob Focke and
                  D{\'{a}}niel Marx and
                  Fionn Mc Inerney and
                  Daniel Neuen and
                  Govind S. Sankar and
                  Philipp Schepper and
                  Philip Wellnitz},
  title        = {Tight Complexity Bounds for Counting Generalized Dominating Sets in
                  Bounded-Treewidth Graphs - Part {I:} Algorithmic Results},
  journal      = {{ACM} Trans. Algorithms},
  volume       = {21},
  number       = {3},
  pages        = {27:1--27:45},
  year         = {2025},
  url          = {https://doi.org/10.1145/3731452},
  doi          = {10.1145/3731452},
  bibsource    = {dblp computer science bibliography, https://dblp.org}
}

@article{toct/FockeMINSSW25,
  author       = {Jacob Focke and
                  D{\'{a}}niel Marx and
                  Fionn Mc Inerney and
                  Daniel Neuen and
                  Govind S. Sankar and
                  Philipp Schepper and
                  Philip Wellnitz},
  title        = {Tight Complexity Bounds for Counting Generalized Dominating Sets in
                  Bounded-Treewidth Graphs Part {II:} Hardness Results},
  journal      = {{ACM} Trans. Comput. Theory},
  volume       = {17},
  number       = {2},
  pages        = {10:1--10:101},
  year         = {2025},
  url          = {https://doi.org/10.1145/3708509},
  doi          = {10.1145/3708509},
  bibsource    = {dblp computer science bibliography, https://dblp.org}
}

@article{Lampis20,
  author    = {Michael Lampis},
  title     = {Finer Tight Bounds for Coloring on Clique-Width},
  journal   = {{SIAM} J. Discret. Math.},
  volume    = {34},
  number    = {3},
  pages     = {1538--1558},
  year      = {2020},
  url       = {https://doi.org/10.1137/19M1280326},
  doi       = {10.1137/19M1280326},
  bibsource = {dblp computer science bibliography, https://dblp.org}
}

@inproceedings{BodlaenderLRV10,
  author       = {Hans L. Bodlaender and
                  Erik Jan van Leeuwen and
                  Johan M. M. van Rooij and
                  Martin Vatshelle},
  editor       = {Petr Hlinen{\'{y}} and
                  Anton{\'{\i}}n Kucera},
  title        = {Faster Algorithms on Branch and Clique Decompositions},
  booktitle    = {Mathematical Foundations of Computer Science 2010, 35th International
                  Symposium, {MFCS} 2010, Brno, Czech Republic, August 23-27, 2010.
                  Proceedings},
  series       = {Lecture Notes in Computer Science},
  volume       = {6281},
  pages        = {174--185},
  publisher    = {Springer},
  year         = {2010},
  url          = {https://doi.org/10.1007/978-3-642-15155-2\_17},
  doi          = {10.1007/978-3-642-15155-2\_17},
  bibsource    = {dblp computer science bibliography, https://dblp.org}
}

@article{KoblerR03,
  author    = {Daniel Kobler and
               Udi Rotics},
  title     = {Edge dominating set and colorings on graphs with fixed clique-width},
  journal   = {Discret. Appl. Math.},
  volume    = {126},
  number    = {2-3},
  pages     = {197--221},
  year      = {2003},
  url       = {https://doi.org/10.1016/S0166-218X(02)00198-1},
  doi       = {10.1016/S0166-218X(02)00198-1},
  bibsource = {dblp computer science bibliography, https://dblp.org}
}

@inproceedings{HegerfeldK23,
  author       = {Falko Hegerfeld and
                  Stefan Kratsch},
  editor       = {Inge Li G{\o}rtz and
                  Martin Farach{-}Colton and
                  Simon J. Puglisi and
                  Grzegorz Herman},
  title        = {Tight Algorithms for Connectivity Problems Parameterized by Clique-Width},
  booktitle    = {31st Annual European Symposium on Algorithms, {ESA} 2023, Amsterdam,
                  The Netherlands, September 4-6, 2023},
  series       = {LIPIcs},
  volume       = {274},
  pages        = {59:1--59:19},
  publisher    = {Schloss Dagstuhl - Leibniz-Zentrum f{\"{u}}r Informatik},
  year         = {2023},
  url          = {https://doi.org/10.4230/LIPIcs.ESA.2023.59},
  doi          = {10.4230/LIPICS.ESA.2023.59},
  bibsource    = {dblp computer science bibliography, https://dblp.org}
}

@article{CorneilR05,
  author       = {Derek G. Corneil and
                  Udi Rotics},
  title        = {On the Relationship Between Clique-Width and Treewidth},
  journal      = {{SIAM} J. Comput.},
  volume       = {34},
  number       = {4},
  pages        = {825--847},
  year         = {2005},
  url          = {https://doi.org/10.1137/S0097539701385351},
  doi          = {10.1137/S0097539701385351},
  bibsource    = {dblp computer science bibliography, https://dblp.org}
}

@inproceedings{BojikianK24,
  author       = {Narek Bojikian and
                  Stefan Kratsch},
  editor       = {Karl Bringmann and
                  Martin Grohe and
                  Gabriele Puppis and
                  Ola Svensson},
  title        = {A Tight Monte-Carlo Algorithm for Steiner Tree Parameterized by Clique-Width},
  booktitle    = {51st International Colloquium on Automata, Languages, and Programming,
                  {ICALP} 2024, Tallinn, Estonia, July 8-12, 2024},
  series       = {LIPIcs},
  volume       = {297},
  pages        = {29:1--29:18},
  publisher    = {Schloss Dagstuhl - Leibniz-Zentrum f{\"{u}}r Informatik},
  year         = {2024},
  url          = {https://doi.org/10.4230/LIPIcs.ICALP.2024.29},
  doi          = {10.4230/LIPICS.ICALP.2024.29},
  bibsource    = {dblp computer science bibliography, https://dblp.org}
}

@inproceedings{BojikianK25,
  author       = {Narek Bojikian and
                  Stefan Kratsch},
  editor       = {Akanksha Agrawal and
                  Erik Jan van Leeuwen},
  title        = {Tight Bounds for Connected Odd Cycle Transversal Parameterized by
                  Clique-Width},
  booktitle    = {20th International Symposium on Parameterized and Exact Computation,
                  {IPEC} 2025, Warsaw, Poland, September 17-19, 2025},
  series       = {LIPIcs},
  volume       = {358},
  pages        = {19:1--19:15},
  publisher    = {Schloss Dagstuhl - Leibniz-Zentrum f{\"{u}}r Informatik},
  year         = {2025},
  url          = {https://doi.org/10.4230/LIPIcs.IPEC.2025.19},
  doi          = {10.4230/LIPICS.IPEC.2025.19},
  bibsource    = {dblp computer science bibliography, https://dblp.org}
}

@article{abs-2307-04628,
  author       = {Vera Chekan and
                  Stefan Kratsch},
  title        = {Tight Algorithmic Applications of Clique-Width Generalizations},
  journal      = {CoRR},
  volume       = {abs/2307.04628},
  year         = {2023},
  url          = {https://doi.org/10.48550/arXiv.2307.04628},
  doi          = {10.48550/ARXIV.2307.04628},
  eprinttype    = {arXiv},
  eprint       = {2307.04628},
  bibsource    = {dblp computer science bibliography, https://dblp.org}
}

@article{CyganNPPRW22,
  author       = {Marek Cygan and
                  Jesper Nederlof and
                  Marcin Pilipczuk and
                  Micha{\l} Pilipczuk and
                  Johan M. M. van Rooij and
                  Jakub Onufry Wojtaszczyk},
  title        = {Solving Connectivity Problems Parameterized by Treewidth in Single
                  Exponential Time},
  journal      = {{ACM} Trans. Algorithms},
  volume       = {18},
  number       = {2},
  pages        = {17:1--17:31},
  year         = {2022},
  url          = {https://doi.org/10.1145/3506707},
  doi          = {10.1145/3506707},
  bibsource    = {dblp computer science bibliography, https://dblp.org}
}

@inproceedings{RooijBR09,
  author       = {Johan M. M. van Rooij and
                  Hans L. Bodlaender and
                  Peter Rossmanith},
  editor       = {Amos Fiat and
                  Peter Sanders},
  title        = {Dynamic Programming on Tree Decompositions Using Generalised Fast
                  Subset Convolution},
  booktitle    = {Algorithms - {ESA} 2009, 17th Annual European Symposium, Copenhagen,
                  Denmark, September 7-9, 2009. Proceedings},
  series       = {Lecture Notes in Computer Science},
  volume       = {5757},
  pages        = {566--577},
  publisher    = {Springer},
  year         = {2009},
  url          = {https://doi.org/10.1007/978-3-642-04128-0\_51},
  doi          = {10.1007/978-3-642-04128-0\_51},
  bibsource    = {dblp computer science bibliography, https://dblp.org}
}

@article{LokshtanovMS11,
  author       = {Daniel Lokshtanov and
                  D{\'{a}}niel Marx and
                  Saket Saurabh},
  title        = {Lower bounds based on the Exponential Time Hypothesis},
  journal      = {Bull. {EATCS}},
  volume       = {105},
  pages        = {41--72},
  year         = {2011},
  url          = {https://sites.cs.ucsb.edu/~daniello/papers/surveyETH.pdf},
  bibsource    = {dblp computer science bibliography, https://dblp.org}
}

@inproceedings{Bonnet25,
  author       = {{\'{E}}douard Bonnet},
  editor       = {Michal Kouck{\'{y}} and
                  Nikhil Bansal},
  title        = {Treewidth Inapproximability and Tight {ETH} Lower Bound},
  booktitle    = {Proceedings of the 57th Annual {ACM} Symposium on Theory of Computing,
                  {STOC} 2025, Prague, Czechia, June 23-27, 2025},
  pages        = {2130--2135},
  publisher    = {{ACM}},
  year         = {2025},
  url          = {https://doi.org/10.1145/3717823.3718117},
  doi          = {10.1145/3717823.3718117},
  bibsource    = {dblp computer science bibliography, https://dblp.org}
}

@inproceedings{FoucaudGK0IST24,
  author       = {Florent Foucaud and
                  Esther Galby and
                  Liana Khazaliya and
                  Shaohua Li and
                  Fionn Mc Inerney and
                  Roohani Sharma and
                  Prafullkumar Tale},
  editor       = {Karl Bringmann and
                  Martin Grohe and
                  Gabriele Puppis and
                  Ola Svensson},
  title        = {Problems in {NP} Can Admit Double-Exponential Lower Bounds When Parameterized
                  by Treewidth or Vertex Cover},
  booktitle    = {51st International Colloquium on Automata, Languages, and Programming,
                  {ICALP} 2024, Tallinn, Estonia, July 8-12, 2024},
  series       = {LIPIcs},
  volume       = {297},
  pages        = {66:1--66:19},
  publisher    = {Schloss Dagstuhl - Leibniz-Zentrum f{\"{u}}r Informatik},
  year         = {2024},
  url          = {https://doi.org/10.4230/LIPIcs.ICALP.2024.66},
  doi          = {10.4230/LIPICS.ICALP.2024.66},
  bibsource    = {dblp computer science bibliography, https://dblp.org}
}

@article{LokshtanovMS18slighly,
  author       = {Daniel Lokshtanov and
                  D{\'{a}}niel Marx and
                  Saket Saurabh},
  title        = {Slightly Superexponential Parameterized Problems},
  journal      = {{SIAM} J. Comput.},
  volume       = {47},
  number       = {3},
  pages        = {675--702},
  year         = {2018},
  url          = {https://doi.org/10.1137/16M1104834},
  doi          = {10.1137/16M1104834},
  bibsource    = {dblp computer science bibliography, https://dblp.org}
}

@book{0086373,
  author       = {Michael Sipser},
  title        = {Introduction to the theory of computation},
  publisher    = {{PWS} Publishing Company},
  year         = {1997},
  isbn         = {978-0-534-94728-6},
  bibsource    = {dblp computer science bibliography, https://dblp.org}
}

@article{FellowsFLRSST11,
  author       = {Michael R. Fellows and
                  Fedor V. Fomin and
                  Daniel Lokshtanov and
                  Frances A. Rosamond and
                  Saket Saurabh and
                  Stefan Szeider and
                  Carsten Thomassen},
  title        = {On the complexity of some colorful problems parameterized by treewidth},
  journal      = {Inf. Comput.},
  volume       = {209},
  number       = {2},
  pages        = {143--153},
  year         = {2011},
  url          = {https://doi.org/10.1016/j.ic.2010.11.026},
  doi          = {10.1016/J.IC.2010.11.026},
  bibsource    = {dblp computer science bibliography, https://dblp.org}
}

@inproceedings{HegerfeldK23-mtw,
  author       = {Falko Hegerfeld and
                  Stefan Kratsch},
  editor       = {Dani{\"{e}}l Paulusma and
                  Bernard Ries},
  title        = {Tight Algorithms for Connectivity Problems Parameterized by Modular-Treewidth},
  booktitle    = {Graph-Theoretic Concepts in Computer Science - 49th International
                  Workshop, {WG} 2023, Fribourg, Switzerland, June 28-30, 2023, Revised
                  Selected Papers},
  series       = {Lecture Notes in Computer Science},
  volume       = {14093},
  pages        = {388--402},
  publisher    = {Springer},
  year         = {2023},
  url          = {https://doi.org/10.1007/978-3-031-43380-1\_28},
  doi          = {10.1007/978-3-031-43380-1\_28},
  bibsource    = {dblp computer science bibliography, https://dblp.org}
}

@inproceedings{BojikianCHK23,
  author       = {Narek Bojikian and
                  Vera Chekan and
                  Falko Hegerfeld and
                  Stefan Kratsch},
  editor       = {Petra Berenbrink and
                  Patricia Bouyer and
                  Anuj Dawar and
                  Mamadou Moustapha Kant{\'{e}}},
  title        = {Tight Bounds for Connectivity Problems Parameterized by Cutwidth},
  booktitle    = {40th International Symposium on Theoretical Aspects of Computer Science,
                  {STACS} 2023, Hamburg, Germany, March 7-9, 2023},
  series       = {LIPIcs},
  volume       = {254},
  pages        = {14:1--14:16},
  publisher    = {Schloss Dagstuhl - Leibniz-Zentrum f{\"{u}}r Informatik},
  year         = {2023},
  url          = {https://doi.org/10.4230/LIPIcs.STACS.2023.14},
  doi          = {10.4230/LIPICS.STACS.2023.14},
  bibsource    = {dblp computer science bibliography, https://dblp.org}
}

@article{GeffenJKM20,
  author       = {Bas A. M. van Geffen and
                  Bart M. P. Jansen and
                  Arnoud A. W. M. de Kroon and
                  Rolf Morel},
  title        = {Lower Bounds for Dynamic Programming on Planar Graphs of Bounded Cutwidth},
  journal      = {J. Graph Algorithms Appl.},
  volume       = {24},
  number       = {3},
  pages        = {461--482},
  year         = {2020},
  url          = {https://doi.org/10.7155/jgaa.00542},
  doi          = {10.7155/JGAA.00542},
  bibsource    = {dblp computer science bibliography, https://dblp.org}
}

@article{JansenN19,
  author       = {Bart M. P. Jansen and
                  Jesper Nederlof},
  title        = {Computing the chromatic number using graph decompositions via matrix
                  rank},
  journal      = {Theor. Comput. Sci.},
  volume       = {795},
  pages        = {520--539},
  year         = {2019},
  url          = {https://doi.org/10.1016/j.tcs.2019.08.006},
  doi          = {10.1016/J.TCS.2019.08.006},
  bibsource    = {dblp computer science bibliography, https://dblp.org}
}

@inproceedings{GroenlandMNS22,
  author       = {Carla Groenland and
                  Isja Mannens and
                  Jesper Nederlof and
                  Krisztina Szil{\'{a}}gyi},
  editor       = {Petra Berenbrink and
                  Benjamin Monmege},
  title        = {Tight Bounds for Counting Colorings and Connected Edge Sets Parameterized
                  by Cutwidth},
  booktitle    = {39th International Symposium on Theoretical Aspects of Computer Science,
                  {STACS} 2022, Marseille, France (Virtual Conference), March 15-18,
                  2022},
  series       = {LIPIcs},
  volume       = {219},
  pages        = {36:1--36:20},
  publisher    = {Schloss Dagstuhl - Leibniz-Zentrum f{\"{u}}r Informatik},
  year         = {2022},
  url          = {https://doi.org/10.4230/LIPIcs.STACS.2022.36},
  doi          = {10.4230/LIPICS.STACS.2022.36},
  bibsource    = {dblp computer science bibliography, https://dblp.org}
}

@inproceedings{BojikianCK25,
  author       = {Narek Bojikian and
                  Vera Chekan and
                  Stefan Kratsch},
  editor       = {Anne Benoit and
                  Haim Kaplan and
                  Sebastian Wild and
                  Grzegorz Herman},
  title        = {Tight Bounds for Some Classical Problems Parameterized by Cutwidth},
  booktitle    = {33rd Annual European Symposium on Algorithms, {ESA} 2025, Warsaw,
                  Poland, September 15-17, 2025},
  series       = {LIPIcs},
  volume       = {351},
  pages        = {13:1--13:17},
  publisher    = {Schloss Dagstuhl - Leibniz-Zentrum f{\"{u}}r Informatik},
  year         = {2025},
  url          = {https://doi.org/10.4230/LIPIcs.ESA.2025.13},
  doi          = {10.4230/LIPICS.ESA.2025.13},
  bibsource    = {dblp computer science bibliography, https://dblp.org}
}

@article{LampisMV25,
  author       = {Michael Lampis and
                  Nikolaos Melissinos and
                  Manolis Vasilakis},
  title        = {Parameterized Max Min Feedback Vertex Set},
  journal      = {{SIAM} J. Discret. Math.},
  volume       = {39},
  number       = {3},
  pages        = {1587--1620},
  year         = {2025},
  url          = {https://doi.org/10.1137/23m1605247},
  doi          = {10.1137/23M1605247},
  bibsource    = {dblp computer science bibliography, https://dblp.org}
}

@article{LampisV24,
  author       = {Michael Lampis and
                  Manolis Vasilakis},
  title        = {Structural Parameterizations for Two Bounded Degree Problems Revisited},
  journal      = {{ACM} Trans. Comput. Theory},
  volume       = {16},
  number       = {3},
  pages        = {17:1--17:51},
  year         = {2024},
  url          = {https://doi.org/10.1145/3665156},
  doi          = {10.1145/3665156},
  bibsource    = {dblp computer science bibliography, https://dblp.org}
}

@inproceedings{LampisV25,
  author       = {Michael Lampis and
                  Manolis Vasilakis},
  editor       = {Henning Fernau and
                  Philipp Kindermann},
  title        = {Structural Parameterizations for Induced and Acyclic Matching},
  booktitle    = {Graph-Theoretic Concepts in Computer Science - 51st International
                  Workshop, {WG} 2025, Otzenhausen, Germany, June 11-13, 2025, Revised
                  Selected Papers},
  series       = {Lecture Notes in Computer Science},
  volume       = {16124},
  pages        = {360--376},
  publisher    = {Springer},
  year         = {2025},
  url          = {https://doi.org/10.1007/978-3-032-11835-6\_26},
  doi          = {10.1007/978-3-032-11835-6\_26},
  bibsource    = {dblp computer science bibliography, https://dblp.org}
}

@inproceedings{LampisV25-node,
  author       = {Michael Lampis and
                  Manolis Vasilakis},
  editor       = {Akanksha Agrawal and
                  Erik Jan van Leeuwen},
  title        = {Parameterized Maximum Node-Disjoint Paths},
  booktitle    = {20th International Symposium on Parameterized and Exact Computation,
                  {IPEC} 2025, Warsaw, Poland, September 17-19, 2025},
  series       = {LIPIcs},
  volume       = {358},
  pages        = {3:1--3:15},
  publisher    = {Schloss Dagstuhl - Leibniz-Zentrum f{\"{u}}r Informatik},
  year         = {2025},
  url          = {https://doi.org/10.4230/LIPIcs.IPEC.2025.3},
  doi          = {10.4230/LIPICS.IPEC.2025.3},
  bibsource    = {dblp computer science bibliography, https://dblp.org}
}

@inproceedings{HanakaLMN0V25-steiner,
  author       = {Tesshu Hanaka and
                  Michael Lampis and
                  Nikolaos Melissinos and
                  Edouard Nemery and
                  Hirotaka Ono and
                  Manolis Vasilakis},
  editor       = {Ho{-}Lin Chen and
                  Wing{-}Kai Hon and
                  Meng{-}Tsung Tsai},
  title        = {Structural Parameters for Steiner Orientation},
  booktitle    = {36th International Symposium on Algorithms and Computation, {ISAAC}
                  2025, Tainan, Taiwan, December 7-10, 2025},
  series       = {LIPIcs},
  volume       = {359},
  pages        = {38:1--38:14},
  publisher    = {Schloss Dagstuhl - Leibniz-Zentrum f{\"{u}}r Informatik},
  year         = {2025},
  url          = {https://doi.org/10.4230/LIPIcs.ISAAC.2025.38},
  doi          = {10.4230/LIPICS.ISAAC.2025.38},
  bibsource    = {dblp computer science bibliography, https://dblp.org}
}

@inproceedings{Lampis25-pw-seth,
  author       = {Michael Lampis},
  editor       = {Yossi Azar and
                  Debmalya Panigrahi},
  title        = {The Primal Pathwidth {SETH}},
  booktitle    = {Proceedings of the 2025 Annual {ACM-SIAM} Symposium on Discrete Algorithms,
                  {SODA} 2025, New Orleans, LA, USA, January 12-15, 2025},
  pages        = {1494--1564},
  publisher    = {{SIAM}},
  year         = {2025},
  url          = {https://doi.org/10.1137/1.9781611978322.47},
  doi          = {10.1137/1.9781611978322.47},
  bibsource    = {dblp computer science bibliography, https://dblp.org}
}

@inproceedings{Lampis26,
  author       = {Michael Lampis},
  editor       = {Kasper Green Larsen and
                  Barna Saha},
  title        = {{k-SUM Hardness Implies Treewidth-SETH}},
  booktitle    = {Proceedings of the 2026 Annual {ACM-SIAM} Symposium on Discrete Algorithms,
                  {SODA} 2026, Vancouver, BC, Canada, January 11-14, 2026},
  pages        = {1916--1944},
  publisher    = {{SIAM}},
  year         = {2026},
  url          = {https://doi.org/10.1137/1.9781611978971.70},
  doi          = {10.1137/1.9781611978971.70},
  bibsource    = {dblp computer science bibliography, https://dblp.org}
}

@inproceedings{Lampis26a,
  author       = {Michael Lampis},
  editor       = {Kasper Green Larsen and
                  Barna Saha},
  title        = {{Circuits and Backdoors: Five Shades of the {SETH}}},
  booktitle    = {Proceedings of the 2026 Annual {ACM-SIAM} Symposium on Discrete Algorithms,
                  {SODA} 2026, Vancouver, BC, Canada, January 11-14, 2026},
  pages        = {1945--2001},
  publisher    = {{SIAM}},
  year         = {2026},
  url          = {https://doi.org/10.1137/1.9781611978971.71},
  doi          = {10.1137/1.9781611978971.71},
  bibsource    = {dblp computer science bibliography, https://dblp.org}
}

@inproceedings{DumontLL0T25,
  author       = {Joanne Dumont and
                  Michael Lampis and
                  Mathieu Liedloff and
                  Anthony Perez and
                  Ioan Todinca},
  editor       = {Akanksha Agrawal and
                  Erik Jan van Leeuwen},
  title        = {On Maximum 2-Clubs},
  booktitle    = {20th International Symposium on Parameterized and Exact Computation,
                  {IPEC} 2025, Warsaw, Poland, September 17-19, 2025},
  series       = {LIPIcs},
  volume       = {358},
  pages        = {13:1--13:14},
  publisher    = {Schloss Dagstuhl - Leibniz-Zentrum f{\"{u}}r Informatik},
  year         = {2025},
  url          = {https://doi.org/10.4230/LIPIcs.IPEC.2025.13},
  doi          = {10.4230/LIPICS.IPEC.2025.13},
  bibsource    = {dblp computer science bibliography, https://dblp.org}
}

@article{FeldmannL25,
  author       = {Andreas Emil Feldmann and
                  Michael Lampis},
  title        = {Parameterized Algorithms for Steiner Forest in Bounded Width Graphs},
  journal      = {{ACM} Trans. Algorithms},
  volume       = {21},
  number       = {4},
  pages        = {47:1--47:26},
  year         = {2025},
  url          = {https://doi.org/10.1145/3748724},
  doi          = {10.1145/3748724},
  bibsource    = {dblp computer science bibliography, https://dblp.org}
}

@article{HanakaKL26,
  author       = {Tesshu Hanaka and
                  Noleen K{\"{o}}hler and
                  Michael Lampis},
  title        = {Core stability in additively separable hedonic games of low treewidth},
  journal      = {J. Comput. Syst. Sci.},
  volume       = {157},
  pages        = {103748},
  year         = {2026},
  url          = {https://doi.org/10.1016/j.jcss.2025.103748},
  doi          = {10.1016/J.JCSS.2025.103748},
  bibsource    = {dblp computer science bibliography, https://dblp.org}
}

@inproceedings{BliznetsH24,
  author       = {Ivan Bliznets and
                  Markus Hecher},
  editor       = {Xujin Chen and
                  Bo Li},
  title        = {Tight Double Exponential Lower Bounds},
  booktitle    = {Theory and Applications of Models of Computation - 18th Annual Conference,
                  {TAMC} 2024, Hong Kong, China, May 13-15, 2024, Proceedings},
  series       = {Lecture Notes in Computer Science},
  volume       = {14637},
  pages        = {124--136},
  publisher    = {Springer},
  year         = {2024},
  url          = {https://doi.org/10.1007/978-981-97-2340-9\_11},
  doi          = {10.1007/978-981-97-2340-9\_11},
  bibsource    = {dblp computer science bibliography, https://dblp.org}
}

@inproceedings{LampisM17,
  author       = {Michael Lampis and
                  Valia Mitsou},
  editor       = {Daniel Lokshtanov and
                  Naomi Nishimura},
  title        = {Treewidth with a Quantifier Alternation Revisited},
  booktitle    = {12th International Symposium on Parameterized and Exact Computation,
                  {IPEC} 2017, Vienna, Austria, September 6-8, 2017},
  series       = {LIPIcs},
  volume       = {89},
  pages        = {26:1--26:12},
  publisher    = {Schloss Dagstuhl - Leibniz-Zentrum f{\"{u}}r Informatik},
  year         = {2017},
  url          = {https://doi.org/10.4230/LIPIcs.IPEC.2017.26},
  doi          = {10.4230/LIPICS.IPEC.2017.26},
  bibsource    = {dblp computer science bibliography, https://dblp.org}
}

@inproceedings{MarxM16,
  author       = {D{\'{a}}niel Marx and
                  Valia Mitsou},
  editor       = {Ioannis Chatzigiannakis and
                  Michael Mitzenmacher and
                  Yuval Rabani and
                  Davide Sangiorgi},
  title        = {Double-Exponential and Triple-Exponential Bounds for Choosability
                  Problems Parameterized by Treewidth},
  booktitle    = {43rd International Colloquium on Automata, Languages, and Programming,
                  {ICALP} 2016, Rome, Italy, July 11-15, 2016},
  series       = {LIPIcs},
  volume       = {55},
  pages        = {28:1--28:15},
  publisher    = {Schloss Dagstuhl - Leibniz-Zentrum f{\"{u}}r Informatik},
  year         = {2016},
  url          = {https://doi.org/10.4230/LIPIcs.ICALP.2016.28},
  doi          = {10.4230/LIPICS.ICALP.2016.28},
  bibsource    = {dblp computer science bibliography, https://dblp.org}
}

@inproceedings{HanakaL22,
  author       = {Tesshu Hanaka and
                  Michael Lampis},
  editor       = {Shiri Chechik and
                  Gonzalo Navarro and
                  Eva Rotenberg and
                  Grzegorz Herman},
  title        = {Hedonic Games and Treewidth Revisited},
  booktitle    = {30th Annual European Symposium on Algorithms, {ESA} 2022, Berlin/Potsdam,
                  Germany, September 5-9, 2022},
  series       = {LIPIcs},
  volume       = {244},
  pages        = {64:1--64:16},
  publisher    = {Schloss Dagstuhl - Leibniz-Zentrum f{\"{u}}r Informatik},
  year         = {2022},
  url          = {https://doi.org/10.4230/LIPIcs.ESA.2022.64},
  doi          = {10.4230/LIPICS.ESA.2022.64},
  bibsource    = {dblp computer science bibliography, https://dblp.org}
}

@article{DubloisLP22,
  author       = {Louis Dublois and
                  Michael Lampis and
                  Vangelis Th. Paschos},
  title        = {Upper Dominating Set: Tight algorithms for pathwidth and sub-exponential
                  approximation},
  journal      = {Theor. Comput. Sci.},
  volume       = {923},
  pages        = {271--291},
  year         = {2022},
  url          = {https://doi.org/10.1016/j.tcs.2022.05.013},
  doi          = {10.1016/J.TCS.2022.05.013},
  bibsource    = {dblp computer science bibliography, https://dblp.org}
}

@article{LampisM24,
  author       = {Michael Lampis and
                  Valia Mitsou},
  title        = {Fine-grained Meta-Theorems for Vertex Integrity},
  journal      = {Log. Methods Comput. Sci.},
  volume       = {20},
  number       = {4},
  year         = {2024},
  url          = {https://doi.org/10.46298/lmcs-20(4:18)2024},
  doi          = {10.46298/LMCS-20(4:18)2024},
  bibsource    = {dblp computer science bibliography, https://dblp.org}
}

@inproceedings{MarxSS21,
  author       = {D{\'{a}}niel Marx and
                  Govind S. Sankar and
                  Philipp Schepper},
  editor       = {Nikhil Bansal and
                  Emanuela Merelli and
                  James Worrell},
  title        = {Degrees and Gaps: Tight Complexity Results of General Factor Problems
                  Parameterized by Treewidth and Cutwidth},
  booktitle    = {48th International Colloquium on Automata, Languages, and Programming,
                  {ICALP} 2021, Glasgow, Scotland (Virtual Conference), July 12-16,
                  2021},
  series       = {LIPIcs},
  volume       = {198},
  pages        = {95:1--95:20},
  publisher    = {Schloss Dagstuhl - Leibniz-Zentrum f{\"{u}}r Informatik},
  year         = {2021},
  url          = {https://doi.org/10.4230/LIPIcs.ICALP.2021.95},
  doi          = {10.4230/LIPICS.ICALP.2021.95},
  bibsource    = {dblp computer science bibliography, https://dblp.org}
}

@inproceedings{EsmerM25,
  author       = {Baris Can Esmer and
                  D{\'{a}}niel Marx},
  editor       = {Anne Benoit and
                  Haim Kaplan and
                  Sebastian Wild and
                  Grzegorz Herman},
  title        = {Generalized Graph Packing Problems Parameterized by Treewidth},
  booktitle    = {33rd Annual European Symposium on Algorithms, {ESA} 2025, Warsaw,
                  Poland, September 15-17, 2025},
  series       = {LIPIcs},
  volume       = {351},
  pages        = {3:1--3:15},
  publisher    = {Schloss Dagstuhl - Leibniz-Zentrum f{\"{u}}r Informatik},
  year         = {2025},
  url          = {https://doi.org/10.4230/LIPIcs.ESA.2025.3},
  doi          = {10.4230/LIPICS.ESA.2025.3},
  bibsource    = {dblp computer science bibliography, https://dblp.org}
}

@inproceedings{EsmerFMR24,
  author       = {Baris Can Esmer and
                  Jacob Focke and
                  D{\'{a}}niel Marx and
                  Pawel Rzazewski},
  editor       = {Timothy M. Chan and
                  Johannes Fischer and
                  John Iacono and
                  Grzegorz Herman},
  title        = {List Homomorphisms by Deleting Edges and Vertices: Tight Complexity
                  Bounds for Bounded-Treewidth Graphs},
  booktitle    = {32nd Annual European Symposium on Algorithms, {ESA} 2024, Royal Holloway,
                  London, United Kingdom, September 2-4, 2024},
  series       = {LIPIcs},
  volume       = {308},
  pages        = {39:1--39:20},
  publisher    = {Schloss Dagstuhl - Leibniz-Zentrum f{\"{u}}r Informatik},
  year         = {2024},
  url          = {https://doi.org/10.4230/LIPIcs.ESA.2024.39},
  doi          = {10.4230/LIPICS.ESA.2024.39},
  bibsource    = {dblp computer science bibliography, https://dblp.org}
}

@inproceedings{HegerfeldK22,
  author       = {Falko Hegerfeld and
                  Stefan Kratsch},
  editor       = {Holger Dell and
                  Jesper Nederlof},
  title        = {Towards Exact Structural Thresholds for Parameterized Complexity},
  booktitle    = {17th International Symposium on Parameterized and Exact Computation,
                  {IPEC} 2022, Potsdam, Germany, September 7-9, 2022},
  series       = {LIPIcs},
  volume       = {249},
  pages        = {17:1--17:20},
  publisher    = {Schloss Dagstuhl - Leibniz-Zentrum f{\"{u}}r Informatik},
  year         = {2022},
  url          = {https://doi.org/10.4230/LIPIcs.IPEC.2022.17},
  doi          = {10.4230/LIPICS.IPEC.2022.17},
  bibsource    = {dblp computer science bibliography, https://dblp.org}
}

@inproceedings{EsmerFMR24-source,
  author       = {Baris Can Esmer and
                  Jacob Focke and
                  D{\'{a}}niel Marx and
                  Pawel Rzazewski},
  editor       = {Karl Bringmann and
                  Martin Grohe and
                  Gabriele Puppis and
                  Ola Svensson},
  title        = {Fundamental Problems on Bounded-Treewidth Graphs: The Real Source
                  of Hardness},
  booktitle    = {51st International Colloquium on Automata, Languages, and Programming,
                  {ICALP} 2024, Tallinn, Estonia, July 8-12, 2024},
  series       = {LIPIcs},
  volume       = {297},
  pages        = {34:1--34:17},
  publisher    = {Schloss Dagstuhl - Leibniz-Zentrum f{\"{u}}r Informatik},
  year         = {2024},
  url          = {https://doi.org/10.4230/LIPIcs.ICALP.2024.34},
  doi          = {10.4230/LIPICS.ICALP.2024.34},
  bibsource    = {dblp computer science bibliography, https://dblp.org}
}

@inproceedings{CurticapeanM16,
  author       = {Radu Curticapean and
                  D{\'{a}}niel Marx},
  editor       = {Robert Krauthgamer},
  title        = {Tight conditional lower bounds for counting perfect matchings on graphs
                  of bounded treewidth, cliquewidth, and genus},
  booktitle    = {Proceedings of the Twenty-Seventh Annual {ACM-SIAM} Symposium on Discrete
                  Algorithms, {SODA} 2016, Arlington, VA, USA, January 10-12, 2016},
  pages        = {1650--1669},
  publisher    = {{SIAM}},
  year         = {2016},
  url          = {https://doi.org/10.1137/1.9781611974331.ch113},
  doi          = {10.1137/1.9781611974331.CH113},
  bibsource    = {dblp computer science bibliography, https://dblp.org}
}

@inproceedings{GroenlandMNPR24,
  author       = {Carla Groenland and
                  Isja Mannens and
                  Jesper Nederlof and
                  Marta Piecyk and
                  Pawel Rzazewski},
  editor       = {Karl Bringmann and
                  Martin Grohe and
                  Gabriele Puppis and
                  Ola Svensson},
  title        = {Towards Tight Bounds for the Graph Homomorphism Problem Parameterized
                  by Cutwidth via Asymptotic Matrix Parameters},
  booktitle    = {51st International Colloquium on Automata, Languages, and Programming,
                  {ICALP} 2024, Tallinn, Estonia, July 8-12, 2024},
  series       = {LIPIcs},
  volume       = {297},
  pages        = {77:1--77:21},
  publisher    = {Schloss Dagstuhl - Leibniz-Zentrum f{\"{u}}r Informatik},
  year         = {2024},
  url          = {https://doi.org/10.4230/LIPIcs.ICALP.2024.77},
  doi          = {10.4230/LIPICS.ICALP.2024.77},
  bibsource    = {dblp computer science bibliography, https://dblp.org}
}

@article{GanianHKOS24,
  author       = {Robert Ganian and
                  Thekla Hamm and
                  Viktoriia Korchemna and
                  Karolina Okrasa and
                  Kirill Simonov},
  title        = {The Fine-Grained Complexity of Graph Homomorphism Parameterized by
                  Clique-Width},
  journal      = {{ACM} Trans. Algorithms},
  volume       = {20},
  number       = {3},
  pages        = {19},
  year         = {2024},
  url          = {https://doi.org/10.1145/3652514},
  doi          = {10.1145/3652514},
  bibsource    = {dblp computer science bibliography, https://dblp.org}
}

@inproceedings{PiecykR21,
  author       = {Marta Piecyk and
                  Pawel Rzazewski},
  editor       = {Markus Bl{\"{a}}ser and
                  Benjamin Monmege},
  title        = {Fine-Grained Complexity of the List Homomorphism Problem: Feedback
                  Vertex Set and Cutwidth},
  booktitle    = {38th International Symposium on Theoretical Aspects of Computer Science,
                  {STACS} 2021, Saarbr{\"{u}}cken, Germany (Virtual Conference),
                  March 16-19, 2021},
  series       = {LIPIcs},
  volume       = {187},
  pages        = {56:1--56:17},
  publisher    = {Schloss Dagstuhl - Leibniz-Zentrum f{\"{u}}r Informatik},
  year         = {2021},
  url          = {https://doi.org/10.4230/LIPIcs.STACS.2021.56},
  doi          = {10.4230/LIPICS.STACS.2021.56},
  bibsource    = {dblp computer science bibliography, https://dblp.org}
}

@article{FockeMR24,
  author       = {Jacob Focke and
                  D{\'{a}}niel Marx and
                  Pawel Rzazewski},
  title        = {Counting List Homomorphisms from Graphs of Bounded Treewidth: Tight
                  Complexity Bounds},
  journal      = {{ACM} Trans. Algorithms},
  volume       = {20},
  number       = {2},
  pages        = {11},
  year         = {2024},
  url          = {https://doi.org/10.1145/3640814},
  doi          = {10.1145/3640814},
  bibsource    = {dblp computer science bibliography, https://dblp.org}
}

@article{JaffkeLL24,
  author       = {Lars Jaffke and
                  Paloma T. Lima and
                  Daniel Lokshtanov},
  title        = {b-Coloring Parameterized by Clique-Width},
  journal      = {Theory Comput. Syst.},
  volume       = {68},
  number       = {4},
  pages        = {1049--1081},
  year         = {2024},
  url          = {https://doi.org/10.1007/s00224-023-10132-0},
  doi          = {10.1007/S00224-023-10132-0},
  bibsource    = {dblp computer science bibliography, https://dblp.org}
}

\end{document}